\documentclass[letterpaper,11pt]{article}
\usepackage{amsmath}
\usepackage{physics}
\usepackage{amsthm}
\usepackage{mathrsfs}
\usepackage{fullpage}
\usepackage{color}
\usepackage{bbold}
\newtheorem{theorem}{Theorem}[section]
\newtheorem{corollary}{Corollary}[section]
\newtheorem{lemma}[theorem]{Lemma}
\newtheorem{fact}[theorem]{Fact}

\theoremstyle{definition}
\newtheorem{definition}{Definition}
\theoremstyle{remark}
\newtheorem{remark}[theorem]{Remark}
\newtheorem{example}{Example}
\usepackage{algorithm}
\usepackage{algorithmic}
\usepackage{graphicx}
\usepackage{url}
\usepackage{subcaption}
\definecolor{Blue}{rgb}{0.0,0.2,0.6}
\usepackage{hyperref}
\hypersetup{
    linktocpage=true,
    colorlinks=true,				
    linkcolor=Blue,				
    citecolor=Blue,				
    urlcolor=Blue,			
}

\newcommand{\argmin}{\mathop{\mathrm{arg\,min}}}
\newcommand{\dist}{\mathop{\mathrm{dist}}}
\newcommand{\dt}{\mathop{\mathrm{dist}}\nolimits_2}
\newcommand{\pj}{\mathop{\texttt{P}}}
\newcommand\lablast{\addtocounter{equation}{1}\tag{\theequation}}
\allowdisplaybreaks

\title{Function Privatization in the Local Model}

\author{
Yuting Liang\thanks{Department of Computer Science, University of Toronto. \href{mailto:yliang@cs.toronto.edu}{yliang@cs.toronto.edu}}
\and
Tian Shu\thanks{Computer Science and Engineering, Hong Kong University of Science and Technology. \href{mailto:tshu@connect.ust.hk}{tshu@cse.ust.hk}  
}
\and
Ke Yi\thanks{Computer Science and Engineering, Hong Kong University of Science and Technology. \href{mailto:yike@cse.ust.hk}{yike@cse.ust.hk}}
}
\date{}
\begin{document}
\maketitle
\begin{abstract}
We study the problem of privately releasing functions, with a particular focus on curves, which are images of continuous functions on some finite interval. Many types of data exist naturally as curves, such as trajectory data or $1$D density curves. 
We shall primarily be interested in the local model setting, where the function to be privatized captures data belonging to one individual, which is the more challenging setting with limited prior work. 

Under the standard notion of local differential privacy (DP), any two arbitrarily different functions are required to be made indistinguishable by privatization, which is too strong to allow meaningful utility; we thus work with a generalized notion of DP known as Geo-Privacy (GP), which allows functions far apart to be distinguished more easily while providing strong protection for near functions. To demonstrate the effectiveness of our framework, we provide experimental evaluation on several datasets.
\end{abstract}
\section{Introduction}
A fundamental problem in data privatization is the following: Suppose a user possesses data in the form of a sequence, where each record can be either 1D (e.g., prices over a period or ECG data) or multi-dimensional (e.g., a trajectory).  How can they privatize the data before sending it to an untrusted analyst?  More abstractly, the data can be modeled as a function $q:[0,T] \rightarrow \mathbb{R}^n$, and we thus term the problem \textit{function privatization}.

Naively, one could privatize some subset of the function values $q(t_1),\dots, q(t_k)$ separately, each using some existing privatization method.  However, this solution consumes privacy linear in $k$. 
It is also not a theoretically elegant solution as it only obtains the function values at discrete points (often chosen heuristically) while $f$ can be defined over the reals.  Importantly, this naive solution ignores the  correlations among the function values, which often exist in real-world data.  In the extreme case, suppose we know \textit{a priori} that the function is a constant.  Then privatizing the function at any one point suffices.  On the other hand, if $q$ can really be an arbitrary function, then there might be little to no correlation and the naive solution is already the best one can hope for.  Thus, the central technical problem to be addressed in this paper is: how to design a privatization method that automatically exploits the correlation if it exists in the data, thereby achieving good utility on typical real-world data.

In our setup where each user privatizes their own data, the most popular privacy definition is the local model of differential privacy (DP), which requires the output to be $\varepsilon$-indistinguishable between any two different instances.  However, this is inappropriate for functional data as it insists a uniform privacy guarantee, e.g., the level of indistuinguishability is the same between $q_1(t):=2t+3$ and $q_2(t):=2.01t+2.99$ vs between $q_1(t):=2t+3$ and $q_3(t):=2000t-3000$.  This in turn leads to a privatized output of low utility and there is no effective function privatization method in the local DP model. 
Therefore, we adopt the notion of generalized differential privacy (GP) \cite{chatzikokolakis2013broadening, liang2023concentrated}, which is often adopted in a local model setting for data in a metric space as a more attractive alternative to local DP (e.g. \cite{andres2013geo,liu2021privacy,yang2022k,liang2024smooth}). GP stipulates that the level of indistuinguishability be proportional to the distance (using an appropriate metric) between the two instances (see Section \ref{sec:prelim} for precise definitions).  For functions, we will adopt the natural $L^2$ metric.  Then in the example above, GP will make $q_1$ and $q_2$ much more indistinguishable than $q_1$ and $q_3$.  Such a more refined privacy definition not only results in good utility (as will be achieved by our methods), but also captures the natural privacy requirement in a metric space.

\subsection{Our Contributions}
In this paper, we present a series of function privatization methods, requiring different correlation assumptions on the function: We start by considering linear functions (Section \ref{sec:priv_lin}), and then generalize the method to functions that can be represented as a linear combination of a fixed number $m$ of basis functions (Section \ref{sec:finitedim_func}).   We show that our method incurs an $L^2$ noise to the privatized function that is proportional to $m$ ($\sqrt{m}$ for CGP, a variant of GP).  Since the privacy noise grows with $m$, it is sometimes beneficial to use less basis functions to reduce the noise, at the expense of some approximation error.  In Section \ref{sec:gen_func}, we develop a method that automatically and privately chooses a suitable number of basis functions from a possibly infinite collection, while trying to balance the two sources of error.  In Section \ref{sec:privcurve}, we focus on curves, which are images of continuous functions on some finite interval.  Curves are best approximated using piecewise basis functions.  We develop algorithms to determine how the curve should be divided into pieces, approximate each piece, and then privatize. 

Compared with the naive solution, our methods enjoy the following advantages: (1) It captures the correlation in the data, in the sense that stronger assumptions (e.g., linear functions or functions that can be represented by a few basis functions) result in better utility.  (2) If the assumption does not strictly hold, our method can approximate it while automatically balancing the approximation error and privacy noise, which allows us to extract the hidden correlations present in the data.  (3) Our methods output functions defined over the reals, not only at discrete points.  (4) It supports vector-valued variables as well as vector-valued functions.  Finally, we conducted an extensive set of experiments to demonstrate the practical advantage of our methods using real-world functional data.

\subsection{Related Work}
The closest work to ours is that of Hall et al.  \cite{hall2013differential}, who studied the function privatization problem in the central model of DP.   In the central model, the trusted analyst has $N$ functions $q_i(\cdot), i=1,\dots,N$, where each $q_i$ depends on the private data $x_i$ and wishes to privatize their mean ${1\over N}\sum_i q_i(\cdot)$. They also reduce the privacy noise by exploiting the correlations in each $q_i(\cdot)$.  In particular, they study the case where $q_i(t):= e^{-{(x_i-t)^2}/{2b^2}}$ for fixed $b>0$, which corresponds to the Gaussian kernel density estimation (KDE) function at point $t$, where the $x_i$'s are assumed to have been drawn from some distribution.  

The work of \cite{wagner2023fast} also privatizes the Gaussian kernel function, where a factorized form $g(x)\cdot p(t)$ is used to approximate $\kappa_t(x)=\frac{1}{N}\sum_i e^{-{\|x_i-t\|^2}}$. There, the central curator releases a privatized $\tilde{g}(x)$, which can then be used by an analyst to perform further queries without further interaction with the curator. The work in \cite{alda2017bernstein} uses Bernstein polynomials to approximate a real-valued function $q_x$ parameterized by private data $x=(x_1,\dotsb,x_N)$; their method builds a lattice of the function domain and privatizes the evaluations of $q_x(\cdot)$ on all lattice points by adding Laplace noise, which requires the function $q_x$ to have small sensitivity w.r.t. changing $x$.

The works that privatize a set of queries by approximating the query functions are also related \cite{thaler2012faster, wang2016differentially}.  In \cite{thaler2012faster}, they showed that for private data $x\in (\{0,1\}^d)^N$, it is possible to approximate a set of $k$-way marginal queries by low-degree polynomials, which leads to efficient running time and requires lower sample complexity than previous works on $k$-way marginals. 
The work of \cite{wang2016differentially} considers privatizing a set of queries of the form $x\mapsto \frac{1}{N}\sum_{i}f(x_i)$, for private data $x$ with $x_i\in [-1,1]^d$ and smooth functions $f$ that can be approximated by trigonometric polynomials, assuming bounded derivatives of some order; there, the private data vector $x\in ([-1,1]^d)^N$ is embedded into a basis and privatized, so that the query answer $\frac{1}{N}\sum_{i}f(x_i)$ can be approximated by a dot product of the coefficients of $f$ and the privatized basis.

The problem of privately releasing trajectory data has been considered in central DP \cite{chen2012differentially,he2015dpt,gursoy2018utility}, where the goal is to generate a synthetic collection of trajectories that is distributionally similar to the original input collection of trajectories. In these works, a privatized empirical distribution on sequences of locations is constructed, from which the synthetic database can then be generated.
For a comprehensive list of works on trajectory data privatization under central DP, we refer the reader to the SoK paper of \cite{miranda2023sok}, which provides a detailed discussion of the works in terms of utility and limitations. This problem has also been attempted under the local model of DP, where the goal is to publish a privatized trajectory from a single input trajectory \cite{cunningham2021real}; however, due to the strong requirement of local-DP, it is not possible for a local-DP mechanism on this problem to have any meaningful utility, unless the data universe is restricted to a very small space.

There are a couple of works using GP for privatization of various types of queries \cite{liang2023concentrated, liang2024smooth}. In particular, \cite{liang2024smooth} also considered privatizing the Gaussian kernel function; however, their mechanism is for privatizing a single query and loses utility quickly when many queries are asked. \cite{liang2023concentrated} considered releasing a collection of points in a trajectory under GP using basic composition; however, such an approach is not suitable when the goal is to preserve a general likeness of the curve, as we will also demonstrate in our experiments.

\section{Preliminaries}
\label{sec:prelim}

\subsection{Geo-Privacy}
\label{sec:gp_prelim}
We briefly recall the definitions and basic properties of GP, which is defined on metric spaces. Let $(U,\dist)$ be a metric space. For $\Lambda > 0$, and $x,x'\in U$, we write $x\sim_{\Lambda}{x'}$ if $\dist(x,x')\le \Lambda$.
\begin{definition} [Geo-privacy, {Definition $6$ in \cite{liang2023concentrated}}]
\label{def:gp}
    Fix $\varepsilon, \delta \ge 0$, $\Lambda \in \mathbb{R}_{>0}\cup \{\infty\}$. A randomized mechanism $M:U\rightarrow V$ satisfies  $(\varepsilon,\delta,\Lambda)$ geo-privacy, or simply $(\varepsilon,\delta,\Lambda)$-GP, if for all measurable $S\subseteq V$ and all $x\sim_{\Lambda}x'$
    \[
    \Pr[M(x)\in S] \le e^{\varepsilon\dist(x,x')} \Pr[M(x')\in S]+\delta.
    \]
\end{definition}
The parameter $\varepsilon$ measures privacy loss per unit distance (e.g. $\varepsilon=0.001$ per meter using the Euclidean metric). 
We denote by $\mathcal{M}(x)$ the distribution of the random variable $M(x)$, and by $m(x)(y)$ its pdf at $y\in \mathcal{M}(x)$. Concentrated geo-privacy (CGP) is defined in terms of the R\'{e}nyi divergence of order $\alpha$ \cite{renyi1961measures, van2014renyi}, defined for distributions $\mathcal{P},\mathcal{Q}$ with pdf's $p(\cdot), q(\cdot)$, respectively:
\[D_{\alpha}(\mathcal{P}\|\mathcal{Q}) :=\frac{1}{\alpha-1}\log\left(\int_V p(y)^{\alpha}q(y)^{1-\alpha} dz\right).\]
\begin{definition} [Concentrated Geo-privacy , {Definition $3$ in \cite{liang2023concentrated}}]
\label{def:cgp}
    Fix $\rho \ge 0$. $\Lambda \in \mathbb{R}_{>0}\cup \{\infty\}$. A mechanism $M:U\rightarrow V$ satisfies $(\rho,\Lambda)$ concentrated geo-privacy, or simply $(\rho,\Lambda)$-CGP, if for all $x\sim_{\Lambda}x'$, and all $\alpha > 1$
    \[
    D_{\alpha}(\mathcal{M}(x)\|\mathcal{M}(x'))\le {\alpha}\rho \dist(x,x')^2.
    \]
\end{definition}
The parameter $\rho$ is privacy loss per unit distance squared. 
Note that by setting $\Lambda=1$ in the definitions above and using the Hamming metric as $\dist(\cdot,\cdot)$, one recovers the usual $(\varepsilon,\delta)$-DP \cite{dwork2006calibrating} and $\rho$-CDP \cite{bun2016concentrated} definitions of the central model; while setting $\dist(\cdot,\cdot)$ to be the discrete metric $\dist_{01}(x,x'):=\mathbb{1}\{x\neq x'\}$ yields the same definitions in the local model. The main appeal of the GP model is in its flexibility to choose a metric suitable for the application and data being studied.

We write $\varepsilon$-GP for $(\varepsilon,0,\infty)$-GP and $\rho$-CGP for $(\rho,\infty)$-{C}GP. Note that one way to show that a mechanism $M$ satisfies $\varepsilon$-GP, is to show that $M$ induces distributions on all pair of inputs $x,x'\in U$ such that $\mathcal{M}(x)$ and $\mathcal{M}(x')$ have the same support, and the ratio of their pdf's satisfies for all $y \in \mathcal{M}(x)$: $
\frac{m(x)(y)}{m(x')(y)} \le e^{\varepsilon \dist(x,x')}.$

An example of such mechanisms is the \textit{exponential mechanism}.
\begin{lemma} [Exponential mechanism for $\varepsilon$-GP \cite{chatzikokolakis2013broadening, mcsherry2007mechanism}]
\label{lm:lap_mech_gen}
    Let $(V,\dist_V)$ be a metric space. Let $g:U\rightarrow V$ be $K$-Lipschizt where $K>0$. The mechanism that, on input $x$, draws a $y$ from a distribution with pdf $m(x)(y)=c(\varepsilon,K) e^{-\frac{\varepsilon}{K}\dist_V(g(x),y)}$, where $c(\varepsilon,K)$ is a constant independent of the input $x$, is $\varepsilon$-GP.
\end{lemma}

We also have the following basic mechanisms for vector-valued queries.
\begin{lemma} [Canonical mechanisms for vector-valued queries \cite{liang2023concentrated}]
\label{lm:canon_mech}
    Let $g:U\rightarrow V\subseteq \mathbb{R}^m$ be $K$-Lipschitz. Let $M:U\rightarrow\mathbb{R}^m$ be the mechanism defined by $M(x)=g(x)+\frac{K}{b}Z$, where $Z$ is an $m$-dimensional random vector. Then
    \begin{enumerate}
        \item $M$ is $\varepsilon$-GP if $b={\varepsilon}$, where $Z$ is drawn from the standard $m$-dimensional spherical Laplace distribution, denoted $Z\sim \mathrm{SLap}(m)$, with pdf $h(z)\propto e^{-\|z\|}$;
        
        \item $M$ is $\rho$-CGP if $b=\sqrt{2\rho}$, where $Z$ is drawn from the standard $m$-dimensional Gaussian distribution, denoted 
        $Z\sim\mathcal{N}(0,I_{{m\times m}})$, with pdf $h(z)\propto e^{-\frac{1}{2}{\|z\|^2}}$.
    \end{enumerate}
\end{lemma}

Similar to DP, GP also enjoys the following composition properties, which will be helpful for developing algorithms composed of multiple steps.
\begin{lemma} [Basic compositions \cite{andres2013geo,liang2023concentrated}]
    \label{lm:basic_comp}
    Fix $k\ge 2$, let $M_j:U\rightarrow V_j$ for $j\in [k]$. Let $M=(M_1,\dotsb,M_k)$ be the $k$-fold composition, where the choice of $M_j$ may depend on the outputs of previous $M_l$'s for $l<j$. Then
    \begin{enumerate}
        \item $M$ is $\varepsilon$-GP, if each $M_j$ is $\varepsilon_j$-GP, where $\varepsilon:=\sum_{j\in[k]}\varepsilon_j$;
        \item $M$ is $\rho$-CGP, if each $M_j$ is $\rho_j$-CGP, where $\rho:=\sum_{j\in[k]}\rho_j$.
    \end{enumerate}
\end{lemma}

A special case of the composition property where all but the first mechanism have $\varepsilon=0$ is known as the \textit{postprocessing} property, i.e., $g_{\text{pp}}(M(\cdot))$ satisfies the same GP guarantee of $M$ for any measurable function $g_{\text{pp}}(\cdot)$. 

The Sparse Vector Technique (SVT) \cite{dwork2014algorithmic}, shown in Algorithm \ref{alg:SVT}, takes in a possibly infinite sequence of queries, and returns the index corresponding (approximately) to the query that first crosses some prespecified threshold $\mathcal{T}$. The SVT is $\varepsilon$-GP \cite{liang2023concentrated} for any metric $\dist(\cdot,\cdot)$, as long as all queries are $K$-Lipschitz w.r.t. the same metric.
\begin{algorithm}
\caption{Sparse Vector Technique}
    \label{alg:SVT}
    \vspace{-8pt}
    \begin{flushleft}
    \textbf{Input}: $x\in U$; $(\varepsilon_1,\varepsilon_2) \text{\;with\;}\varepsilon=\varepsilon_1+\varepsilon_2$; $\mathcal{T}$; $K$; $g_1, g_2, \dotsb$ each $K$-Lipschitz\\
    \textbf{Output}: variable-length sequence $y_1, y_2, \dotsb$
    \end{flushleft}
     \vspace{-5pt}
    \begin{algorithmic}[1]
    \STATE draw $W\sim\mathrm{Lap}(K/\varepsilon_1)$
    \FOR{$j = 1,\dotsb$}
    \STATE draw $V_j \sim\mathrm{Lap}(2K/\varepsilon_2)$
    \IF{$g_j(x)+V_j \geq \mathcal{T}+W$}
    \STATE output $y_j = \top$ and \textbf{HALT}
    \ELSE
    \STATE output $y_j = \bot$
    \ENDIF
    \ENDFOR
    \end{algorithmic}
\end{algorithm}

The SVT has the following utility guarantee, which can be verified by standard analysis (e.g. see \cite{dong2023universal}):
\begin{lemma}
\label{lm:svt_util}
    Let $\bar{k}$ be the index $j$ when $\mathrm{SVT}(x,(\frac{\varepsilon}{3},\frac{2\varepsilon}{3}),\mathcal{T},K,g_1,g_2,\dotsb)$ halts. Suppose there is $t$ such that $g_t(x)\ge \mathcal{T}+6K\ln(2/\beta)/\varepsilon$. Then with probability at least $1-\beta$, $\bar{k}$ satisfies $\bar{k}\le t$ and $g_{\bar{k}}(x) \ge \mathcal{T}
    -\frac{3K}{\varepsilon}\left(\ln\left(\frac{2t}{\beta}\right)+\ln\left(\frac{2}{\beta}\right)\right)$. 
\end{lemma}

\subsection{The Problem of Function Privatization}
We consider functions of the form $q:I\rightarrow \mathbb{R}^n$, where $I=I_1\times\dotsb \times I_d$ with $I_i\subseteq \mathbb{R}$ for $i\in[d]:=\{1,\dotsb,d\}$. A natural metric between two functions is the $L^2$ distance $\dist_2(q,q'):=\left(\int_I \|q(t)-q'(t)\|^2dt\right)^{1/2}$, where $\|\cdot\|$ denotes the usual Euclidean metric.
Our goal is to compute, for a given function $q$, a privatized function $\tilde{q}$ that satisfies GP (or CGP) w.r.t. the $\dist_2$ metric. We assume $q\in L^2(I)$, where $L^2(I)$ denotes the set of square-integrable functions on $I$, i.e., $ \int_I \|q(t)\|^2 dt < \infty$;
so $\dt(q,q')<\infty$ for every pair $q,q'$.

We focus primarily on the setting where $d=1$ and $I=[0,T]$ for $T>0$.  We distinguish between the restricted case $I=[0,T]$ and the general case $I$ by explicitly writing $[0,T]$ in the former case. In particular, we are interested in continuous functions on $[0,T]$, which correspond to \textit{curves} in $\mathbb{R}^n$. We use $C(I)$ to denote the set of continuous functions on $I$.  

\subsection{On the $L^2$ Metric}
\label{sec:l2_metric}
The $L^2$ metric is arguably the most
natural/important metric for physical sciences and engineering; it is preferred by
physicists and engineers because of its desirable properties such as completeness, rotation-invariance \cite{stein2011fourier, arfken2011mathematical}, relations to physical
quantities such as the energy of the signal \cite{oppenheim1997signals}, etc. 

Conceptually, the $L^2$ metric is to functions as the $\ell_2$ (Euclidean) metric is to points, and the latter is commonly used for point privatization \cite{andres2013geo,liang2023concentrated, liang2024smooth}. Concretely, consider a function $q:[0,1]\rightarrow\mathbb{R}$: The function $q$ can
be identified with its image $f([0,1])$, which is a collection of infinitely many points.
Recall we can approximate the $L^2$ distance between two curves using point evaluations at
$t_1,\dots,t_n$ for increasingly many points; in particular, we have
\[\int_0^1(q(t)-q'(t))^2\;{d}t = \lim_{n\rightarrow \infty}
\frac{1}{n}\sum_{j=1}^n (q(t_j)-q'(t_j))^2.\]
Thus, the $L^2$ distance between two functions is essentially an average of $\ell_2$
distances between the collections of points $(q(t_1),\dots,q(t_n))$ and
$(q'(t_1),\dots,q'(t_n))$ for infinitely many points.
\begin{figure}[h]
    \centering
        \includegraphics[width=0.5\textwidth]{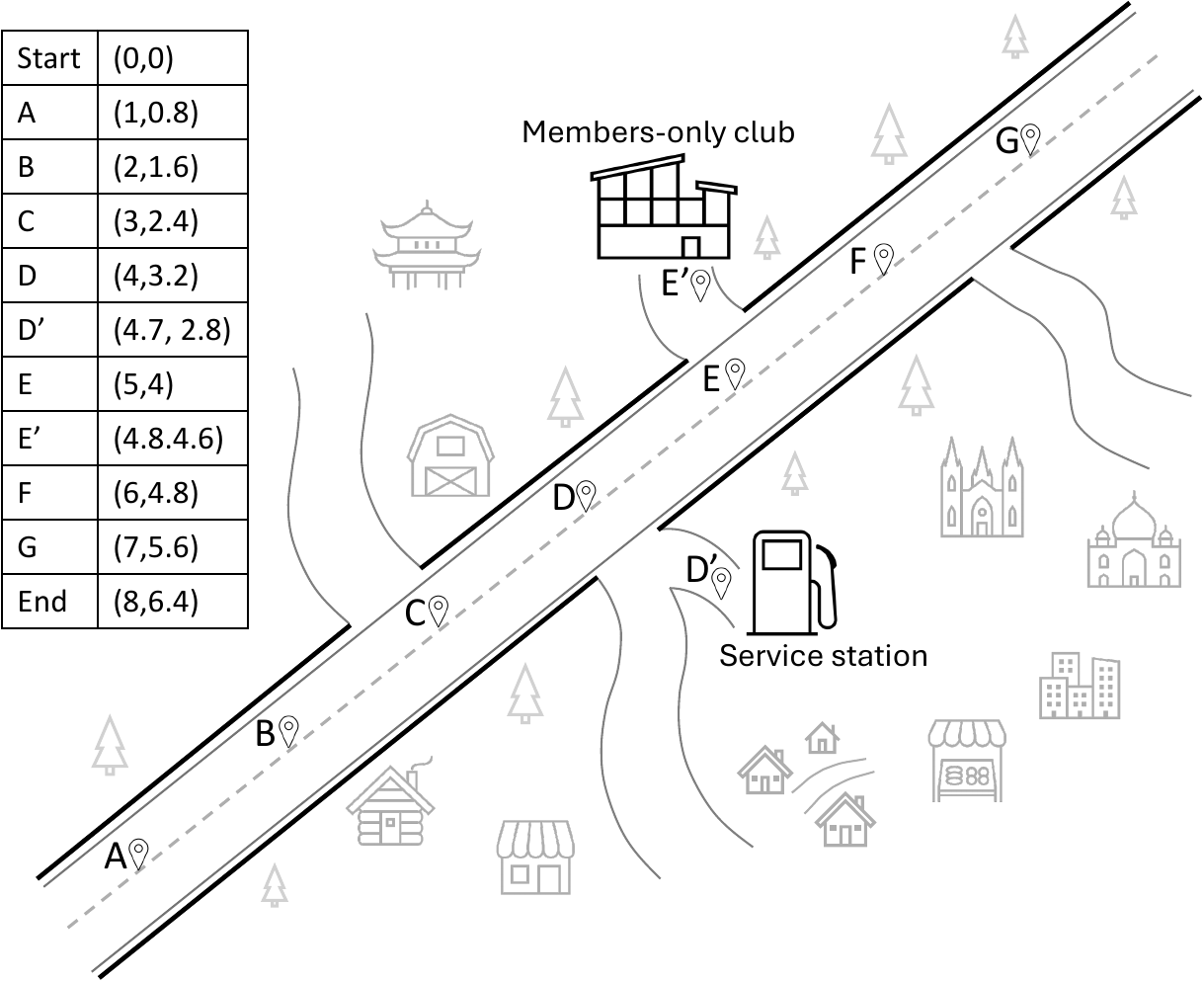}
        \caption{{Reference locations for example trajectories.}}
        \label{fig:ref_loc}
\end{figure}
\paragraph{Privacy protection via $L^2$ metric.}
The $L^2$ metric measures variations across the entire domain, it is robust against variations in small areas. An analyst observing a privatized function would not be able to infer with high confidence, whether and where the original function had a brief but otherwise distinguishing shift.
For example, when applied to trajectory data, it protects against sudden and short changes in the curve. 
Consider the following scenario: We have three trajectories, $q_1,q_2,q_3$ traveling on a straight road from time $0$ to time $1$, see Fig.~\ref{fig:ref_loc} for an illustration of the road and values of the locations. The trajectories are listed as (time, location) tuples, where the traveling speed between two tuples is assumed to be uniform; they are similar except $q_2, q_3$ will briefly stop at different locations before continuing down the road:
\begin{align*}
      q_1&: \{(0.13,A),(0.26,B),(0.4,C),(0.53,D),(0.66,E),(0.8,F),(0.9,G)\}\\
      q_2&: \{(0.12,A),(0.24,B),(0.35,C),(0.45,D),(0.52,D'),(0.54,D'),(0.6,D),(0.7,E),(0.8,F),(0.9,G)\}\\
      &\text{\;\;\;(briefly stops at $D'$)}\\
      q_3&: \{(0.12,A),(0.24,B),(0.35,C),(0.45,D),(0.55,E),(0.61,E'),(0.63,E'),(0.7,E),(0.8,F),(0.9,G)\} \\
      &\text{\;\;\;(briefly stops at $E'$)}
\end{align*}
The three trajectories have small pairwise $L^2$ distances; thus, an analyst observing a privatized trajectory would not be able to tell whether the original trajectory had been $q_1$, $q_2$ or $q_3$ with high confidence. Also, the pairwise $L^2$ distance is smaller than the pairwise maximum $\ell_2$ distance at any point, for each pair:
\begin{align*}
\dt(q_1,q_2) &= 0.40\;<\;\max\nolimits_{t\in [0,1]}\|q_1(t)-q_2(t)\|=0.85\\
\dt(q_1,q_3) &=0.48\;<\;\max\nolimits_{t\in [0,1]} \|q_1(t)-q_3(t)\|=1.08\\
\dt(q_2,q_3) &=0.50\;<\;\max\nolimits_{t\in [0,1]}\|q_2(t)-q_3(t)\|=1.54.
\end{align*}
Thus, the distinguishability level is smaller under the $L^2$ (i.e. the privacy protection is stronger for the same privacy parameter), compared to the maximum point-wise distance.

Finally, the $L^2$ metric provides protection for single points in a way that is consistent with (and at least as strong as) that provided by the Euclidean metric for single points. For trajectory data, the longer a curve stays at a single point, the more similar the privacy guarantee becomes to that offered by point privatization under the Euclidean metric used in previous works \cite{andres2013geo,liang2023concentrated,liang2024smooth}. 
Concretely, consider two constant trajectories $q, q':[0,1]\rightarrow \mathbb{R}^d$  (e.g. if a person stays at a location for the entire duration of the trajectory) such that $q(t)=p$ and $ q'(t)=p'$ for some $p,p'\in \mathbb{R}^d$ for all $t$.  
The $L^2$ distance between them is
\[
\left(\int_0^1\|q(t)-q'(t)\|^2\;{d}t\right)^{1/2} = \left(\|p-p'\|^2\right)^{1/2} = \|p-p'\|,
\]
which is the same level of distinguishability required for point privatization in the aforementioned works. 

\section{Privatization of functions}
\label{sec:priv_func}
\subsection{Linear functions}
\label{sec:priv_lin}
We start by considering the simple case of linear functions, which serves as an illustrative example to demonstrate some key ideas of our framework. Let $f_a:[0,T]\rightarrow \mathbb{R}$, defined by $f_a(t)=a_1t+a_2$ with $a=(a_1,a_2)\in \mathbb{R}^2$. Let $U_f$ be the space consisting of all such linear functions. Restating the GP definition for this special case, a mechanism $M_f$ on $(U_f,\dist_2)$ satisfies $\varepsilon$-GP, if for all pairs $f_a$ and $f_{a'}$, and all measurable subsets $S_f\subseteq V= U_f$,
\[
\Pr[M_f(f_{a})\in S_f] \le e^{\varepsilon \dist\nolimits_2\left(f_{a},f_{a'}\right)}\Pr[M_f(f_{a'})\in S_f]. 
\]

Note that we can identify $f_{a}$ with the point $a\in \mathbb{R}^2$, so to privatize $f_a$ it would suffice to privatize $a$, 
i.e., we want to find a mechanism $M:\mathbb{R}^2\rightarrow \mathbb{R}^2$ such that for all $S_f\subseteq U_f$, and all pairs $a, a'$
\begin{equation}
\Pr[M(a)\in S] \le e^{\varepsilon \dist\nolimits_2\left(f_a,f_{a'}\right)}\Pr[M(a'))\in S]
\label{eqn:gp_req_func2R2}
\end{equation}
where $S:=\{c\in\mathbb{R}^2 : f_{c}\in S_f\}$. Observing also the one-to-one correspondence between $S$ and $S_f$, we may simply write that \eqref{eqn:gp_req_func2R2} is required to hold for all measurable $S\subseteq \mathbb{R}^2$.

We first show that a natural attempt fails to solve the problem.  As stated above, for linear functions the problem reduces to privatizing a point $a\in \mathbb{R}^2$.  Thus, it is tempting to simply apply {Lemma~\ref{lm:canon_mech}} on $a$ while taking the identity function as $g$ (thus $K=1$).  For concreteness, consider the GP version of {Lemma~\ref{lm:canon_mech}} with $\varepsilon=1$, which adds $Z$ to $a$, where $Z$ follows the $2$-dimensional spherical Laplace distribution.  This outputs a $y=(y_1, y_2)$ with pdf $\propto e^{-\|y-a\|}$. Now consider the following two functions on $[0,1]$: $f_a:t\mapsto 2t+0.5$ and $f_{a'}:t\mapsto 0.5t+2$; i.e., $a=(2,0.5)$ and $a'=(0.5,2)$. Then the ratio of the respective pdf's at $y=(3,1.5)$ is 
\begin{align*}
\frac{e^{-\sqrt{(y_1-a_1)^2+(y_2-a_2)^2}}}{e^{-\sqrt{(y_1-a'_1)^2+(y_2-a'_2)^2}}}&=e^{-\sqrt{2}+\sqrt{2.5^2+0.5^2}} = e^{1.135} \\
&>e^{1.5/\sqrt{3}}=e^{\sqrt{\int_{0}^1 ((a_1-a'_1)t+(a_2-a'_2))^2 dt}}=e^{\dt(f_a,f_{a'})},
\end{align*}
which violates the GP requirement.  Critically, this naive method fails because the $L^2$ distance between $f_a$ and $f_{a'}$ is \textit{not} equal to the $\ell_2$ distance between $a$ and $a'$: two points far apart in the Euclidean space, which this method does not protect, may correspond to two functions that are very close in the function space, which are required to be indistinguishable under GP.  

To fix the issue, we should apply Lemma~\ref{lm:lap_mech_gen} with the correct metric, i.e., draw $y$ from a distribution with pdf $m({a})(y) \propto e^{-\varepsilon \dist_2\left(f_{y},f_{a}\right)}$.  This is no longer an isotropic spherical Laplace distribution. To draw from such a distribution efficiently, we rewrite
\begin{align}
\nonumber
\dist\nolimits_2\left(f_{y},f_{a}\right)^2 &= \int_{0}^T \left(f_{(y_1,y_2)}(t)- f_{(a_1,a_2)}(t)\right)^2 dt \\
\nonumber
&= \int_0^{T} \left((y_1-a_1)t+(y_2-a_2)\right)^2 dt\\
\nonumber
&= \int_0^{T} \left((y_1-a_1)^2 t^2+2(y_1-a_1)(y_2-a_2)t + (y_2-a_2)^2\right) dt\\
\nonumber
&= \left[\frac{1}{3}(y_1-a_1)^2 t^3 + (y_1-a_1)(y_2-a_2)t^2 + (y_2-a_2)^2 t\right]_{t=0}^T \\
\nonumber
&= \frac{1}{3}(y_1-a_1)^2 T^3 + (y_1-a_1)(y_2-a_2)T^2 + (y_2-a_2)^2 T\\
\nonumber
&= \begin{bmatrix}
    (y_1-a) & (y_2-a_2)
\end{bmatrix}
\underbrace{\begin{bmatrix}
    \frac{1}{3}T^3 & \frac{1}{2}T^2 \\
    \frac{1}{2}T^2 & T
\end{bmatrix}}_{\Sigma^{-1}}
\begin{bmatrix}
    (y_1-a_1)\\
    (y_2-a_2)
\end{bmatrix}\\
\label{eqn:funcdist2_to_l2}
 &= \left\| \Sigma^{-1/2} \begin{bmatrix}
    (y_1-a_1)\\
    (y_2-a_2)
\end{bmatrix}\right\|^2.
\end{align}
The matrix $\Sigma^{-1}$ can be shown to be positive definite, hence the distribution with pdf 
\begin{equation}
\label{eqn:linear_mech}
m(a)(y) \propto e^{-\varepsilon \dist_2\left(f_{y},f_{a}\right)} = e^{-\varepsilon \|\Sigma^{-1/2}[(y_1-a_1),\; (y_2-a_2)]^{\mathsf{T}}\|}
\end{equation}
precisely corresponds to outputting $\frac{1}{\varepsilon}\Sigma^{1/2} Z + a$ where $Z$ follows the standard 2D spherical Laplace distribution.  Equivalently speaking, we must add noise to $a$ in the eigen-directions of $\Sigma$.

It is important to note that the scaling matrix $\Sigma$ depends on $T$; in particular, it requires $T<\infty$.  This is not a restriction, though, since linear functions over $[0, \infty)$ is not square-integrable.  In Sections~\ref{sec:gen_func} and ~\ref{sec:curvepriv_unbounded}, we will see how square-integrable functions over an unbounded domain can still be privatized.

\begin{remark}
    The mechanism given by the pdf in ~\eqref{eqn:linear_mech} above can be extended to piecewise linear functions. Following the derivations above, one finds that the matrix $\Sigma^{-1}$ becomes a block matrix of size $2k\times 2k$, if the function is continuous within each of the $k\ge 1$ segments. We will discuss piecewise functions in more detail in Section~\ref{sec:privcurve} when we discuss the privatization of curves.
\end{remark}
\subsection{Multi-basis Functions}
\label{sec:finitedim_func}
Linear functions can be considered as linear combinations of two basis functions $\phi_1(t)=t, \phi_2(t)=1$, and our method in the previous subsection can be easily extended to functions with $m$ basis functions.  Such functions are thus parameterized by points in $\mathbb{R}^m$, i.e., an $m$-dimensional function space.   
Specifically, let $U_f=\{f_a : a\in \mathbb{R}^m\}$ where each $f_a:I\rightarrow \mathbb{R}$ is defined by $f_a(t) = \sum_{j=1}^m a_j\phi_j(t)$, for basis functions $\phi_1(\cdot),\dotsb,\phi_m(\cdot)\in L^2(I)$, with each $\phi_j:I\rightarrow \mathbb{R}$ for $j\in [m]$.  We assume the $\phi_j(\cdot)$'s are linearly independent to avoid redundancy.  A common set of basis functions are $\{\phi_j(t):=t^{m-j}\}_{j\in [m]}$, which induce the space of all  polynomials of degree up to $m-1$.  Note that linear independence among the $\phi_j(\cdot)$'s implies that $f_a(\cdot) = f_{a'}(\cdot) \iff a=a'$, thus preserving the one-to-one correspondence between the sets $S_f\subseteq U_f$ and $S:=\{c\in \mathbb{R}^m: f_c\in S_f\}$.

Next, we extend our $\varepsilon$-GP mechanism for linear functions to this more general function space. 
For $y=(y_1,\dotsb,y_m)\in \mathbb{R}^m$, we have
\begin{align}
	\nonumber
	\dist\nolimits_2(f_y, f_a)^2 &= \int_I (f_y(t)-f_a(t))^2 dt \\
	\nonumber
	&= \int_I\left((y_1-a_1)\phi_1(t)+\dotsb+(y_m-a_m)\phi_m(t)\right)^2 dt\\
	\nonumber
	&=\int_I \left(\sum_{j=1}^m \sum_{k=1}^m (y_j-a_j)\phi_j(t)\cdot (y_k-a_k)\phi_l(t)\right) dt\\
	\nonumber
	&=\sum_{j=1}^m\sum_{k=1}^m (y_j-a_j)(y_k-a_k)\cdot{\int_I\phi_j(t)\phi_l(t) dt}\\
	\nonumber
	&=\begin{bmatrix}
		(y_1-a_1) & \dotsb & (y_m-a_m)
	\end{bmatrix} \underbrace{\begin{bmatrix}
		\int_I\phi_1(t)\phi_1(t) dt & \dotsb & \int_I\phi_1(t)\phi_m(t) dt\\
		\nonumber
		\vdots & \ddots &\vdots\\
		\int_I\phi_m(t)\phi_1(t) dt & \dotsb &\int_I\phi_m(t)\phi_m(t) dt
	\end{bmatrix}}_{\Sigma^{-1}}
	\begin{bmatrix}
		y_1-a_1\\
		\vdots\\
		{y_m-a_m}
	\end{bmatrix}\\
	\label{eqn:func_dist2l2_Rd}
	&=\left\|\Sigma^{-1/2}\begin{bmatrix}
		y_1-a_1\\
		\vdots\\
		y_m-a_m
	\end{bmatrix}\right\|^2.
\end{align}
Thus, we can privatize $f_a\in U_f$ by privatizing $a$ as $M(a):=a+\frac{1}{\varepsilon}\Sigma^{1/2}Z$ where $Z$ is drawn from a $m$-dimensional spherical Laplace distribution, and $[\Sigma^{-1}]_{j,l}=\int_I\phi_j(t)\phi_l(t) dt$. Note that $\Sigma$ is positive definite, since $\Sigma^{-1}$ has full rank and is positive semi-definite as it is an outer product of a vector of linearly independent functions. {In particular, if $\{\phi_j(\cdot)\}_{j\in [m]}$ is a set of orthogonal basis such that $\int_I \phi_j(t) \phi_l(t) dt = 0$ for $j\neq l$, then all of the off-diagonal entries in $\Sigma$ vanish.} 
We summarize this result below.
\begin{theorem}
\label{thm:gp_funcRm}
    Fix a set of linearly independent functions $\{\phi_j\}_{j\in [m]}$ $\subset L^2(I)$, where each $\phi_j:I\rightarrow \mathbb{R}$. Let $U_f:=\{f_a(\cdot) := \sum_j a_j\phi_j(\cdot)\; |\; a\in \mathbb{R}^m\}$, equipped with the $L^2$ metric $\dist_2(\cdot,\cdot)$. Let $\Sigma^{-1}$ be the matrix whose $(j,l)$th entry is given by $[\Sigma^{-1}]_{j,l} = \int_I \phi_j(t)\phi_l(t) dt$ for $j,l\in [m]$. 
    Then for any $f_a\in U_f$, the mechanism that releases $f_{\tilde{a}}(\cdot)$ where $\tilde{a}=M(a):=a+\frac{1}{\varepsilon}\Sigma^{1/2}Z$, and $Z\sim \mathrm{SLap}(m)$,
    is $\varepsilon$-GP.
\end{theorem}
The identity in \eqref{eqn:func_dist2l2_Rd} allows us to also privatize any $f_a\in U_f$ under CGP. In fact, the proof to the theorem below gives an alternative way to derive the GP guarantee of the mechanism mentioned above using Lemma~\ref{lm:canon_mech}.
\begin{theorem}
\label{thm:cgp_funcRm}
    Let $(U_f,\dist_2)$ and $\Sigma$ be as defined in Theorem~\ref{thm:gp_funcRm}. Then for any $f_a\in U_f$, the mechanism that releases $f_{\tilde{a}}(\cdot)$ where $\tilde{a}=M(a):=a+\frac{1}{\sqrt{2\rho}}\Sigma^{1/2}Z$, and $Z\sim\mathcal{N}(0,I_{m\times m})$, satisfies $\rho$-CGP.
\end{theorem}
\begin{proof}
    Let $g:U_f\rightarrow \mathbb{R}^m$ be given by $g(f_a):=\Sigma^{-1/2} a$. Then
 \begin{align*}
 \|g(f_a)-g(f_{a'})\| &= \left\|\Sigma^{-1/2}\begin{bmatrix}
     a_1\\
     \vdots\\
     a_m
 \end{bmatrix}-\Sigma^{-1/2}\begin{bmatrix}
     a'_1\\
     \vdots\\
     a'_m
 \end{bmatrix}\right\| 
 = \left\|\Sigma^{-1/2}\begin{bmatrix}
     a_1-a'_1\\
     \vdots\\
     a_m-a'_m
 \end{bmatrix}\right\| = \dist\nolimits_2(f_{a}, f_{a'}),
 \end{align*}
where the last equality is due to identity~\eqref{eqn:func_dist2l2_Rd}. Thus, $g$ is $1$-Lipschitz w.r.t. $\dist_2$ and by Lemma~\ref{lm:canon_mech}, the mechanism $g(f_a)\mapsto g(f_a)+\frac{1}{\sqrt{2\rho}}Z$ for $Z\sim \mathcal{N}(0,I_{m\times m})$ satisfies $\rho$-CGP. 
Moreover, by post-processing, we get that the mechanism $M$ given by 
\[M(a):=\Sigma^{1/2}\left(g(f_a)+\frac{1}{\sqrt{2\rho}}Z\right) =\Sigma^{1/2}\left(\Sigma^{-1/2}a+\frac{1}{\sqrt{2\rho}}Z\right) = a+\frac{1}{\sqrt{2\rho}}\Sigma^{1/2}Z\]
satisfies $\rho$-CGP.
\end{proof}

Below, we analyze the utility of our mechanism, i.e., the $L^2$ error induced by the privacy noise, both in expectation and with high probability. 

\begin{lemma}
    \label{lm:util_funcRm_slap}
    Let $\tilde{a}=M(a):=a+\frac{1}{\varepsilon}\Sigma^{1/2}Z$, where $Z\sim \mathrm{SLap}(m)$. Then $M$ is $\varepsilon$-GP, and
    \begin{enumerate}
        \item $\mathbb{E}[\dist\nolimits_2(f_{\tilde{a}},f_a)]=\frac{m}{\varepsilon}$; 
        \item $\dist\nolimits_2(f_{\tilde{a}},f_a) \le \frac{\eta(\beta;m,1)}{\varepsilon}$ with probability at least $1-\beta$,
        \end{enumerate}
        where $\eta(\beta;m,1):=\frac{e}{e-1}(m+\ln\frac{1}{\beta})$.
\end{lemma}

\begin{proof}
We have
         \[
\dist\nolimits_2(f_{\tilde{a}},f_a) =\left\|\Sigma^{-1/2}\begin{bmatrix}
     \tilde{a}-a\\
 \end{bmatrix}\right\| = \left\|\Sigma^{-1/2}\begin{bmatrix}
     a+\frac{1}{\varepsilon}\Sigma^{1/2}Z-a\\
 \end{bmatrix}\right\| = \frac{1}{\varepsilon}\left\|Z\right\|.\]
The magnitude $\|Z\|$ follows a generalized gamma distribution (see Lemma~\ref{lm:gen_gamma}), where $\|Z\|\sim \mathcal{G}(1,m,1)$ for $Z$ drawn from a spherical Laplace distribution. The rest follows from Corollary~\ref{cor:expZ_slap_gauss} and Lemma~\ref{lm:slap_highprob}.
\end{proof}

Similar arguments apply for the $\rho$-CGP variant.
Notice that the magnitude of an $m$-dimensional standard Gaussian random vector has distribution $\mathcal{G}(\sqrt{2},m,2)$, then applying Corollary~\ref{cor:expZ_slap_gauss} and Lemma~\ref{lm:gauss_mag_laurent}, we get:

\begin{lemma}
    \label{lm:util_funcRm_gauss}
    Let $\tilde{a}=M(a):=a+\frac{1}{\sqrt{2\rho}}\Sigma^{1/2}Z$, where $Z\sim \mathcal{N}(0,I_{m\times m})$. Then $M$ is $\rho$-CGP, and
    \begin{enumerate}
        \item $\mathbb{E}[\dist\nolimits_2(f_{\tilde{a}},f_a)]=\sqrt{2}\frac{\Gamma(m/2+1/2)}{\Gamma(m/2)\sqrt{2\rho}}\le \frac{\sqrt{m}}{\sqrt{2\rho}}$; 
        \item $\dist\nolimits_2(f_{\tilde{a}},f_a) \le \frac{\eta(\beta;m,2)}{\sqrt{2\rho}}$ with probability at least $1-\beta$,
        \end{enumerate}
        where $\eta(\beta;m,2):={\sqrt{m+2\sqrt{m\ln(1/\beta)}+2\ln(1/\beta)}}$.
\end{lemma}

\paragraph{Multivariate functions.}
It is easy to see that the mechanisms in Theorems~\ref{thm:gp_funcRm} and \ref{thm:cgp_funcRm} hold nearly verbatim for multivariate functions, or equivalently, when the variable $t$ is a vector $(t_1,\dotsb,t_d)$, where $t_i\in I_i\subseteq \mathbb{R}$, with the integrals appropriately defined over $I:=I_1\times\dotsb\times I_d$. The prototypical example is the space of multivariate polynomials over the reals. Note that in this case the computation of the matrix $\Sigma$ is the same, and $\Sigma$ remains $m\times m$, where $m$ is the number of linearly independent functions $\phi_j(\cdot)$'s we choose in the basis.

\begin{remark}
All the mechanisms described in Section~\ref{sec:priv_func} (above and henceforth) generalize to any metric induced by an inner product $\dist(u,v):=\sqrt{\langle u-v, u-v \rangle}$.  
In the most general case, the entries of $\Sigma^{-1}$ become $[\Sigma^{-1}]_{j,l} = \langle \phi_j, \phi_l\rangle$, and the matrix $\Sigma^{-1}$ is also called the Gram matrix of $\{\phi_j\}_{j\in [m]}$.
\end{remark}

\subsection{Vector-valued Functions}
\label{sec:ext_vec_func}
In the above, we discussed privatizing functions of the form $f_a:t \mapsto \sum_{j=1}^m a_j\phi_j(t)$ for real-valued basis functions $\phi_j(\cdot)$, and coefficients $a_j\in \mathbb{R}$ for $j\in [m]$, so $f_a(\cdot)$ is also real-valued. Next, we discuss the scenario where $f_a(\cdot)$ is vector-valued, which arises when either the basis functions $\phi_j(\cdot)$'s or the coefficients $a_j$'s are vector-valued.

\paragraph{Vector-valued basis, real-valued coefficients.}
Suppose the basis functions $\phi_j(\cdot)$'s are vector-valued, with $\phi_j(\cdot)=\left(\phi_{j,1}(\cdot),\dotsb,\phi_{j,n}(\cdot)\right)$ for $j\in [m]$. Let $U_f=\mathrm{span}_{\mathbb{R}}(\{\phi_j\}_{j\in [m]})$. Then the $L^2$ distance $\dist_2(\cdot,\cdot)$ in this case can be computed from 
\begin{align*}
\dist\nolimits_2(f_y,f_{a})^2 &= \int_I \|f_y(t)-f_{a}(t)\|^2 dt\\
&=\int_I \left\|
\begin{bmatrix}
\sum_j (y_j-a_j)\phi_{j,1}(t)  & \dotsb &  \sum_{j\in [m]} (y_j-a_j)\phi_{j,n}(t) \end{bmatrix}^{\mathsf{T}}
\right\|^2 dt\\
&=\int_I \sum_{r=1}^n \left(\sum_{j,l \in [m]} (y_j-a_j)(y_l-a_l)\phi_{j,r}(t)\phi_{l,r}(t)\right) dt\\
&= \sum_{r=1}^n \sum_{j,l\in [m]} (y_j-a_j)(y_l-a_l) \left(\int_I \phi_{j,r}(t)\phi_{l,r}(t)\right) dt\\
&= \sum_{r=1}^n \begin{bmatrix}
    (y_1-a_1) & \dotsb & (y_m-a'_m) 
\end{bmatrix} {\Sigma^{-1}}^{(r)} \begin{bmatrix}
    (y_1-a_1) \\
    \vdots \\
    (y_m-a_m) 
\end{bmatrix} \\
&= \begin{bmatrix}
    (y_1-a_1) & \dotsb & (y_m-a'_m) 
\end{bmatrix} \left(\sum_{r=1}^n {\Sigma^{-1}}^{(r)} \right)\begin{bmatrix}
    (y_1-a_1) \\
    \vdots \\
    (y_m-a_m) 
\end{bmatrix} 
\end{align*} 
Thus, $\Sigma$ is still $m\times m$, and can be computed from $\Sigma^{-1}={\Sigma^{-1}}^{(1)}+\dotsb {\Sigma^{-1}}^{(n)}$, where $[{\Sigma^{-1}}^{(r)}]_{j,l}=\int_I \phi_{j,r}(t)\phi_{l,r}(t) dt$ for $r\in [n]$. The dimension of $\Sigma$ is not affected because although the basis functions are vector-valued, one coefficient is used for each basis function vector, no matter how large it is. 
\begin{figure}[htbp]
    \centering
    \begin{subfigure}[t]{.3\linewidth}
        \includegraphics[width=\textwidth]{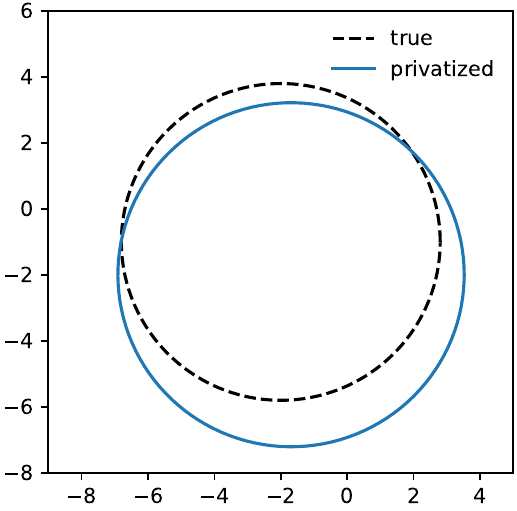}
          \subcaption{$I=[0,2\pi],\Sigma=\frac{1}{2\pi}I_{3\times 3}$}
        \label{fig:example_vecbasis0}
    \end{subfigure}
    \;\;\;\;\;
    \begin{subfigure}[t]{.3\linewidth}
        \includegraphics[width=\textwidth]{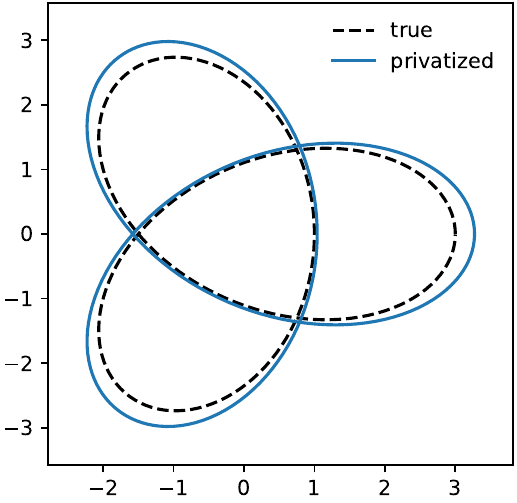}
        \subcaption{$I=[0,4\pi], \Sigma=\frac{1}{4\pi}I_{2\times 2}$} 
        \label{fig:example_vecbasis1}
    \end{subfigure}
    \caption{Example curves with vector-valued basis, fix $\varepsilon=1$. \\
    \eqref{fig:example_vecbasis0}: $[x_1(t),x_2(t)]^{\mathsf{T}}=4.8\cdot[\cos(t),\sin(t)]^{\mathsf{T}}-2\cdot[1,0]^{\mathsf{T}}-1\cdot[0,1]^{\mathsf{T}}$.\\
    \eqref{fig:example_vecbasis1}: $[x_1(t),x_2(t)]^{\mathsf{T}}=2\cdot[\cos(t),\sin(t)]^{\mathsf{T}}+1\cdot[\cos(t/2),\sin(t/2)]^{\mathsf{T}}$. }
    \label{fig:example_vecfunc1}
\end{figure}
\begin{example}
    Consider the basis $\{\phi_1\}$ where the function $\phi_1:t\mapsto [\cos(t),\sin(t)]^{\mathsf{T}}$ with $I=[0,2\pi]$, then $a[\cos(t),\sin(t)]^{\mathsf{T}}$ describes a circle of radius $|a|$ centered at the origin. Using our mechanism above, we only need to privatize the coefficient $a$, which is intuitive since all information about the curve is captured by the radius. If we want to account for shift, we can add the constant vectors $\phi_2(\cdot)\equiv [1,0]^{\mathsf{T}}, \phi_3(\cdot)\equiv [0,1]^{\mathsf{T}}$ as basis functions. See Fig.~\ref{fig:example_vecbasis0} for an illustration of this example, and Fig.~\ref{fig:example_vecbasis1} for another example.
\end{example}

\paragraph{Real-valued basis, vector-valued coefficients.} In some cases, a function might be described more naturally by vector coefficients, while the basis functions are real-valued. 
\begin{figure}[htbp]
    \centering
    \begin{subfigure}[t]{.3\linewidth}
        \includegraphics[width=\textwidth]{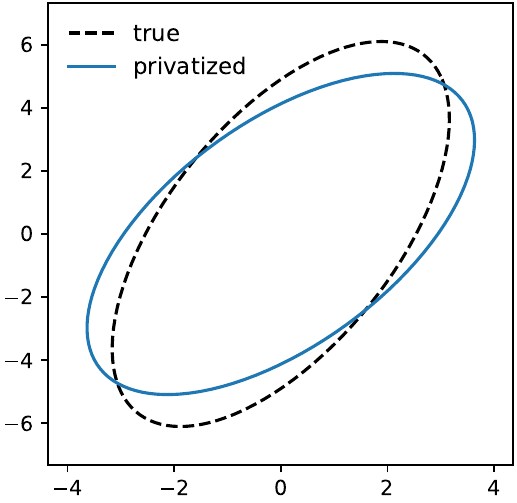}
    \subcaption{$I=[0,2\pi], \Sigma=\frac{1}{\pi}I_{4\times 4}$}
    \label{fig:example_veccoeff0}
    \end{subfigure}
    \;\;\;\;\;
    \begin{subfigure}[t]{.3\linewidth}
        \includegraphics[width=\textwidth]{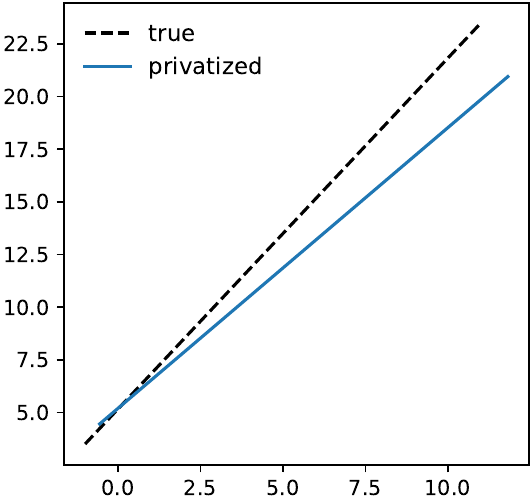}
        \subcaption{$I=[0,4], \Sigma = \begin{bmatrix}
            \frac{3}{16} & \frac{-6}{16}\\
            \frac{-6}{16} & 1
        \end{bmatrix}$} 
        \label{fig:example_veccoeff}
    \end{subfigure}
    \caption{Examples with vector-valued coefficients, fix $\varepsilon=1$. \\
    \eqref{fig:example_veccoeff0}: $[x_1(t),x_2(t)]^{\mathsf{T}}=[3,-1]^{\mathsf{T}}\cdot\cos(t)+[5,3.5]^{\mathsf{T}}\cdot\sin(t)$.\\
    \eqref{fig:example_veccoeff}: $[x_1(t),x_2(t)]^{\mathsf{T}}=[3,-1]^{\mathsf{T}}\cdot t+[5,3.5]^{\mathsf{T}}\cdot 1$. }
    \label{fig:example_vecfunc2}
\end{figure}

\begin{example}
    An ellipse can be described by a two-variable polynomial, or a parametric curve with vector coefficients as
\[
[x_1(t), x_2(t)]^{\mathsf{T}} = a_1 \cos(t) + a_2\sin(t) + a_3
\]
where $a_j = [a_{j,1}, a_{j,2}]^{\mathsf{T}}$ for $j\in [3]$, and $I=[0,2\pi]$. See Fig~\ref{fig:example_veccoeff0} for a zero-centered ellipse.
\end{example}

\begin{example}
\label{exp:2d_linear}
Another example is the basis $\{\phi_1, \phi_2\}$ where $\phi_1:t\mapsto t$, $\phi_2:t\mapsto 1$, and $I=[0,T]$. Then \[
\begin{bmatrix}
    x_1(t)\\
    x_2(t)
\end{bmatrix}
= \begin{bmatrix}
    a_{1,1}\\
    a_{1,2}
\end{bmatrix}\cdot \phi_1(t) + \begin{bmatrix}
    a_{2,1}\\
    a_{2,2}
\end{bmatrix}\cdot \phi_2(t)
\]
describes a $2$D linear curve (see Fig.~\ref{fig:example_veccoeff}), which is useful for approximating $2$D trajectories. 

\end{example}
In general, suppose each $a_j = [a_{j,1},\dotsb, a_{j,n}]^{\mathsf{T}}$ for $j\in [m]$. We have $U_f=\mathrm{span}_{\mathbb{R}^n}(\{\phi_j\}_{j\in [m]})$ where each $\phi_j:I\rightarrow \mathbb{R}$. 
Then
\begin{align*}
	\dist\nolimits_2(f_y,f_a)^2 &= \int_I \|f_y(t)-f_{a}(t)\|^2 dt\\
	&=\int_I\;\Bigl\|
	\begin{bmatrix}
		(y_{1,1}-a_{1,1}) & \dotsb & (y_{1,n}-a_{1,n})
	\end{bmatrix}^{\mathsf{T}} \phi_1(t) +\dotsb \\
	&\qquad\quad+\begin{bmatrix}
		(y_{m,1}-a_{m,1}) & \dotsb & (y_{m,n}-a_{m,n})
	\end{bmatrix}^{\mathsf{T}} \phi_m(t) \Bigr\|^2 dt\\
	&=\int_I \sum_{r=1}^n \left((y_{1,r}-a_{1,r})\phi_1(t)+(y_{m,r}-a_{m,r})\phi_m(t) \right)^2 dt \\
	&=\sum_{r=1}^n \sum_{j,l\in [m]} (y_{j,r}-a_{j,r})(y_{l,r}-a_{l,r})\int_I\phi_{j}(t) \phi_l(t) dt \\
	&= \sum_{r=1}^n \begin{bmatrix}
		(y_{1,r}-a_{1,r}) & \dotsb & (y_{m,r}-a_{m,r})
	\end{bmatrix} \Sigma_0^{-1} \begin{bmatrix}
		(y_{1,r}-a_{1,r}) \\
		\vdots \\
		(y_{m,r}-a_{m,r})
	\end{bmatrix} \\
	&= \begin{bmatrix}
		y_{1,1}-a_{1,1} \\ \vdots \\ y_{m,1}-a_{m,1} \\ \vdots \\ y_{1,r}-a_{1,r} \\ \vdots \\ y_{m,r}-a_{m,r}
	\end{bmatrix}^{\mathsf{T}} \Sigma^{-1} \begin{bmatrix}
		y_{1,1}-a_{1,1} \\ \vdots \\ y_{m,1}-a_{m,1} \\ \vdots \\ y_{1,r}-a_{1,r} \\ \vdots \\ y_{m,r}-a_{m,r}
	\end{bmatrix}
\end{align*}
where
$\Sigma_0^{-1}$ is an $m\times m$ matrix with entries $[\Sigma_0^{-1}]_{j.l}=\int_I\phi_{j}(t) \phi_l(t) dt$, and $\Sigma^{-1}$ is an $mn\times mn$ block matrix, with $\Sigma_0^{-1}$ repeated in the diagonal.

For a finite-dimensional function space $U_f$, let $\dim(U_f)$ denote its dimension; i.e. $\dim(U_f)$ is the number of columns of its associated Gram matrix. From the discussions above and Lemma~\ref{lm:util_funcRm_slap}, we have:
\begin{corollary}
\label{cor:projpriv_util_slap}
    For $f_a=\sum_{j}a_j\phi_j \in U_f:=\mathrm{span}_{\mathbb{R}^{n_a}}(\{\phi_j\}_{j\in [m]})$, let $\tilde{a}=M(a):=a+\frac{1}{\varepsilon
    }\Sigma^{1/2}Z$ where $Z$ is a standard spherical Laplace random vector and $\dim(U_f)= n_a\cdot m$. Then $M$ is $\varepsilon$-GP, and
    \begin{enumerate}
        \item $\mathbb{E}[\dist\nolimits_2(f_{\tilde{a}},f_a)]=\frac{\dim(U_f)}{\varepsilon}$;
        \item $\dist\nolimits_2(f_{\tilde{a}},f_a)\le \frac{\eta(\beta;\dim(U_f),1)}{\varepsilon}$ with probability at least $1-\beta$.
    \end{enumerate}
\end{corollary}
Similarly, from Lemma~\ref{lm:util_funcRm_gauss}, we have:
\begin{corollary}
\label{cor:projpriv_util_gauss}
    For $f_a, U_f$ and $\dim(U_f)$ as in Corollary~\ref{cor:projpriv_util_slap},   let $\tilde{a}=M(a):=a+\frac{1}{\sqrt{2\rho}
    }\Sigma^{1/2}Z$ where $Z$ is a standard Gaussian random vector. Then $M$ is $\rho$-CGP, and
    \begin{enumerate}
        \item $\mathbb{E}[\dist\nolimits_2(f_{\tilde{a}},f_a)]=\sqrt{2}\frac{\Gamma(\dim(U_f)/2+1/2)}{\Gamma(\dim(U_f)/2)\sqrt{2\rho}}\le \frac{\sqrt{\dim(U_f)}}{\sqrt{2\rho}}$;
        \item $\dist\nolimits_2(f_{\tilde{a}},f_a)\le \frac{\eta(\beta;\dim(U_f),2)}{\sqrt{2\rho}}$ with probability at least $1-\beta$.
    \end{enumerate}
\end{corollary}

\subsection{Arbitrary Functions}
\label{sec:gen_func}
So far, we have focused on functions that are linear combinations of a given set of basis functions.  In this subsection, we consider the most general case where the function $q$ to be privatized is not assumed to take any specified form, i.e., it is taken from $L^2(I)$, the entire space of square-integrable functions. Note that this is an infinite-dimensional space.  We shall focus on GP algorithms; translating such algorithms to their CGP variants is straightforward, with details deferred to Appendix~\ref{app:priv_CGP}.

\paragraph{Project-and-privatize.} 
A natural approach to reduce this problem to the previous one would be to first project $q$ into $U_f$ defined with respect to some basis $\{\phi_j\}_j$, then privatize its projection. Specifically, let $\pj_{U_f}:U\rightarrow U_f$ be a function that computes the least squares {projection} onto $U_f$. We also write $\pj_{U_f}(q) =f_{a(q)}$ where $a(q)$ denotes the coefficients of the approximation. We summarize this simple approach in Algorithm~\ref{alg:algo_projpriv}.

\begin{algorithm}
\caption{Project-and-Privatize}
    \label{alg:algo_projpriv}
    \begin{flushleft}
    \vspace{-8pt}
    \textbf{Input}: $q:I\rightarrow \mathbb{R}^n$; finite-dimensional function space $U_f$; $\varepsilon >0$\\
    \textbf{Output}: privatized $\tilde{q}$
    \end{flushleft}
    \vspace{-5pt}
    \begin{algorithmic}[1]
    \STATE compute $\Sigma$ from $U_f$
    \STATE $a\gets $ coefficients of $\pj_{U_f}(q)$
    \STATE ${\tilde{a}}\gets a+\frac{1}{\varepsilon}\Sigma^{1/2} Z$, $Z\sim\mathrm{SLap}(\dim(U_f))$
    \RETURN $f_{\tilde{a}}$
    \end{algorithmic}
\end{algorithm}

To prove the privacy of Algorithm~\ref{alg:algo_projpriv}, we need the following standard result (e.g. see \cite{deutsch2001best}):

\begin{lemma}
\label{lm:fidim_lipschitz}
    Let $U_f$ be a finite-dimensional function space. Then for any pair of $q, q'\in U$, $\dist_2(f_{a(q)},f_{a(q')})\le \dist_2(q,q')$, where $f_{a(q)},f_{a(q')}$ denote the least squares projections of $q$ and $q'$ onto $U_f$, respectively.
\end{lemma}

\begin{lemma}
    Algorithm~\ref{alg:algo_projpriv} is $\varepsilon$-GP.
\end{lemma}
\begin{proof} 
    Let $M_f:U_f\rightarrow U_f$ be the mechanism that privatizes $\pj_{U_f}(q)$, which is $\varepsilon$-GP on $(U_f,\dt)$ by previous discussions. Let $M:U\rightarrow U_f$ denote the mechanism in Algorithm~\ref{alg:algo_projpriv}. We have for any measurable $S_f\subseteq U_f$ and any pair $q,q'\in U$
    \begin{align*}
    \Pr[M(q)\in S_f] &= \Pr[M_f(\pj\nolimits_{U_f}(q))\in S_f]\\
    &\le  e^{\varepsilon \dist_2\left(f_{a(q)},f_{a(q')}\right)} \Pr[M_f(\pj\nolimits_{U_f}(q'))\in S_f]\\
    &= e^{\varepsilon \dist_2\left(f_{a(q)},f_{a(q')}\right)} \Pr[M(q')\in S_f].
    \end{align*}
    Thus, if we can show that $\dist_2\left(f_{a(q)},f_{a(q')}\right) \le \dist_2(q,q')$, then $M$ satisfies the GP requirement on $U$ w.r.t. $\dt$. The required inequality then follows from Lemma \ref{lm:fidim_lipschitz}.
\end{proof}

The detailed formulas for computing the least squares projections for both real-valued and vector-valued functions are provided in Appendix~\ref{sec:app_ls}.

Next, we discuss the $L^2$ error of Algorithm~\ref{alg:algo_projpriv}. By simple triangle inequality, we have $\mathbb{E}[\dt(q,f_{\tilde{a}})]\le \dt(q,f_a)+\frac{1}{b}\mathbb{E}[\|Z\|]$. Below, we provide a tighter analysis. 
\begin{lemma}
\label{lm:finapprox_priv_distsquared}
    Let $U_f$ be the subspace given by the span of $\{\phi_j(\cdot)\}_{j\in [m]}$. For $q\in U$, let $f_a\in U_f$ be a projection of $q$ in $U_f$. For $\tilde{a}:=a+\frac{1}{b}\Sigma^{1/2}Z$ where $Z$ is a random vector with mean zero, we have
        \[
    \mathbb{E}[\dt(q,f_{\tilde{a}})^2] = \dt(q,f_{a})^2 + \frac{1}{b^2}\mathbb{E}[\|Z\|^2].
    \]
\end{lemma}

As one would expect, the simple mechanism incurs large error if $q$ cannot be approximated well by any $f_a\in U_f$.  On the other hand, if $U_f$ has large dimension, then we might still end up with a large error from the noise vector $\|Z\|$ even if $q$ is well approximated by $U_f$. Thus, the challenge is to select a basis $\{\phi_j\}_j$ that strikes a balance between the two sources of error.  Critically, the selection must also be done in a private manner.  

\paragraph{A private, balanced selection.}
Let $\{\phi_j(\cdot)\}_{j\in J}$ be a pre-selected collection of basis functions  indexed by $J$. We assume that $q$ can be approximated by these basis functions to some accuracy $\gamma =O(1/\varepsilon)$ that the analyst is content with, absent privacy. Otherwise, a general-purpose dense basis allowing arbitrary $\gamma>0$ can be used, where $J$ is infinite, such as Hermite functions for $L^2(\mathbb{R})$,  the Fourier basis $\{1,\sin(jt),\cos(jt)\}_{j\in \mathbb{N}}$ for $L^2([-\pi,\pi])$, or monomials $\{1,t^j\}_{j\in \mathbb{N}}$ for $L^2([0,T])$.
Our algorithm (Algorithm~\ref{alg:algo_privsel}) follows roughly these steps:
\begin{enumerate}
    \item Identify a large enough basis that approximates $q$ well;
    \item Project and privatize $q$, obtain the first approximation $\tilde{q}_0$;
    \item Based on the coefficients of $\tilde{q}_0$, identify and eliminate the ones with (relatively) small magnitude, as they contribute less to the approximation quality but increase the noise;
    \item Project $q$ onto the space spanned by the remaining basis functions, and privatize.
\end{enumerate}
\begin{algorithm}
\caption{PrivFuncSelect}
    \label{alg:algo_privsel}
    \vspace{-8pt}
    \begin{flushleft}
    \textbf{Input}: $q:I\rightarrow \mathbb{R}^n; \{\phi_j\}_{j\in J};\beta > 0; \varepsilon >0$\\
    \textbf{Output}: privatized $\tilde{q}$
    \end{flushleft}
    \vspace{-5pt}
    \begin{algorithmic}[1]
    \STATE $k_0 \gets \mathrm{SVT}(q,(\frac{\varepsilon}{12},\frac{\varepsilon}{6}),-\frac{4}{\varepsilon},1,g_1,g_2,\dots)$, $g_j(\cdot):=-\dist_2(\pj_{U_f^{(j)}}(\cdot), \cdot)$
    \STATE $\tilde{q}_0\gets \mathrm{Project}$-$\mathrm{and}$-$\mathrm{Privatize}(q,U_f^{(k_0)},\varepsilon/4)$
    \STATE $\tilde{c}\gets$ coefficients of $\tilde{q}_0$ in $U_f^{(k_0)}$
    \STATE $s^{(r)}\gets r$th largest magnitude among $\tilde{c}$, $r\in [k_0]$
    \STATE ${k_1}\gets \mathrm{SVT}(q,(\frac{\varepsilon}{12},\frac{\varepsilon}{6}),0,1,g'_1,g'_2,\dotsb)$, $N_r:=\{j: \|\tilde{c}_j\| \ge s^{(r)}\}$, $g'_r(\cdot):=\frac{n|N_r|}{\varepsilon/4}-\dist_2(f_{b^{(r)}}, \cdot)$, $b_j^{(r)}:=\tilde{c}_j\cdot\mathbb{1}\{j\in N_r\}$
    \STATE $U_f \gets \mathrm{span}(\{\phi_j: j\in N_{k_1}\})$
    \STATE $f_{\tilde{a}}\gets \mathrm{Project}$-$\mathrm{and}$-$\mathrm{Privatize}(q,U_f,\varepsilon/4)$
    \RETURN $f_{\tilde{a}}$
    \end{algorithmic}
\end{algorithm}

For $j\ge 1$, let $U_f^{(j)}$ denote the subspace spanned by the first $j$ basis functions, Let $\pj_{U_f^{(j)}}(q)$ denote the least-squares {projection} of $q$ in $U_f^{(j)}$. In step ($1$) above, we want to find $j$ such that $\pj_{U_f^{(j)}}(q)$ is sufficiently small (e.g. $O(1/\varepsilon)$). We use SVT on $g_1,g_2,\dots$ where each $g_j(q):=-\dist_2(\pj_{U_f^{(j)}}(q), q)$ is $1$-Lipschitz, as in line $1$ of Algorithm~\ref{alg:algo_privsel}. Then we apply project-and-privatize in line $2$-$3$, where $\tilde{c}$ are the privatized coefficients of $\tilde{q}_0$. To decide which basis functions to keep, we order the noisy coefficients in $\tilde{c}$ based on their magnitude, and create candidate sets $N_r:=\{j: \|\tilde{c}_j\| \ge s^{(r)}\}$, where $s^{(r)}$ is the $r$th largest magnitude.  
We use  $b_j^{(r)}=\tilde{c}_j\cdot\mathbb{1}\{j\in N_r\}$ to approximate the projection that corresponds to muting the coefficients not in $N_r$.
We run SVT to find which candidate set has approximately balanced noise and approximation error, by using the sequence $g'_1,g'_2,\dotsb$ where $g'_r(\cdot):=\frac{n|N_r|}{\varepsilon/4}-\dist_2(f_{b^{(r)}}, \cdot)$, which are $1$-Lipschitz, and finding where they first cross zero, as in line $5$. Finally, we perform a final project-and-privatize, using the basis set $N_{k_1}$ identified by the second SVT call, as in lines $6$-$7$.

\paragraph{Privacy.} Algorithm~\ref{alg:algo_privsel} is $\varepsilon$-GP by basic composition in Lemma~\ref{lm:basic_comp}, where $\varepsilon/4$ is used in each of lines $1$, $2$, $5$, $7$.
\paragraph{Utility.} 
The goal of the first SVT is to identify a large enough subset of basis functions from a given basis that is assumed to approximate $q$ well, while that of the second SVT is to further refine the selection by approximately balancing the privatization noise and approximation error. By the generic guarantee of SVT in Lemma~\ref{lm:svt_util}, both goals are approximately achieved: In each SVT call, we might overshoot the target in value by $O(\log(1/\beta)/\varepsilon)$ resulting in an extra subset of basis functions of size $O(\log(1/\beta))$, or undershoot in value from the threshold by at most $O\left(\left(\log({1}/{\beta})+\log ({k^*}/{\beta})\right)/\varepsilon\right)$, where $k^*$ denotes the index of the balance point. 
Thus, Algorithm~\ref{alg:algo_privsel} meets our goal of providing a roughly balanced selection.

While Algorithm \ref{alg:algo_privsel} has addressed the problem of how to select a set of good basis functions from a given collection, we have not touched the issue of how to decide which collection to use.  Certainly, the choice of the basis also has a large impact on the utility of the algorithm, but this is largely application-dependent. For example, if the functions tend to be periodic, then the Fourier basis is preferred; for data arising from signal processing (e.g. ECG data), sinc functions can be used; 
if the function is a curve, then piecewise basis functions are better (see Section \ref{sec:privcurve}).  Thus, this decision is best made at the application level. 

\subsection{Projecting $[0,T]$-functions onto $\mathbb{R}$-functions}
\label{sec:infbasis_finitefunc}
In previous discussions, we had assumed that both the function $q$ and the basis $\{\phi_j\}_{j}$ are functions on the same domain $I$. In fact, it is possible to separate the domains, in particular for the case where one is $\mathbb{R}$ and the other is a finite interval.
In some cases, it might even be beneficial to do so, for example when it gives a matrix $\Sigma$ that is easier to work with. Let $\mathrm{sinc}(t):=\frac{\sin(\pi t)}{\pi t}$ be the normalized sinc function. The set of functions $\{\mathrm{sinc}(t-k)\}_{k\in \mathbb{Z}}$ is square-integrable on $\mathbb{R}$ (hence also on $[0,T]$); moreover, the set $\{\mathrm{sinc}(t-k)\}_{k\in \mathbb{Z}}$ is an orthonormal basis with $\int_{-\infty}^{\infty} \mathrm{sinc}(t-k)^2=1$ and $\int_{-\infty}^{\infty} \mathrm{sinc}(t-k)\cdot\mathrm{sinc}(t-l) dt =0$ for $k\neq l$. We would like to use this readily available orthonormal basis, which is also commonly used for approximating functions that arise in signal processing.
We will show that using such a basis in $L^2(\mathbb{R})$ for a curve in $C([0,T])$ only requires little additional work\footnote{Note the approximation quality could be different, see discussion in Appendix~\ref{sec:app_bounded_via_unbounded}}, and the privacy guarantee follows similarly to the discussion above. Here, $I$ always refers to the domain of the function $q$, and $\dt$ is defined in terms of $I$.
\paragraph{Using $L^2(\mathbb{R})$ for   $(C([0,T]),\dist_{2})$.}  We consider the case where 
$q\in U=C([0,T])$ is a continuous function on the interval $[0,T]$, but its privatization uses functions from $L^2(\mathbb{R})$. Define a function $L:U\rightarrow L^2(\mathbb{R})$ by
\[
L(q)(t):=\begin{cases}
    q(t) & t\in [0,T]\\
    0 & \text{otherwise}.
\end{cases}
\]
I.e., $L(q)$ is the extension of $q$ which takes inputs from $\mathbb{R}$. 
Let $L(U):=\{L(q): q\in U\}$. 
Notice $\int_{-\infty}^{\infty} L(q)(t)^2 dt = \int_{[0,T]} q(t)^2 dt <\infty $, so $L(U)\subset L^2(\mathbb{R})$. 
Given an $m$-dimensional subspace $L_f\subseteq L^2(\mathbb{R})$, let $\pj:U \rightarrow L_f$ be the function that computes the least squares {projection} of $L(q)$ in $L_f$, i.e. $\pj:q\mapsto f_{a(L(q))}$. Let us write $\dist_{2,\mathbb{R}}(u,u')=\int_{-\infty}^{\infty} (u(t)-u'(t))^2 dt$ for $u, u'\in L^2(\mathbb{R})$.
Now, from Section~\ref{sec:finitedim_func} we have a GP mechanism $M_f$ for privatizing $\pj(q)\in L_f$ w.r.t. $\dist_{2,\mathbb{R}}$. 
Define a mechanism $M_1: U\rightarrow L_f$, where $M_1(q) = M_f(\pj(q))$.
Then $M_1$ satisfies for all $q,q'$, and all measurable $S_f\subseteq L_f$
\begin{align*}
\Pr[M_1(q)\in S_f] &= \Pr[M_f(\pj(q))\in S_f]\\
&\le e^{\varepsilon\dist_{2,\mathbb{R}}(\pj(q),\pj(q'))}\Pr[M_f(\pj(q'))\in S] \\
&= e^{\varepsilon\dist_{2,\mathbb{R}}(\pj(q),\pj(q'))}\Pr[M_1(q')\in S_f].
\end{align*}
Also, since $L(U)\subset L^2(\mathbb{R})$, Lemma~\ref{lm:fidim_lipschitz} holds when we replace the metric space $(U,\dist_2)$ with $(L(U),\dist_{2,\mathbb{R}})$ in the statement. Thus,
\begin{align*}
\dist\nolimits_{2,\mathbb{R}}(\pj(q),\pj(q')) &= \dist\nolimits_{2,\mathbb{R}} (f_{a(L(q))}, f_{a(L(q'))}) \overset{Lemma~\ref{lm:fidim_lipschitz}}{\le} \dist\nolimits_{2,\mathbb{R}}(L(q),L(q')) = \dt(q,q');
\end{align*}
i.e. $M_1$ is $\varepsilon$-GP w.r.t. $\dist_2$. 

\section{Privatization of Curves}
\label{sec:privcurve}
A curve is a continuous function $[0,T]\rightarrow R^n$.  
Unlike the case of arbitrary function privatization which assumes a suitable $U_f$ is given, for privatizing curves there are well known sets of basis functions that approximate continuous functions well. Based on this fact, in this section we provide another algorithm more suitable for curves.

\subsection{Privatization of piecewise functions}
Before we present our algorithm for arbitrary curve privatization, we discuss how piecewise functions can be privatized using the framework developed in Section~\ref{sec:finitedim_func}.  For ease of discussion we assume the coefficients are real-valued in Sections~\ref{sec:def_piecewise} and \ref{sec:ensure_cont} below; the discussions adapt to the vector-valued case by repeating the operations in each coordinate, as explained in Section~\ref{sec:ext_vec_func}.

\subsubsection{Defining piecewise basis functions}
\label{sec:def_piecewise}
Let $\mathbb{1}_A:I\rightarrow \{0,1\}$ be the indicator function that checks whether a point is in the set $A$, i.e., $\mathbb{1}_A(t)=1$ if $t\in A$ and $0$ otherwise. Clearly, two indicator functions $\mathbb{1}_A(\cdot)$ and $\mathbb{1}_{A'}(\cdot)$ are linearly independent if $A\neq A'$. Let $A_s=[T_{s-1},T_{s})$ for $s\in [k]$, and $0=T_0 < T_1 < \dotsb < T_k=T$. 
If a set of functions $\{\phi_j(\cdot)\}_{j \in [m]}$ is linearly independent on each $A_s$,
then the set $\{\phi_{s,j}\}_{s\in [k], j\in [m]}$, where $\phi_{s,j}(t):=\mathbb{1}_{A_s}(t)\phi_j(t)$, is also linearly independent on $[0,T]$ {as long as each $\phi_{s,j}(\cdot)$ is not identically zero}. Since the latter set has size $mk$, the matrix $\Sigma$ becomes $mk\times mk$. In particular, ordering the coefficients as $a=[a_{1,1},\dots,a_{1,m},a_{2,1},\dots,a_{k,m}]^{\mathsf{T}}$, so 
$f_a(t):=\sum_s \sum_j a_{s,j}\cdot\mathbb{1}_{A_s}(t)\phi_j(t)$, 
we get 
\[\Sigma^{-1} := \begin{bmatrix}
	{\Sigma^{(1)}}^{-1} & \mathbf{0}          & \cdots & \mathbf{0} \\
	\mathbf{0}          & {\Sigma^{(2)}}^{-1} & \cdots & \mathbf{0} \\
	\vdots              & \vdots              & \ddots & \vdots \\
	\mathbf{0}          & \mathbf{0}          & \cdots & {\Sigma^{(k)}}^{-1}
\end{bmatrix},\]
where each ${\Sigma^{(s)}}^{-1}$ is $m\times m$, with $[{\Sigma^{(s)}}^{-1}]_{j,l} = \int_{t\in A_s} \phi_j(t)\phi_l(t) dt$ for $s\in [k]$. For example, applying to piecewise linear functions with $\{\phi_1:t
\mapsto t, \phi_2:t\mapsto 1\}$, we get ${\Sigma^{(s)}}^{-1}=\begin{bmatrix}
    \frac{1}{3}(T_s^3-T_{s-1}^3) & \frac{1}{2}(T_s^2-T_{s-1}^2) \\
    \frac{1}{2}(T_s^2-T_{s-1}^2)  & T_s-T_{s-1}
\end{bmatrix}$, which gives
\[\Sigma^{(s)}=\begin{bmatrix}
    \frac{12} {(T_s-T_{s-1})^3}& \frac{-6(T_s+T_{s-1})} {(T_s-T_{s-1})^3}\\
    \frac{-6(T_s+T_{s-1})} {(T_s-T_{s-1})^3} & \frac{4(T_s^2+T_sT_{s-1}+T_{s-1}^2)}{(T_s-T_{s-1})^3}
\end{bmatrix}.\]
Thus, although the matrix $\Sigma$ for piecewise linear functions is $2k\times 2k$, matrix multiplication can be done at each sub-matrix $\Sigma^{(s)}$
after obtaining the $2k$-dimensional random vector $Z$, i.e., the number of multiplications is $4k$ instead of $4k^2$.
\begin{remark}
    N.B. Although the multiplications can happen at the sub-vector level, the noise vector $Z$ from the spherical Laplace distribution must be drawn together as a single vector. Drawing the noise vector together enables an effect that is reminiscent of \textit{parallel composition} in DP; using separately drawn sub-vectors would require splitting the privacy budget among the component privatizations that use these sub-vectors, which would not be ideal.
    
    On the other hand, Gaussian noises can be separately drawn, as uncorrelated Gaussian random variables are also independent, a desirable property enjoyed solely by the Gaussian distribution!
\end{remark}

\subsubsection{Ensuring continuity}
\label{sec:ensure_cont}
The privatization mechanism where noise is added to each segment described above might result in a privatized curve that is only continuous within each segment, and discontinuous at the segment end points $T_s$'s. If continuity is desired, a post-processing step may be added where the privatized curve is replaced with a continuous one that is nearest to it. That is, given privatized function $f_{\tilde{a}}$, we find $f_{\hat{a}}$ with $\lim_{t\uparrow T_s} f_{\hat{a}}(t)=f_{\hat{a}}(T_s)$ for $1\le s\le k-1$, such that $\dist_2(f_{\tilde{a}},f_{\hat{a}})$ is minimized. Such a process may increase the initial error by at most two times. Indeed, let $f_a\in U_f$ denote the original continuous function, then 
\[\dist\nolimits_2(f_{\hat{a}},f_a)\le \dist\nolimits_2(f_{\hat{a}},f_{\tilde{a}})+\dist\nolimits_2(f_{\tilde{a}},f_a) \le 2\dist\nolimits_2(f_{\tilde{a}},f_a),\]
where the last inequality is due to $f_{\hat{a}}$ being the nearest continuous curve to $f_{\tilde{a}}$ in $U_f$. 

Equivalently, we want to minimize $\dist_2(f_y,f_{\tilde{a}})^2=(y-\tilde{a})^{\mathsf{T}}\Sigma^{-1}(y-\tilde{a})=y^{\mathsf{T}}\Sigma^{-1}y - 2\tilde{a}^{\mathsf{T}}\Sigma^{-1}y+\tilde{a}^{\mathsf{T}}\tilde{a}$, and $f_y$ is continuous at $T_s$ iff $\sum_{j\in [m]} y_{s,j}\cdot\phi_j(T_s)= \sum_{j\in [m]} y_{s+1,j}\cdot\phi_j(T_s)\iff \sum_{j\in [m]} (y_{s,j}-y_{s+1,j})\phi_j(T_s)=0$, for $1\le s\le k-1$.
I.e., the projection process can be formulated as the following quadratic program with linear equality constraints:
\begin{align*}
	\underset{y}{\text{minimize\;\;}}  &\frac{1}{2}y^{\mathsf{T}}\Sigma^{-1} y - \tilde{a}^{\mathsf{T}}\Sigma^{-1} y\\
	\text{subject to\;\;} & C y = 0,
\end{align*}
where $y = [y_{1,1},\cdots,y_{1,m},\cdots,y_{k,1},\cdots,y_{k,m}]^{\mathsf{T}}\in\mathbb{R}^{km}$, and the matrix $C\in\mathbb{R}^{(k-1)\times(km)}$ is
\begin{align*}
	C:=\begin{bmatrix}
		\phi_{1:m}(T_1) & -\phi_{1:m}(T_1) & \mathbf{0}       & \cdots           & \mathbf{0} \\
		\mathbf{0}      & \phi_{1:m}(T_2)  & -\phi_{1:m}(T_2) & \cdots          & \mathbf{0} \\
		\vdots          & \vdots           & \vdots           & \ddots             & \vdots     \\
		\mathbf{0}      & \mathbf{0}       & \mathbf{0}       & \cdots  &-\phi_{1:m}(T_{k-1})
	\end{bmatrix}
\end{align*}
where $\phi_{1:m}(T_s)=[\phi_1(T_s),\cdots,\phi_m(T_s)]\in\mathbb{R}^{m}$.

\subsection{Privately approximating curves via piecewise functions}
In this subsection, we present an algorithm for arbitrary curve privatization using piecewise basis functions. 
Our algorithm tries to find a suitable finite-dimensional $U_f$, and performs a $\mathrm{Project}$-$\mathrm{and}$-$\mathrm{Privatize}$ of $q$ onto $U_f$. The space $U_f$ is generated from a set of sub-intervals and an input real-valued basis $\{\phi_j\}_j$, and has the property that the error from the approximation $\pj_{U_f}(q)$ is roughly equal to $\dim(U_f)$. 
The process to find $U_f$ is composed of two parts, provided in sub-algorithms $\mathrm{PrivFuncSeg}$ (Algorithm~\ref{alg:algo1}) and $\mathrm{ReduceSeg}$ (Algorithm~\ref{alg:algo_redseg}), which we discuss in detail below.

\paragraph{$\mathrm{\textbf{PrivFuncSeg}}$.} We are given a (small) set of basis functions $\{\phi_j\}_j$, which should provide good approximations to continuous functions when restricted to a small sub-interval. The goal of $\mathrm{PrivFuncSeg}$ is find the number of sub-intervals, such that when we fit $\{\phi_j\}_j$ to each sub-interval, we get an overall curve which roughly balance the two sources of error mentioned above. We do so via SVT where we double the number of sub-intervals in each check. 
Let $q:[0,T]\rightarrow \mathbb{R}^n$ be the input curve. For $j=0,1,2,\dotsb$, let 
$\pj_{U_f^{(j)}}(\cdot)$ compute the best approximation using $2^j$ sub-intervals with end points at $\{r\frac{T}{2^j}:r=0,1,2,\dotsb, 2^j\}$; i.e. $\pj_{U_f^{(j)}}(q)$ is the least squares projection of $q$ onto the space  $U_f^{(j)}:=\mathrm{span}_{\mathbb{R}^n}\left(\{\phi_{l}\cdot\mathbb{1}_{A_{r}^{(j)}}\}_{l\in [m],r\in [2^j]}\right)$, where $A_{r}^{(j)}:=[(r-1)\frac{T}{2^j},r\frac{T}{2^j})$. We run $\mathrm{SVT}$ on the queries $g_j(q):=\tau_j-\dist_2(\pj_{U_f^{(j)}}(q), q)$ with privacy budget $\frac{\varepsilon}{4}$, where $\tau_j:=\frac{\dim(U_f^{(j)})}{\varepsilon/4}=\frac{2^jmn}{\varepsilon/4}$ is the expected error due to privatization of $\pj_{U_f^{(j)}}(q)$. Note that $\dist_2(\pj_{U_f^{(j)}}(\cdot), \cdot)$ is $1$-Lipschitz, and so is $g_j(\cdot)$. Moreover, 
the sequence $g_j(q)$ strictly increases, since $\tau_j$ strictly increases while $\dist_2(\pj_{U_f^{(j)}}(q), q)$ is non-increasing with $j$. 
\begin{algorithm}[t]
\caption{PrivFuncSeg}
    \label{alg:algo1}
    \vspace{-8pt}
    \begin{flushleft}
    \textbf{Input}: $q:[0,T]\rightarrow \mathbb{R}^n; \{\phi_j\}_{j\in [m]};\beta > 0; \varepsilon >0$\\
    \textbf{Output}: privatized $\tilde{q}$
    \end{flushleft}
    \vspace{-5pt}
    \begin{algorithmic}[1]
    \STATE $\bar{k} \gets \mathrm{SVT}(q,(\varepsilon/12,\varepsilon/6),0,1,g_0,g_1,g_2,\dotsb)$, where $\tau_j=\frac{2^jmn}{\varepsilon/4}$, $g_j(q):=\tau_j-\dist_2(\pj_{U_f^{(j)}}(q), q)$, $\pj_{U_f^{(j)}}$ computes best approximation using $2^{j}$ sub-intervals
    \STATE $B\gets 3\varepsilon/4$, privacy budget for remaining steps
    \STATE $S\gets \{j{T}/{2^{\bar{k}}}: j=0,1,2,\dotsb,2^{\bar{k}}\}$
    \STATE $k_1\gets \min(\bar{k}-2,4)$
    \IF{$k_1\ge 1$}
    \FOR{$i=1,\dots,4$}
    \STATE $\hat{S}_i\gets \mathrm{ReduceSeg}(q,(i-1)T/4,i T/4,k_1,1,S_i,\beta,B,\varepsilon/16)$, $S_i:= S\cap [(i-1)T/4,i T/4]$
    \ENDFOR
    \ENDIF
    \STATE $\hat{S}\gets \cup_{i} \hat{S}_i=\{0=T_0,T_1,\dots,T_{N}=T\}$;\;\; $N\gets |\hat{S}|-1$
    \STATE $U_f\gets \mathrm{span}(\{\phi_j\cdot \mathbb{1}_{[T_{s-1},T_s)}\}_{s\in [N], j\in [m]})$
    \STATE $f_{\tilde{a}}\gets \mathrm{Project}$-$\mathrm{and}$-$\mathrm{Privatize}(q,U_f,B)$
    \RETURN $f_{\tilde{a}}$
    \end{algorithmic}
\end{algorithm}

\begin{algorithm}[h]
\caption{ReduceSeg}
    \label{alg:algo_redseg}
    \vspace{-5pt}
    \begin{flushleft}
    \textbf{Input}: $q:[0,T]\rightarrow \mathbb{R}^n; [t_{s},t_{e})\subset [0,T]; k_1; l\le k_1;$ set $S$ containing breakpoints to be checked; $\beta > 0; B\ge \varepsilon' >0$\\
    \textbf{Output}: $\hat{S}$ after reducing the number of sub-intervals
    \end{flushleft}
    \vspace{-5pt}
    \begin{algorithmic}[1]
    \IF{$l > k_1$}
    \RETURN ${S}$
    \ENDIF
    \STATE $k\gets \log_2(|S|-1)$; 
    \STATE $S'\gets \{2r\frac{t_e-t_s}{2^k}+t_s: r=0,\dotsb,2^{k-1}\}$\;\;\; \COMMENT{try to reduce number of sub-intervals by half}
    \STATE $f_{[t_s,t_e]} \gets$ best approx. of $q$ on $[t_s,t_e]$ using break points in $S'$
    \STATE $\mathrm{err}\gets \dist_2(f_{[t_s,t_e]},q\cdot \mathbb{1}_{[t_s,t_e]} )+\frac{2^l}{\varepsilon'}Z$, $Z\sim\mathrm{Lap}(1)$
    \STATE $B\gets B-\frac{\varepsilon'}{2^l}$
    \IF{$\mathrm{err}+\frac{2^l}{\varepsilon'}(k_1\ln(2)+\ln(1/\beta))\le \frac{2^{k-1}emn}{2(e-1)B}$} 
    \STATE $S'_L\gets S'\cap [t_s,(t_s+t_e)/2]$,\; $S'_R\gets S'\cap [(t_s+t_e)/2,t_e]$
    \STATE $\hat{S}_L\gets \mathrm{ReduceSeg}(q, t_s, (t_s+t_e)/2, k_1, l+1, S'_L,\beta,{B,}\varepsilon'/2)$
    \STATE $\hat{S}_R\gets \mathrm{ReduceSeg}(q, (t_s+t_e)/2, t_e, k_1, l+1,S'_R,\beta,{B,}\varepsilon'/2)$
    \STATE $\hat{S} \gets \hat{S}_L \cup \hat{S}_R$
    \ELSE
    \STATE $\hat{S}\gets S$
    \ENDIF
    \RETURN $\hat{S}$
    \end{algorithmic}
\end{algorithm}
\begin{figure}[h]
    \vspace{5pt}
    \centering
\includegraphics[width=0.65\textwidth]{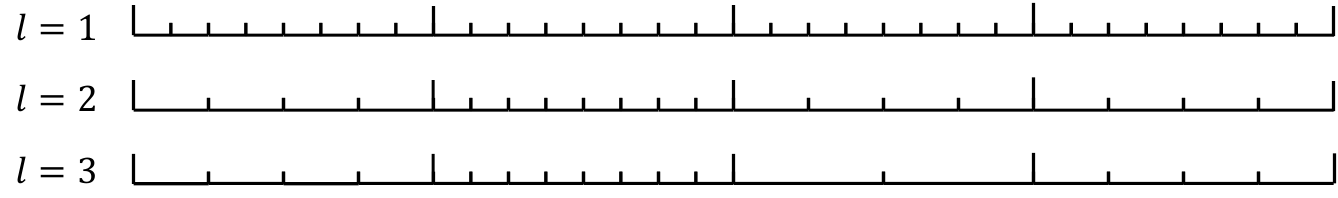}
        \caption{{Example illustration for sub-interval removal process in $\mathrm{ReduceSeg}$. $32$ sub-intervals at the start of level $1$; $20$ sub-intervals at level $2$; $18$ sub-intervals at level $3$, and the  process terminates here.}}
        \label{fig:redseg_rec}
\end{figure}
\vspace{-5pt}
\paragraph{$\mathrm{\textbf{ReduceSeg}}$.} After we find that $2^{\bar{k}}$ (evenly spaced) sub-intervals suffice to roughly balance the two sources of error, the goal of $\mathrm{ReduceSeg}$ is to remove sub-intervals that are not needed, since the curve might be flatter in some segment of $[0,T]$. 
To this end, we partition the sub-intervals into $4$ parts, where we then try to reduce the number of sub-intervals in each part. 
For each part, we perform at most $2^{4}$ checks over at most $k_1\le 4$ levels. Let $\varepsilon'$ be the allocated budget for the current interval $[t_s,t_e]$. Starting at level $l=1$, we use $\frac{2^l}{\varepsilon'}$ to check whether reducing the number of segments by half results in a decrease in error. If so, we remove half of the segments and proceed to check the children intervals $[t_s,(t_e+t_s)/2]$ and $[(t_e+t_s)/2,t_e]$ at level $l=2$, and so on, until we finish level $l=k_1$. See Fig.~\ref{fig:redseg_rec} for an example illustration of the sub-interval removal process corresponding to running $\mathrm{ReduceSeg}$ on all four parts together (lines $5$-$9$ of Algorithm~\ref{alg:algo1}).  
It is straightforward to check that $\mathrm{ReduceSeg}$ consumes at most $\varepsilon'$ of the privacy budget, since at each level $l$ there are at most $2^{l-1}$ calls, each uses $\frac{1}{2^l}$ of $\frac{\varepsilon'}{2^{l-1}}$, for a total of $\sum_{l=1}^{k_1} 2^{l-1}\frac{\varepsilon'}{2^{l-1}}\frac{1}{2^l}=(1-\frac{1}{2^{k_1}})\varepsilon'$.

\paragraph{Privacy.} The entire algorithm satisfies $\varepsilon$-GP, since $\varepsilon/4$ is used in the SVT call, at most $\varepsilon/4$ is used for $\mathrm{ReduceSeg}$, and the remaining is used for $\mathrm{Project}$-$\mathrm{and}$-$\mathrm{Privatize}$; the remaining budget depends on the actual number of times $\mathrm{ReduceSeg}$ is invoked, but the overall privacy guarantee is the same by a privacy filter argument (e.g. see \cite{feldman2021individual,liang2026adaptive}).
\paragraph{Utility.} In the SVT call, the potential overshoot in the target position is $O(\log\log(1/\beta))$, since the queries $g_j$ increase at a rate of $2\times$ after the balance point. This also means the number of sub-intervals $2^{\bar{k}}$ is potentially $O(\log(1/\beta))$ times larger than the size at the balance point, hence the need to reduce. $\mathrm{ReduceSeg}$ is capable of reducing the number by a factor of $2^{k_1}$, comparable to $O(\log(1/\beta))$ for any constant $\beta$, if every call results in a reduction. In any case, we show below that in the worst case, $\mathrm{ReduceSeg}$ does not increase the overall error.

Let $\mathrm{ErrBound}(f_{\tilde{a}},\varepsilon,\beta;q,1)):=\dt(f_a,q)+\frac{\eta(\beta;\dim(U_f),1)}{\varepsilon}$ for $U_f\ni f_a:=\pj_{U_f}(q)$, and $f_{\tilde{a}}$ obtained from privatizing $f_a$ as in Theorem~\ref{thm:gp_funcRm}. Note that $\mathrm{ErrBound}(f_{\tilde{a}};q,\varepsilon,\beta,1))$ provides a high probability bound on $\dt(f_{\tilde{a}},q)$, the $L^2$ error of $f_{\tilde{a}}$ as an approximation for $q$, since
\begin{align*}
\dt(f_{\tilde{a}},q) &\le \dt(f_a,q) + \dt(f_a,f_{\tilde{a}}) \\
&= \dt(f_a,q) + \frac{1}{\varepsilon}\|Z\|\le \mathrm{ErrBound}(f_{\tilde{a}},\varepsilon,\beta;q,1)
\end{align*}
where the latter holds with probability $1-\beta$ for $Z\sim\mathrm{SLap}(\dim(U_f))$. 

\begin{lemma}
\label{lm:util_redseg}
    Let $\hat{S}$ be the set of breakpoints resulting from running $\mathrm{ReduceSeg}(q,t_s,t_e, k_1, k, 1, S, \frac{\beta}{2}, B, \varepsilon')$. Let $\tilde{q}_{\hat{S}}$ be obtained from the $\varepsilon/2$-GP privatization of the best approximation of $q$ on the interval $[t_s, t_e]$ using breakpoints in $\hat{S}$, let $\tilde{q}_S$ be defined similarly for the set $S$. Then with probability $1-\frac{3}{2}\beta$, $\dt(\tilde{q}_{\hat{S}}, q)\le \mathrm{ErrBound}(\tilde{q}_S;q,\varepsilon/2,\beta,1))$.
\end{lemma}

\paragraph{Additional discussion on $\mathrm{ReduceSeg}$.} $\mathrm{ReduceSeg}$ is an optional step; its goal is remove redundant sub-intervals, which may happen when the curve is relatively flat on a sub-segment of the domain, as stated earlier. As such, it will not be useful for curves that do not stay flat for sufficiently long. Example curves where $\mathrm{ReduceSeg}$ is useful and not useful are provided in Fig.~\ref{fig:example_redseg}. By Lemma~\ref{lm:util_redseg} above, however, there is no harm in choosing to run $\mathrm{ReduceSeg}$ (minus the privacy budget it consumes), since it will not remove sub-intervals that help with the approximation. In particular, line $9$ in Algorithm~\ref{alg:algo_redseg} ensures that sub-intervals are only removed when it results in an overall decrease in the error.
\begin{figure}[htbp]
    \centering
    \begin{subfigure}[t]{.23\linewidth}
        \includegraphics[width=\textwidth]{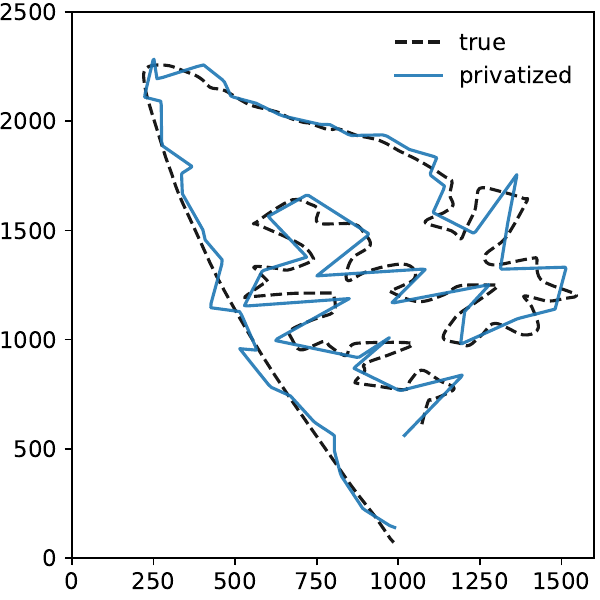}
        \subcaption{\scriptsize{$q_1$ without $\mathrm{ReduceSeg}$,\\${\;\;\;\;\;\;} N=64$}}
        \label{fig:example1_noredseg}
    \end{subfigure}
    \;
    \begin{subfigure}[t]{.23\linewidth}
    \includegraphics[width=\textwidth]{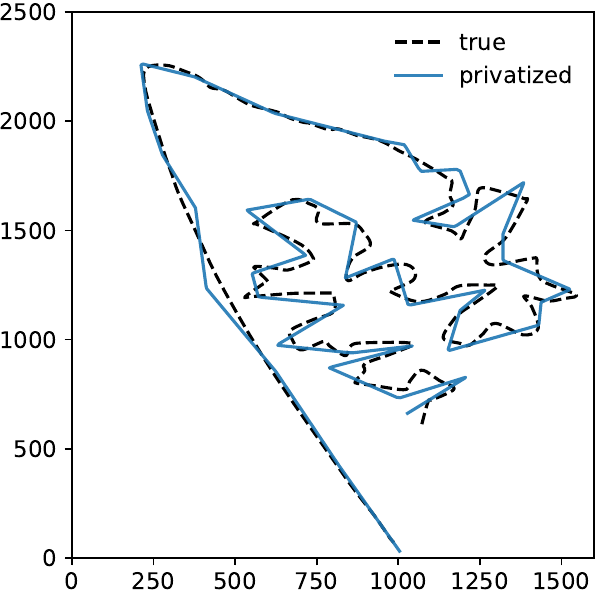}
        \subcaption{\scriptsize{$q_1$ with $\mathrm{ReduceSeg}$,\\ ${\;\;\;\;\;\;}N=42$}}
        \label{fig:example1_redseg}
    \end{subfigure}
    \;\;
    \begin{subfigure}[t]{.23\linewidth}
        \includegraphics[width=\textwidth]{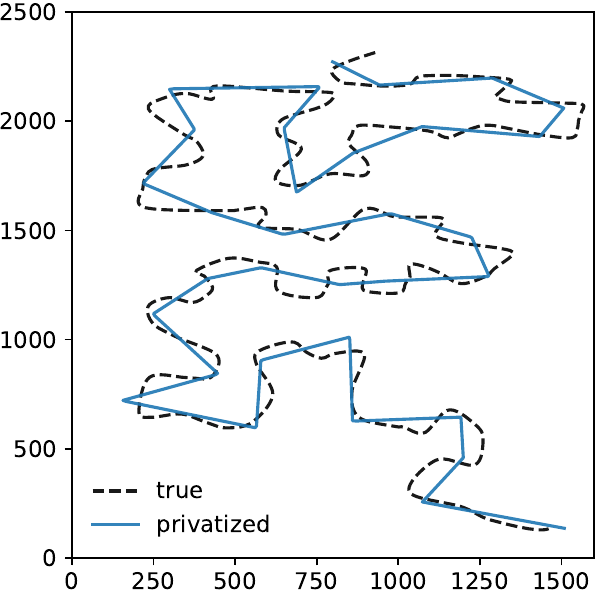}
        \subcaption{\scriptsize{$q_2$ without $\mathrm{ReduceSeg}$, \\ ${\;\;\;\;\;\;}N=32$}}
        \label{fig:example2_noredseg}
    \end{subfigure}
    \;
    \begin{subfigure}[t]{.23\linewidth}
        \includegraphics[width=\textwidth]{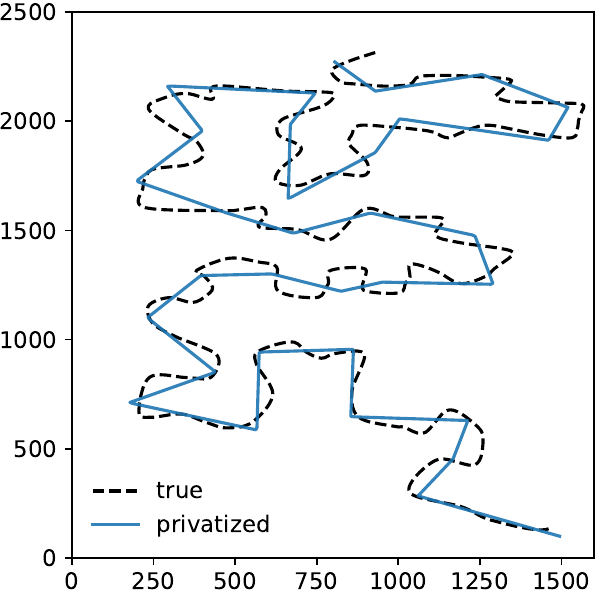}
        \subcaption{\scriptsize{$q_2$ with $\mathrm{ReduceSeg}$, \\ ${\;\;\;\;\;\;}N=32$}}
        \label{fig:example2_redseg}
    \end{subfigure}
    \caption{{Stylized trajectories $q_1,q_2$ privatized without and with $\mathrm{ReduceSeg}$, each using total budget $\varepsilon=1$. Here, $q_1$ benefits from $\mathrm{ReduceSeg}$ while $q_2$ does not. Sub-intervals were reduced on flatter segments of $q_1$, as shown in (\subref{fig:example1_redseg}), resulting in noise reduction.}}
    \label{fig:example_redseg}
\end{figure}

The acute reader may notice that some parameters in Algorithm~\ref{alg:algo_redseg} may be left as hyper-parameters that can be further tuned. Indeed, there might be better choices for the number of partitions on which to call $\mathrm{ReduceSeg}$, as well as the maximum depth $k_1$ of the recursion tree in $\mathrm{ReduceSeg}$. These two parameters might be further tuned, for example, using public data; together they determine the total number of times $\mathrm{ReduceSeg}$ is called. The privacy budget allocated for each such call decreases as the number of calls increases, which makes $\mathrm{ReduceSeg}$ less useful since it will return early at line $9$ without successful reduction (due to noisier estimation of the potential utility gain). On the other hand, a small number of calls to $\mathrm{ReduceSeg}$ directly limits how many sub-intervals can be removed. For simplicity, we set the parameters conservatively to be $4$, knowing that by Lemma~\ref{lm:util_redseg} even sub-optimal parameters do not harm the utility.

\subsection{Extending to unbounded intervals}
\label{sec:curvepriv_unbounded}
In this subsection, we briefly describe how we can extend the previous algorithm for curve privatization on $[0,T]$ to an unbounded interval. This is useful for when the function is defined on $\mathbb{R}$, and we wish to estimate it using functions integrable on finite intervals (e.g. polynomials). 

\paragraph{Using $L^2([0,T])$ for $(C(\mathbb{R}),\dt)$.}
Suppose $U\subseteq C(\mathbb{R})\cap L^2(\mathbb{R})$, so $q\in U$ is both continuous and square-integrable on the entire real line. 
In order to use Algorithm~\ref{alg:algo1} above for such a function, we first find an approximate $\bar{T}$ such that clipping the function to the interval $[-\bar{T},\bar{T}]$ (i.e. setting the function value to zero outside the interval) only introduces a small amount of error, say $O(1/\varepsilon)$. To find such a $\bar{T}$, define a function $g_{T}:q\mapsto \sqrt{\int_{-T}^T q(t)^2 dt} - \sqrt{\int_{-\infty}^{\infty} q(t)^2 dt}$. Note that $g_T$ is $2$-Lipschitz w.r.t. to the $\dist_{2}$ metric. Then we can use SVT on the sequence of queries $g_{T_1},g_{T_2},\dots$, where $T_j=2^j$ for $j\ge 1$, with threshold set to $-1/\varepsilon$. After obtaining $\bar{T}$, we run Algorithm~\ref{alg:algo1} on the function $q\cdot \mathbb{1}_{[-\bar{T},\bar{T}]}$, and return the $\mathbb{R}$-extension of the privatized curve as output.

\section{Experiments}
\label{sec:exp}
In this section, we provide experimental evaluation of our algorithms on several datasets. Here, we present the results for mechanisms in the GP model, deferring results of their CGP variants to Appendix~\ref{sec:app_cgpexp}. 
\paragraph{Baseline methods.} 
There are no prior work for function privatization in the local model, to the best of our knowledge. However, a method for releasing (a subset of) a collection of points under GP/CGP was provided in \cite{liang2023concentrated}; since the privatized points can be connected together to form a piecewise curve, we will use this as our main baseline. The method in \cite{liang2023concentrated} samples $k$ points for some fixed $k>1$, and privatizes each sampled point under GP/CGP using Lemma~\ref{lm:canon_mech}; by applying basic composition to $k$ points, the release of these points together satisfies GP w.r.t. the $\dist_{\infty}$ metric, where $\dist_{\infty}(q,q'):=\sup_{[0,T]} \|q(t)-q'(t)\|$. 
The privacy definitions using the two metrics $\dt$ and $\dist_{\infty}$ are incomparable in terms of strength and purpose, where $\dt$ considers the functional data as a whole, while $\dist_{\infty}$ aims to provide protection for each released point. Nevertheless, the aforementioned method from \cite{liang2023concentrated} serves as a good reference for understanding the problem at hand, and for lacking a better alternative.

In the privatized points, substantial fluctuations might occur between consecutive points due to the large noise injected, especially when the privacy budget is small. To mitigate such effects, we also apply smoothing in the baseline method with sliding windows (labeled as ``Baseline (smoothed)''). We evaluated various combinations of the number of samples $k$, where the points are evenly spaced, and the sliding window size $s$; in each plot, the best combination with the minimum error is shown for each parameter setting. 
\paragraph{Datasets.}
We conduct experiments on three datasets:
\begin{enumerate}
	\item \textbf{ECG records.}
	The PTB-XL Electrocardiography (ECG) dataset \cite{ptbxl2022,ptbxlarticle2020,physionet2000} contains 21799 clinical 12-lead ECGs of 10 seconds length and 100 Hz sample frequency.
	We select the Lead II data of the first 100 records, scale the samples to the unit of microvolt, and perform linear interpolation on them.
	\item \textbf{Taxi trajectory.}
	The CRAWDAD dataset \cite{crawdad2022} contains the mobility traces of approximately 500 cabs in San Francisco Bay Area. The GPS coordinates are converted to points in $\mathbb{R}^2$, with meters as the unit. We perform linear interpolation between consecutive points to convert the time series data into continuous trajectory functions from $[0,T]$ to $\mathbb{R}^2$, for a total of $304$ curves recorded over a $12$-hour window.
    \item \textbf{Synthetic Gaussian mixtures.}
	We randomly generate $50$ synthetic curves, where each instance is a linear combination of $1$ to $5$ random Gaussian functions.
\end{enumerate}

We report both the average $L^2$ error and the $L^2$-squared error (the mean of which is the functional analogue of MSE). The averages are computed over  all simulations of all functions in the respective dataset, where the errors are normalized by the $L^2$ norm of the original input curve.

\subsection{Evaluation on Project-and-Privatize}
Some types of data are known to be approximable by certain finite-dimensional spaces, either due to public available data or expert knowledge. ECG data are one such example, and they can be well approximated by a sufficient number of sinc functions. Here, we show that the simple Project-and-Privatize method (Algorithm~\ref{alg:algo_projpriv}) suffices for privatization. We use the normalized sinc function basis 
$\{\phi_j(t)=\text{sinc}(t-j)\}_{j\in[m]}$ 
defined in Section~\ref{sec:infbasis_finitefunc}, and incorporate the technique discussed there where the basis functions are treated as being defined on $\mathbb{R}$.  
Since a normal QRS length{\footnote{The period of ventricular depolarization, covering the Q, R and S waves.\label{fn:qrs}}} is between $0.08$-$0.10$ seconds (public information), we scale the time domain by 80 times (so $T=800$) to allow the central zero-crossings of a sinc function to approximately capture that length (alternatively, sinc functions with appropriately scaled non-integer zero-crossings can be used). We use $m=800$ sinc functions for our algorithm. The baseline methods use the number of samples $k\in\{100,200,800\}$, and smoothing parameter $s\in \{k/20, k/10\}$. Each method is applied on each record for $30$ simulations. 
\begin{figure}[htbp]
    \centering
    \begin{subfigure}[t]{0.3\linewidth}
        \includegraphics[width = \textwidth]{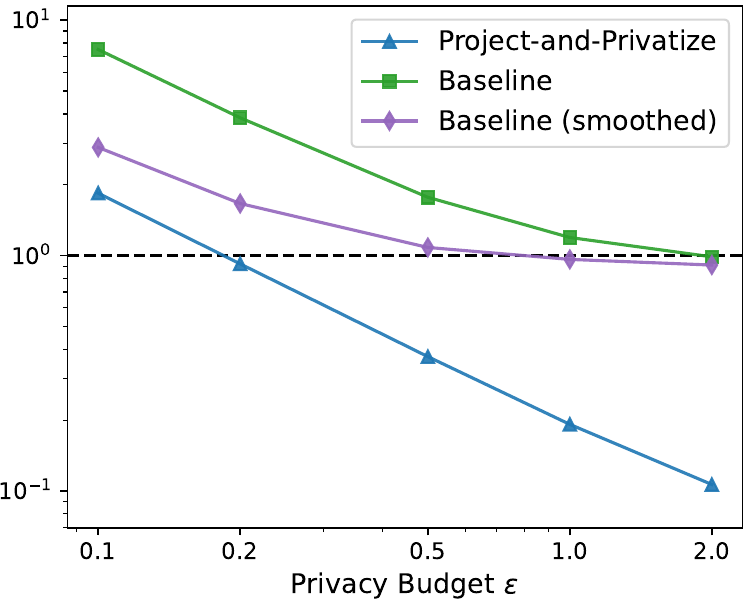}
        \subcaption{$L^2$ error}
        \label{fig:ECG-GP_avg}
    \end{subfigure}
    \;\;\;\;\;
    \begin{subfigure}[t]{0.3\linewidth}
        \includegraphics[width = \textwidth]{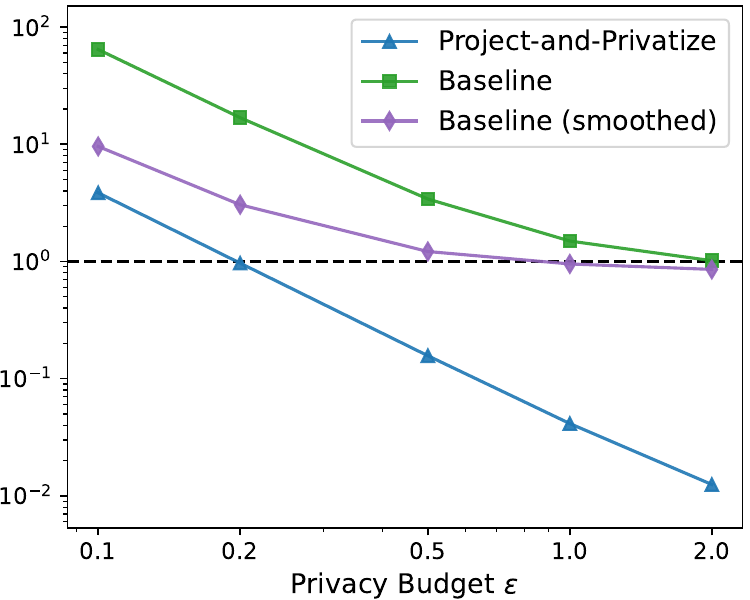} 
        \subcaption{$L^2$-squared error (MSE)}
        \label{fig:ECG-GP_mse}
    \end{subfigure}
    \vspace{-3pt}
    \caption{GP results on ECG dataset}
    \label{fig:ECG-GP}
    \vspace{-3pt}
\end{figure}

\begin{figure}[h]
    \centering
	\includegraphics[width = 0.65\linewidth]{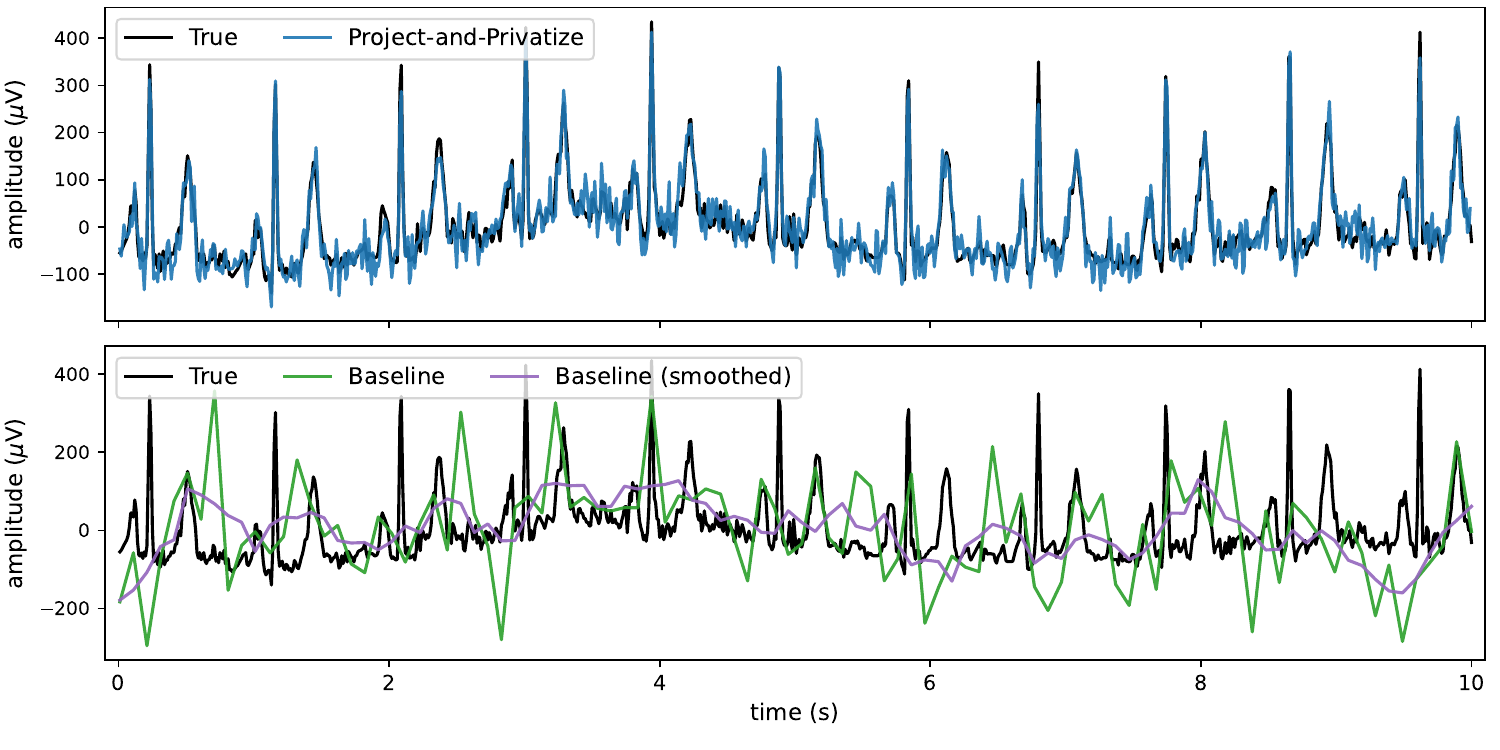}
    \vspace{-5pt}
	\caption{Example of privatizing a ECG record. $\varepsilon=1.0$}
	\label{fig:ECG-Eg}
    \vspace{-5pt}
\end{figure}

In Fig.~\ref{fig:ECG-GP}, we evaluate the performances of the methods at various privacy levels. 
We see significant utility improvement in our mechanisms over the baseline methods, reaching up to an order of magnitude for the average $L^2$ error, and up to 
$2$ orders orders of magnitude in the MSE, for some parameter settings. Also, the utility gap gets larger as the privacy budget increases.

We also present an example curve privatized by GP in Fig.~\ref{fig:ECG-Eg} to see how the methods perform in terms of capturing the signals in the data. The baseline method presented uses $k=100$ samples, with smoothing $s=5$. We see that our mechanism is able to capture all the large and medium-sized peaks, thus retaining the important signals in the data. In contrast, the baseline methods struggle to capture the positions of most of the peaks, as ECG records are rich in information which are difficult to be captured with sampling. Also, even though smoothing helps with the $L^2$ error, it also removes the important signals from the data in this case.

\subsection{Evaluation on privately selecting a basis.}
In this subsection, we evaluate the performance of our adaptive basis selection method (Algorithm~\ref{alg:algo_privsel}) against the baseline methods, on synthetic Gaussian mixtures over the interval $I=[0,100]$. Our mechanism selects a basis from the set of monomials of degree at most $32$, and the baseline methods use the number of samples $k\in\{10,20,50,100\}$ with smoothing parameter $s=k/10$. Each method is applied on each curve for $30$ simulations.
\begin{figure}[h]
    \centering
    \begin{subfigure}[t]{0.3\linewidth}
        \includegraphics[width = \textwidth]{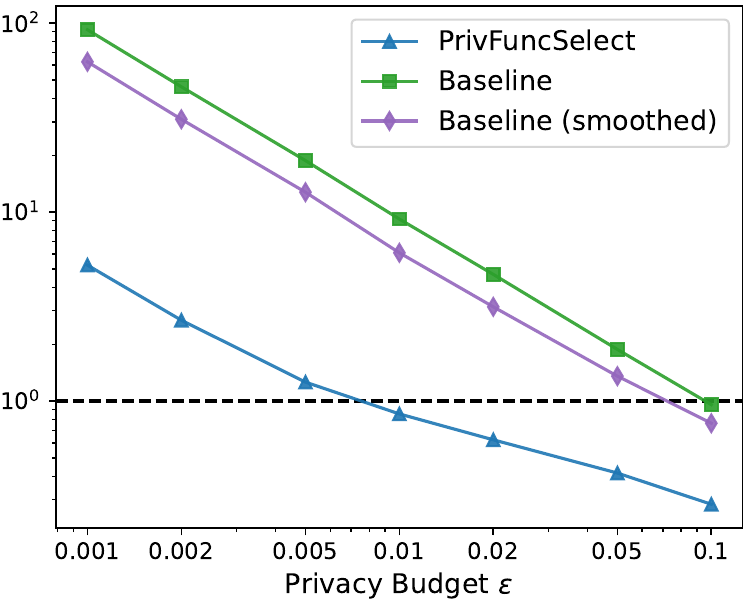} 
        \subcaption{$L^2$ error}
        \label{fig:Synth-Adapt-GP_avg}
    \end{subfigure}
    \;\;\;\;\;
    \begin{subfigure}[t]{0.3\linewidth}
        \includegraphics[width = \textwidth]{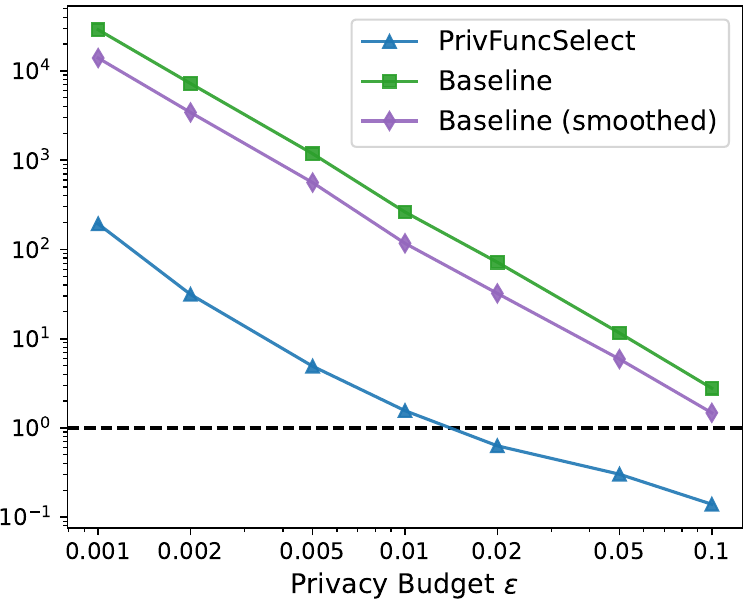}
        \subcaption{$L^2$-squared error (MSE)}
        \label{fig:Synth-Adapt-GP_mse}
    \end{subfigure}
    \vspace{-3pt}
	\caption{GP results on synthetic Gaussian curves}
	\label{fig:Synth-Adapt-GP}
\end{figure}

\vspace{-3pt}
We evaluate the methods at various privacy levels in Fig.~\ref{fig:Synth-Adapt-GP}. 
We see obvious improvement in our mechanisms over the baseline methods, in both the $L^2$ and MSE error, over all parameter settings. Note that Gaussian mixtures are very smooth, so it seems reasonable to use the baseline methods for privatization on these; however, they are still outperformed by our mechanisms.
\begin{figure}[h]
    \centering
	\includegraphics[width = .53\linewidth]{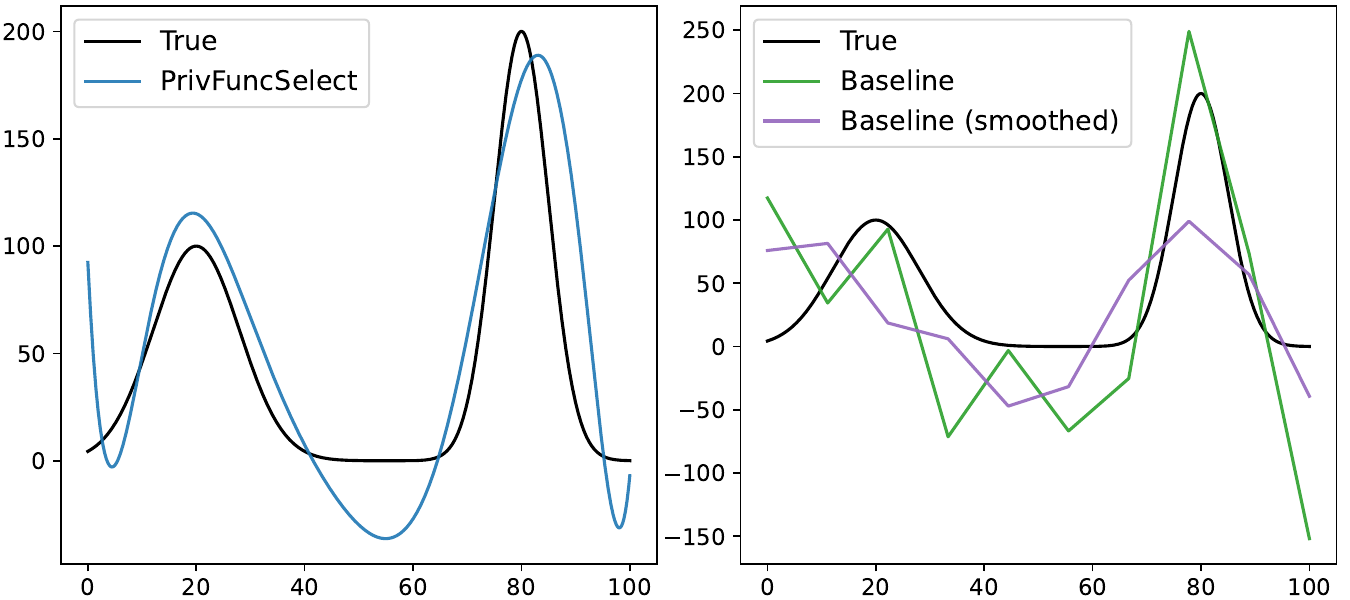}
    \vspace{-3pt}
	\caption{Example of privatizing a synthetic curve. $\varepsilon=0.1$.}
	\label{fig:Synth-Adapt-eg}
\end{figure}

Fig. \ref{fig:Synth-Adapt-eg} shows an example of privatizing a synthetic curve, which is a linear combination of two Gaussian functions. Our mechanism selects 8 monomial basis. The baseline methods uses 10 samples and smoothing parameter $s=3$.

\subsection{Evaluation on curve privatization}
Finally, we evaluate the performance of our curve privatization mechanism (Algorithm~\ref{alg:algo1}) against the baseline methods, on the taxi trajectory dataset and on synthetic Gaussian curves.  
Each method is applied on each curve for $30$ simulations.
For a trajectory curve with $N$ points, the number of samples $k\in\{N/10, N/5\}$ and smoothing parameter $s\in \{k/20, k/10\}$ are used in the baseline methods. We also apply the post-processing program in Section~\ref{sec:ensure_cont} to obtain continuous versions of our privatized curves, and they are labeled as ``PrivFuncSeg (continuous)''.

\paragraph{Results on trajectory dataset.} 
\begin{figure}[htbp]
    \centering
    \begin{subfigure}[t]{0.3\linewidth}
        \includegraphics[width = \textwidth]{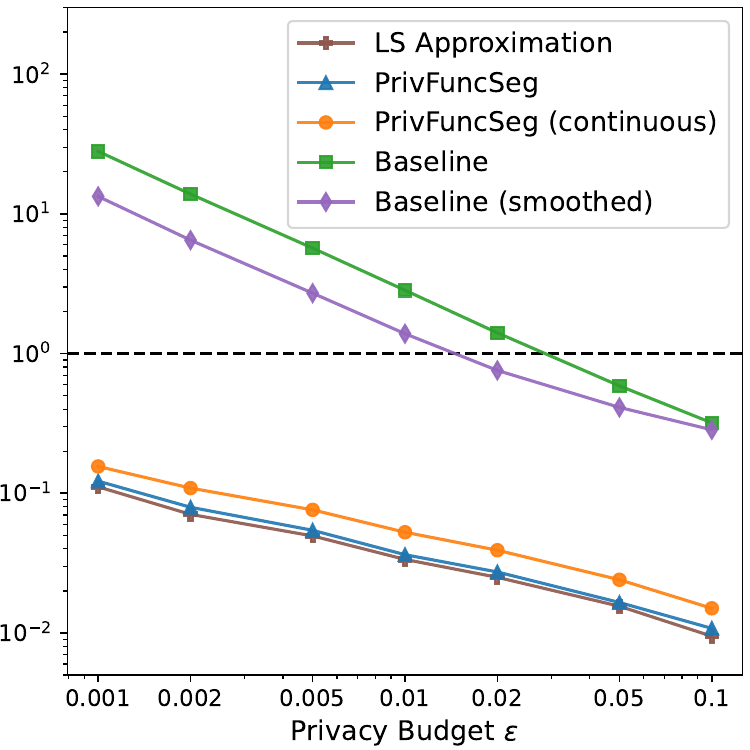}
        \subcaption{$L^2$ error}
        \label{fig:Taxi-GP_avg}
    \end{subfigure}
    \;\;\;\;\;
    \begin{subfigure}[t]{0.3\linewidth}
        \includegraphics[width = \textwidth]{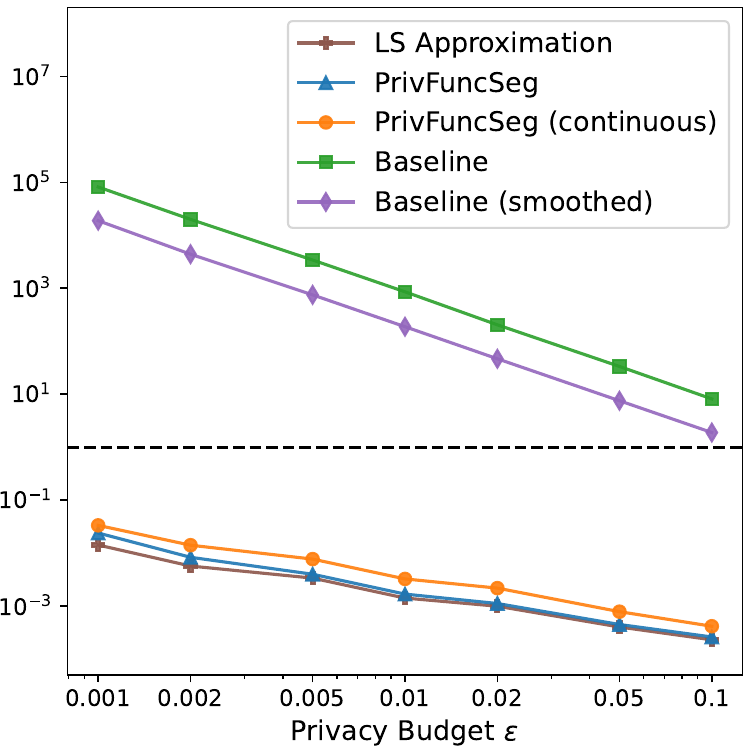}
        \subcaption{$L^2$-squared error (MSE)}
        \label{fig:Taxi-GP_mse}
    \end{subfigure}
    \vspace{-3pt}
	\caption{GP results on the taxi trajectory dataset}
	\label{fig:Taxi-GP}
\end{figure}
We apply our $\mathrm{PrivFuncSeg}$ algorithm with the simple linear basis given in Example~\ref{exp:2d_linear}. 
The results for various privacy levels are reported in Fig.~\ref{fig:Taxi-GP}. We see a significant advantage in our mechanisms over the baseline methods, where our mechanisms offer a $1$ to $2$ orders of magnitude improvement in the $L^2$ error, over all parameter settings. The improvement reaches up to $3$ orders of magnitude in the MSE. 
To better understand the advantage of our mechanisms, we also plotted the errors in the (partially private) least-squares {projection} (labeled as ``LS Approximation'') obtained from our procedure to find the best set of break points for piecewise privatization (i.e. 
$\mathrm{PrivFuncSeg}$ and $\mathrm{ReduceSeg}$ (Algorithms~\ref{alg:algo1} and ~\ref{alg:algo_redseg}) without the final privatization step). 
An error in $\mathrm{PrivFuncSeg}$ that is close to that in 
$\mathrm{LS\;Approximation}$ indicates that the mechanism is successful in balancing the error due to approximation and that due to privatization (noise).

Next, we present an example trajectory curve privatized by the methods in Fig.~\ref{fig:Taxi-Eg}, to examine how well the mechanisms are able to provide outputs with a general likeness to the original curve. The baseline method here uses $k=10$ samples, with smoothing parameter $s=3$.
\begin{figure}[h]
    \centering
	\includegraphics[width = .53\linewidth]{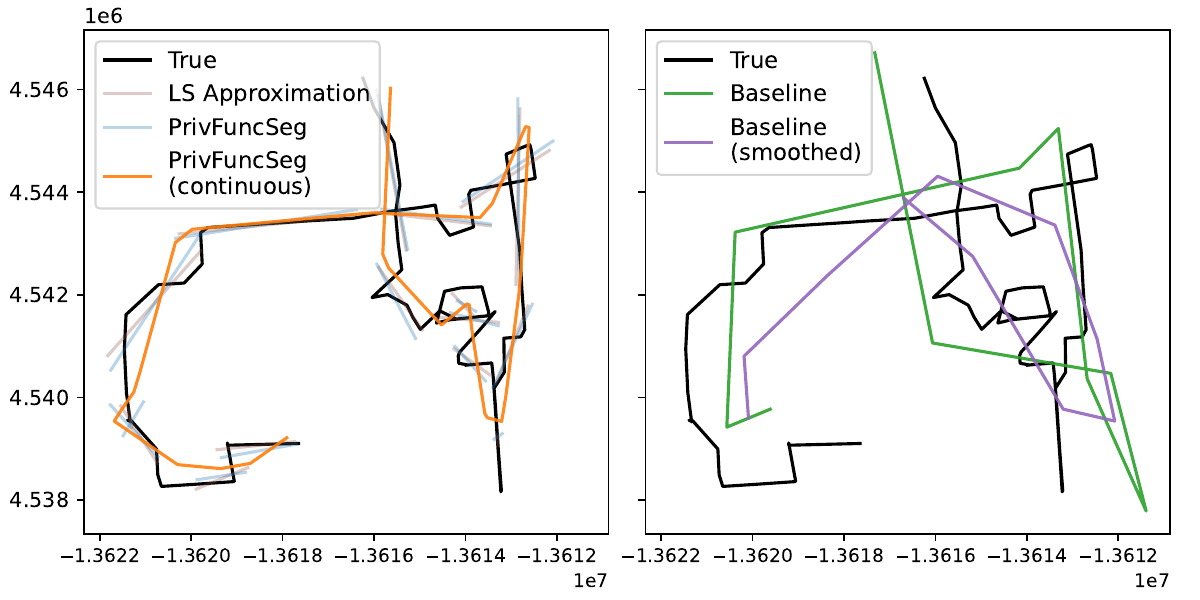}
    \vspace{-3pt}
	\caption{Example of privatizing a taxi trajectory. $\varepsilon=0.01$.}
	\label{fig:Taxi-Eg}
\end{figure}
We see that our mechanism offers a privatized trajectory that retains the overall shape and major directional turns, in both the continuous or piecewise continuous curve. In contrast, the curves provided by the baseline method (smoothed or not), are only able to retain some likeness in terms of shape. The locations of the sampled points and the noise added to the them can cause the baseline methods to miss meaningful segments of the trajectory, as demonstrated here.

\begin{figure}[h]
    \vspace{5pt}
    \centering
	\includegraphics[width = 0.85\linewidth]{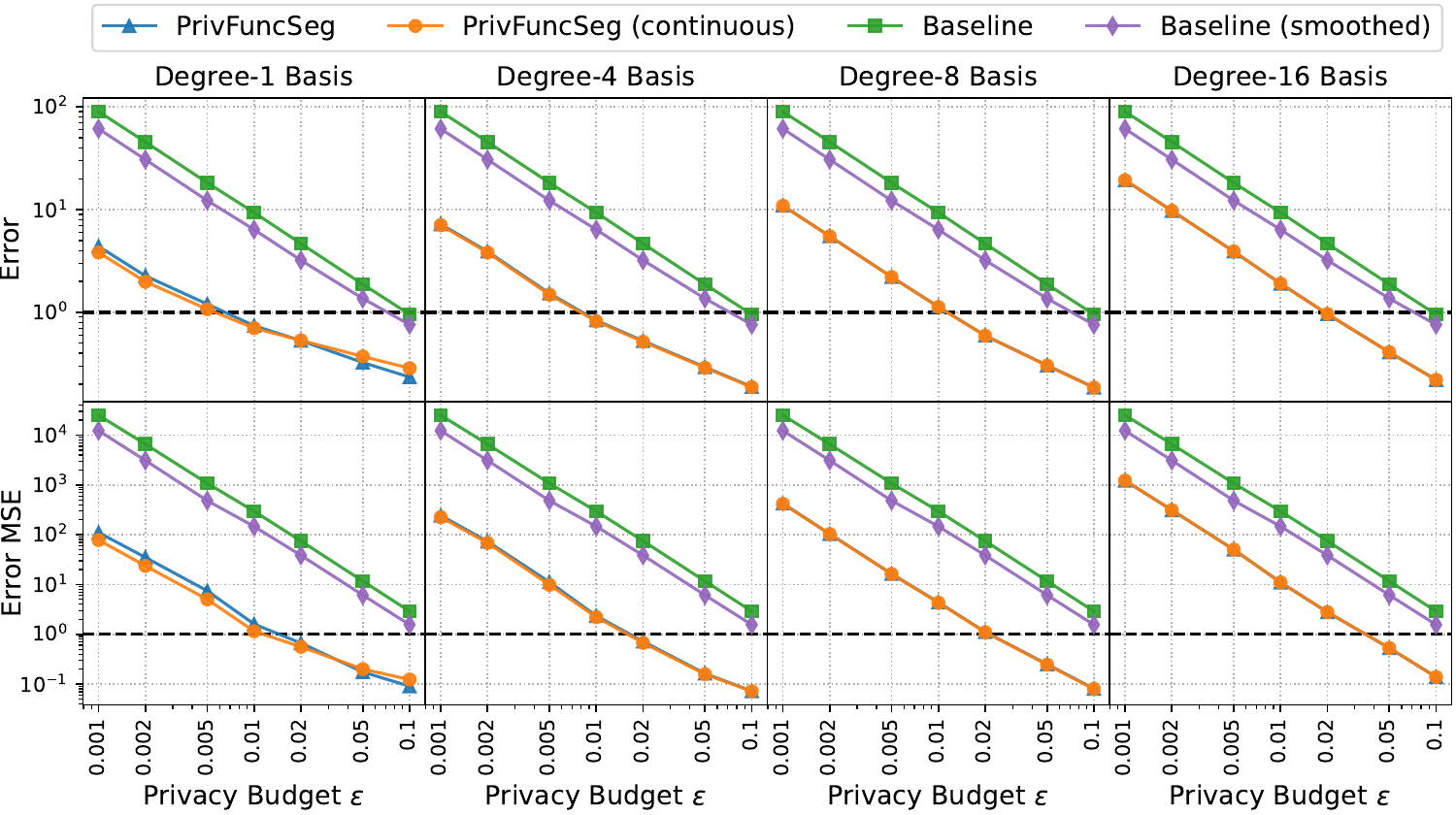}
    \vspace{-3pt}
	\caption{GP results on synthetic Gaussian curves}
	\label{fig:Synth-GP}
\end{figure}
\paragraph{Results on synthetic Gaussian curves.}
We apply $\mathrm{PrivFunSeg}$ to synthetic Gaussian curves with the monomial basis $\{t^{m-j}\}_{j\in [m]}$ for $m\in \{2,5,9,17\}$, corresponding to polynomials of degrees upto $1, 4, 8, 16$, respectively. 
The errors at various privacy levels are shown in Fig.~\ref{fig:Synth-GP}. Our mechanisms outperformed the baseline methods at all parameter settings, for all selected bases, in both the $L^2$ error and MSE. Note that for small privacy parameters, it is preferable to use a smaller basis, since the privatization error (due to noise) dominates. For a sufficiently large privacy budget, it becomes preferable to use a larger basis to improve the approximation error. 

\section{Concluding Remarks}
\label{sec:conclude}
We have provided a framework for function privatization in the local model of geo-privacy, offering algorithms tailored to different types of functions. 
Our experimental evaluation based on multiple datasets and various parameter settings demonstrates the advantage of our methods, especially for preserving the features and signals in the functions.

The privacy guarantee adopted in this work is based on similarity of functions measured by the $L^2$ metric. 
Depending on the application, in some cases another metric may be preferred, where the ideas developed in this framework might be useful. In particular, our framework readily generalizes to any metric induced by an inner product, as remarked in Section~\ref{sec:priv_func}. 
For trajectories, other popular choices for measuring distance include the Hausdrorff distance or
Fr\'echet distance, which are interesting directions for future work.

\section*{Acknowledgements}
This work is supported in part by HKRGC under grants 16205422, 16204223, and 16203924. We thank the anonymous reviewers of CCS '26 for their helpful feedback and valuable comments.

\bibliographystyle{alpha}
\bibliography{ccs_full}

\appendix

\section{Ethical Considerations and Open Science}
\label{app:open_science}
In this paper, we considered privatizing functions under the local model. Compared to the central model, the local model allows stronger privacy protection; i.e. it benefits individual data owners.
We have performed experimental evaluation of our methods using synthetic data, and public available datasets: taxi trajectories \cite{crawdad2022} and ECG records \cite{ptbxl2022,ptbxlarticle2020,physionet2000}. The public datasets had been previously sanitized.

Our approach considers an input function to be privatized as a whole, and over the entire domain of the input function. This has the benefit that it allows a general likeness of the input function to be retained, thus allowing further meaningful analytics on the features/signals contained in the function. I.e. researchers can benefit from such improved utility. 

The code to all experiments can be found at: \href{https://github.com/hkustDB/FuncPrivGP}{https://github.com/hkustDB/FuncPrivGP}, where we also provide 
detailed information on the datasets, implementations and running the experiments.

\section{Useful facts and inequalities}
\label{sec:app_facts}
\begin{fact}
\label{par:point_example}
Privatizing point set via Lemma~\ref{lm:canon_mech} is not GP w.r.t $\dt$.
\end{fact}
\begin{proof}
We give an example to show that privatizing and releasing a fixed set of points does not suffice to guarantee GP w.r.t. $\dt$. Consider again the example in Section~\ref{sec:priv_lin}, where we have two linear functions $f_a$ and $f_a'$ on $[0,1]$, with $a=(2,0.5)$ and $a'=(0.5,2)$. Suppose we privatize the points at $t_1=0, t_2=0.5, t_3=1$, with $\varepsilon=1$, which for $f_a$ corresponds to a pdf at $z=(z_1,z_2,z_3)$ that is $\propto e^{-\sqrt{(z_1-f_a(t_1))^2+(z_2-f_a(t_2))^2+(z_3-f_a(3))^2}}$. Then the ratio of the respective pdf's at $z=(0.5,1.0,2.0)$ is
\begin{align*}
\frac{e^{-\sqrt{(z_1-f_a(t_1))^2+(z_2-f_a(t_2))^2+(z_3-f_a(3))^2}}}{e^{-\sqrt{(z_1-f_a(t_1))^2+(z_2-f_a(t_2))^2+(z_3-f_a(3))^2}}}&=e^{-\sqrt{0.5}+\sqrt{1.5^2+1.25^2+0.5^2}} >e^{1.5/\sqrt{3}}=e^{\dt(f_a,f_{a'})},
\end{align*}
thus violating the GP requirement.
\end{proof}

\begin{theorem} [Hilbert projection theorem]
        Let $C\subseteq H$ be any non-empty closed convex subset of a Hilbert space ${H}$ whose inner product induces norm $\|\cdot\|_{H}$. Then for every $y\in {H}$, there exists a unique $c(y)\in C$ such that $\|y-c(y)\|_{H}=\inf_{z\in C}\|y-z\|_H$. Moreover, if $C$ is also a subspace of $H$, then $c(y)\in C$ is such that $y-c(y)$ is orthogonal to $C$.
\end{theorem}

\begin{fact}
\label{clm:xt_exp_xe_inq}
    Fix any $x, t>0$. Then $\left(\frac{x}{t}\right)^t\le e^{x/e}$.
\end{fact}
\begin{proof}
    The function $z\mapsto e^{z-1}-z$ is convex and has minimum value $0$, i.e. $e^{z-1}-z\ge 0 \iff e^z \ge ez \implies z\ge 1+\ln z $ for $z > 0$. Choose $z=\frac{x}{et}>0$, then 
    \begin{align*}
      \ln\left(\frac{x}{t}\right)=  1+\ln\left(\frac{x}{et}\right) &\le \frac{x}{et}\\
t \ln\left(\frac{x}{t}\right) &\le \frac{x}{e}\\
 \left(\frac{x}{t}\right)^{t} &\le e^{x/e}.
    \end{align*}
\end{proof}

\begin{lemma} [\cite{sandor2006certain}]
\label{lm:sandor_gamma}
    For any $x > 0$,
    \[
    \sqrt{x} \le \frac{\Gamma(x+1)}{\Gamma(x+\frac{1}{2})} \le \sqrt{x+\frac{1}{2}}.
    \]
\end{lemma}

\begin{lemma} [\cite{natalini2000inequalities}]
\label{lm:nat_upgamma_ineq}
 For any $c, t> 1$ and $x$ where $x> \frac{c}{c-1}(x-1)$,
 \[\Gamma(t,x)< c x^{t-1}e^{-x}.\]
\end{lemma}

\begin{lemma} [Generalized Gamma Distribution \cite{stacy1962generalization, stacy1965parameter}] 
\label{lm:gen_gamma}
Let $\mathcal{G}(\lambda,d,p)$ denote the generalized gamma distribution with parameters $\lambda, d, p > 0$, whose pdf is given by $g(r)=\frac{p}{\lambda^d \Gamma(d/p)}r^{d-1}e^{-(r/\lambda)^p}$. Then for $G\sim \mathcal{G}(\lambda,d,p)$, 
\[
\Pr[G\le r]=\frac{\gamma(d/p,(r/\lambda)^p)}{\Gamma(d/p)} \;\;\text{and} \;\;\;\mathbb{E}[G^k]=\lambda^k\frac{\Gamma(d/p+k/p)}{\Gamma(d/p)} , k\ge 1.
\]
\end{lemma}

Let $Z=[z_1,\dotsb,z_m]$ be a random vector. If $Z$ is drawn from a spherical Laplace distribution, then $\|Z\|\sim \mathcal{G}(1,m,1)$; if $Z\sim\mathcal{N}(0,I_{m\times m})$, then $\|Z\|\sim \mathcal{G}(\sqrt{2},m,2)$ (see \cite{liang2023concentrated}, Appendix $A$). Thus, using Lemmas \ref{lm:sandor_gamma} and \ref{lm:gen_gamma} above, we have:
\begin{corollary}
\label{cor:expZ_slap_gauss}
    Let $Z=[z_1,\dotsb,z_m]$ be a random vector. Then
    \begin{enumerate}
        \item $\mathbb{E}[\|Z\|] = \frac{\Gamma(m+1)}{\Gamma(m)}=m$ and $\mathbb{E}[\|Z\|^2]=\frac{\Gamma(m+2)}{\Gamma(m)}=m(m+1)$, if $Z\sim \mathrm{SLap}(m)$;
        \item $\mathbb{E}[\|Z\|]=\sqrt{2}\frac{\Gamma(m/2+1/2)}{\Gamma(m/2)}\in \left[\sqrt{m-1},\sqrt{m}\right]$ and $\mathbb{E}[\|Z\|^2]=2\frac{\Gamma(m/2+1)}{\Gamma(m/2)}=m$, if $Z\sim\mathcal{N}(0,I_{m\times m})$. 
    \end{enumerate}
\end{corollary}

We can derive high probability bounds for the magnitude $\|Z\|$:
\begin{lemma}
\label{lm:slap_highprob}
    Fix $1 < m\in\mathbb{N}$, let $R\sim \mathcal{G}(1,m,1)$. Then for any $1>\beta>0$, $R\le \frac{e}{e-1}(m+\ln(1/\beta))$ with probability at least $1-\beta$.
\end{lemma}
\begin{proof}
    Let $r:=\frac{e}{e-1}(m+\ln(1/\beta))$. We will show $\Pr[R\le r] \ge 1-\beta$. We have $\Gamma(m)=\gamma(m,r)+\Gamma(m,r)$, where $\Gamma(m,r)=\int_{r}^{\infty} t^{m-1}e^{-t} dt$ is the upper incomplete gamma function. By Lemma~\ref{lm:gen_gamma}, $\Pr[R\le r] = \frac{\gamma(m,r)}{\Gamma(m)}=\frac{\Gamma(m)-\Gamma(m,r)}{\Gamma(m)}=1-\frac{\Gamma(m,r)}{\Gamma(m)}$, so it suffices to show $ \frac{\Gamma(m,r)}{\Gamma(m)}\le\beta$. 
    To this end, we use the inequality from Lemma~\ref{lm:nat_upgamma_ineq}: $\Gamma(m,r)< c r^{m-1}e^{-r}$ for $c>1$, and $r>\frac{c}{c-1}(m-1)$. Choose $c=e$, then $r=\frac{e}{e-1}(m+\ln(1/\beta)) > \frac{c}{c-1}(m-1)$, so
    \begin{align*}
    \frac{\Gamma(m,r)}{\Gamma(m)} &< \frac{e r^{m-1}e^{-r}}{\Gamma(m)} = \frac{e r^{m-1}e^{-r}}{(m-1)!} \\
    &\le \frac{e r^{m-1} e^{-r}}{\sqrt{2\pi(m-1)}(m-1)^{m-1}e^{-(m-1)}e^{\frac{1}{12(m-1)+1}}}\;\;\;\;\;\text{(Stirling's)}\\
    &\le \frac{1}{\sqrt{2\pi(m-1)}}\cdot \underbrace{\left(\frac{r}{m-1}\right)^{m-1}}_{\le e^{r/e}\text{, Fact~\ref{clm:xt_exp_xe_inq}}.} e^{-(r-1)+(m-1)}\\
    &\le \frac{e^{-r(1-\frac{1}{e})+m}}{\sqrt{2\pi(m-1)}} 
    = \frac{e^{-\frac{e}{e-1}(m+\ln\frac{1}{\beta})\frac{e-1}{e}+m}}{\sqrt{2\pi(m-1)}} \\
    &= \frac{e^{-(m+\ln\frac{1}{\beta})-m}}{\sqrt{2\pi(m-1)}}
    = \frac{\beta}{\sqrt{2\pi(m-1)}}\le \beta.
    \end{align*}
\end{proof}

\begin{lemma}[\cite{laurent2000adaptive}]
\label{lm:gauss_mag_laurent}
Given $Z\sim\mathcal{N}\left(0,I_{m\times m}\right)$, then with probability at least $1-\beta$, 
\[\|Z\| \le  \sqrt{m+2\sqrt{m\ln(1/\beta)}+2\ln(1/\beta)}.\]
\end{lemma}

\begin{fact}
   Fix $k>0$. Let $\mathrm{sinc}_k:t\mapsto \frac{\sin(k\pi t)}{k\pi t}$. Then 
    \begin{enumerate}
        \item $\int_{-\infty}^{\infty} \mathrm{sinc}_k(t-a)\cdot \mathrm{sinc}_k(t-b) dt = \mathrm{sinc}_k(a-b)/k$
        \item $\int_{-\infty}^{\infty} \mathrm{sinc}_k(t-a)^2 dt = 1/k$
    \end{enumerate}
\end{fact}

\begin{fact}
\label{ft:integ_linear_mat}
 For $u=[u_1,\dotsb,u_m]^T, v=[v_1,\dotsb,v_m]^T\in \mathbb{R}^m$, and $S(t) = \begin{bmatrix}
    S_{1,1}(t) & \dotsb & S_{1,m}(t)\\
    \vdots & &\vdots\\
    S_{m,1}(t) & \dotsb & S_{m,m}(t)
\end{bmatrix}$ where $S_{j,l}(t)=S_{l,j}(t)\in L^2(I)$ for all $j,l\in [m]$, we have
\begin{align*}
&{}\int_I \begin{bmatrix}
    u_1 & \dotsb & u_m
\end{bmatrix} S(t) \begin{bmatrix}
    v_1\\
    \vdots\\
    v_m
\end{bmatrix} dt 
= \begin{bmatrix}
    u_1 & \dotsb & u_m
\end{bmatrix} \begin{bmatrix}
    \int_{t \in I} S_{1,1}(t)dt & \dotsb & \int_{t \in I} S_{1,m}(t) dt\\
    \vdots & &\vdots\\
    \int_{t \in I} S_{m,1}(t) dt & \dotsb & \int_{t \in I} S_{m,m}(t) dt
\end{bmatrix}\begin{bmatrix}
    v_1\\
    \vdots\\
    v_m
\end{bmatrix}.
\end{align*} 
\end{fact}

\section{Privatization under CGP}
\label{app:priv_CGP}
In this section, we provide the CGP variants of the algorithms provided above. For the steps that involve SVT, which is a GP mechanism, we use the following relationship to convert the privacy parameters from GP to CGP.
\begin{lemma} [GP conversion to CGP \cite{liang2023concentrated}]
\label{lm:gp_to_cgp}
    Any $\varepsilon$-GP mechanism is also $\frac{\varepsilon^2}{2}$-CGP.
\end{lemma}

\subsection{Function Privatization}
The CGP variant of $\mathrm{Project}$-$\mathrm{and}$-$\mathrm{Privatize}$ is provided in Algorithm~\ref{alg:algo_projprivcgp}; its expected squared error is given by 
Lemma~\ref{lm:finapprox_priv_distsquared}, where the noise vector for CGP has magnitude proportional to $\sqrt{\dim(U_f)}$. The latter is reflected in CGP variant of $\mathrm{PrivFuncSelect}$ (Algorithm~\ref{alg:algo_privselcgp}), in the step where we try to balance the errors due to approximation and noise addition (line $5$). We use Lemma~\ref{lm:gp_to_cgp} to convert the GP privacy parameter used in the SVT call (line $1$) to a CGP privacy parameter, then similar to the GP case, the privacy of Algorithm~\ref{alg:algo_privselcgp} follows from basic composition.

\begin{algorithm}
\caption{Project-and-Privatize-CGP}
    \label{alg:algo_projprivcgp}
    \vspace{-5pt}
    \begin{flushleft}
    \textbf{Input}: $q:I\rightarrow \mathbb{R}^n$; finite-dimensional function space $U_f$; $\rho >0$\\
    \textbf{Output}: privatized $\tilde{q}$
    \end{flushleft}
      \vspace{-5pt}
    \begin{algorithmic}[1]
    \STATE compute $\Sigma$ from $U_f$
    \STATE $a\gets $ coefficients of $\pj_{U_f}(q)$
    \STATE ${\tilde{a}}\gets a+\frac{1}{\sqrt{2\rho}}\Sigma^{1/2} Z$, $Z\sim\mathcal{N}(0,I_{\dim(U_f)})$
    \RETURN $f_{\tilde{a}}$
    \end{algorithmic}
\end{algorithm}

\begin{algorithm}
\caption{PrivFuncSelectCGP}
    \label{alg:algo_privselcgp}
      \vspace{-5pt}
    \begin{flushleft}
    \textbf{Input}: $q:I\rightarrow \mathbb{R}^n; \{\phi_j\}_{j\in J};\beta > 0; \rho >0$\\
    \textbf{Output}: privatized $\tilde{q}$
    \end{flushleft}
      \vspace{-5pt}
    \begin{algorithmic}[1]
    \STATE $k_0 \gets \mathrm{SVT}(q,(\varepsilon_0/3,2\varepsilon_0/3),-\frac{1}{\varepsilon_0},1,g_1,g_2,\dotsb)$, where $\varepsilon_0=\frac{\sqrt{\rho}}{{\sqrt{2}}}$, $g_j(\cdot):=-\dist_2(\pj_{U_f^{(j)}}(\cdot), \cdot)$
    \STATE $\tilde{q}_0\gets \mathrm{Project}$-$\mathrm{and}$-$\mathrm{Privatize}$-$\mathrm{CGP}(q,U_f^{(k_0)},\rho/4)$
    \STATE $\tilde{c}\gets$ coefficients of $\tilde{q}_0$ in $U_f^{(k_0)}$
    \STATE $s^{(r)} \gets r$th largest magnitude in $\tilde{c}$, for $r\in [k_0]$
        \STATE ${k_1}\gets \mathrm{SVT}(q,(\frac{\varepsilon_0}{3},\frac{2\varepsilon_0}{3}),0,1,g'_1,g'_2,\dotsb)$, $N_r:=\{j: \|\tilde{c}_j\| \ge s^{(r)}\}$, $g'_r(\cdot):=\frac{\sqrt{n|N_r|}}{\varepsilon_0}-\dist_2(f_{b^{(r)}}, \cdot)$, $b_j^{(r)}:=\tilde{c}_j\mathbb{1}\{j\in N_r\}$
    \STATE $U_f \gets \mathrm{span}(\{\phi_j: j\in N_{k_1}\})$
    \STATE $f_{\tilde{a}}\gets \mathrm{Project}$-$\mathrm{and}$-$\mathrm{Privatize}$-$\mathrm{CGP}$$(q,U_f,\rho/4)$
    \RETURN $f_{\tilde{a}}$
    \end{algorithmic}
\end{algorithm}

\subsection{Curve Privatization}
\label{sec:arbitrary_curve_cgp}
The CGP variant of $\mathrm{PrivFuncSeg}$ for curve privatization is provided in Algorithm~\ref{alg:algo1_cgp}. The privacy guarantee follows similarly to the case of GP. The steps that use the magnitude related to noise addition reflect the $\sqrt{(\cdot)}$ relationship with that of GP: in line $1$ of Algorithm~\ref{alg:algo1_cgp}, and line $9$ of Algorithm~\ref{alg:algo_redseg_cgp}.

\begin{algorithm}[H]
\caption{PrivFuncSegCGP}
    \label{alg:algo1_cgp}
    \vspace{-5pt}
    \begin{flushleft}
    \textbf{Input}: $q:[0,T]\rightarrow \mathbb{R}^n; I=[0,T]; \{\phi_j\}_{j\in [m]};\beta > 0; \rho >0$\\
    \textbf{Output}: privatized $\tilde{q}$
    \end{flushleft}
    \vspace{-5pt}
    \begin{algorithmic}[1]
    \STATE $\bar{k} \gets \mathrm{SVT}(q,(\varepsilon_0/3,2\varepsilon_0/3),0,1,g_0,g_1,g_2,\dotsb)$, where $\varepsilon_0=\frac{\sqrt{\rho}}{\sqrt{2}}$, $\tau_j=\frac{\sqrt{2^jmn}}{\sqrt{\rho/2}}$, $g_j(q):=\tau_j-\dist_2(\pj_{U_f^{(j)}}(q), q)$, $\pj_{U_f^{(j)}}$ computes best approximation using $2^{j}$ sub-intervals
    \STATE $B\gets 3\rho/4$
    \STATE $S\gets \{j{T}/{2^{\bar{k}}}: j=0,1,2,\dotsb,2^{\bar{k}}\}$
    \STATE $k\gets \min(\bar{k}-2,4)$
    \IF{$k_1\ge 1$}
    \FOR{$i=1,...,4$}
    \STATE $\hat{S}_i\gets \mathrm{ReduceSegCGP}(q,{(i-1)T}/{4}, {iT}/{4},k_1,1,S_i,\beta,B,{\frac{\rho}{16}})$, $S_i:= S\cap [{(i-1)T}/{4},{i T}/{4}]$
    \ENDFOR
    \ENDIF
    \STATE $\hat{S}\gets \cup_{i} \hat{S}_i=\{0=T_0,T_1,\dotsb,T_{N}=T\}$;\;\; $N\gets |\hat{S}|-1$
    \STATE $U_f\gets \mathrm{span}(\{\phi_j\cdot \mathbb{1}_{[T_{s-1},T_s)}\}_{s\in [N], j\in [m]})$
    \STATE $f_{\tilde{a}}\gets \mathrm{Project}$-$\mathrm{and}$-$\mathrm{Privatize}$-$\mathrm{CGP}(q,U_f,B)$
    \RETURN $f_{\tilde{a}}$
    \end{algorithmic}
\end{algorithm}
\begin{algorithm}[H]
\caption{ReduceSegCGP}
    \label{alg:algo_redseg_cgp}
    \vspace{-5pt}
    \begin{flushleft}
    \textbf{Input}: $q:[0,T]\rightarrow \mathbb{R}^n; t_{s}<t_{e}; k_1; l;$ set $S$ containing breakpoints to be checked; $\beta > 0; B\ge \rho' >0$\\
    \textbf{Output}: $\hat{S}$ after reducing the number of sub-intervals
    \end{flushleft}
    \vspace{-5pt}
    \begin{algorithmic}[1]
    \IF{$l > k_1$}
    \RETURN ${S}$
    \ENDIF
    \STATE $k\gets \log_2(|S|-1)$; 
    \STATE $S'\gets \{2r\frac{t_e-t_s}{2^k}+t_s: r=0,\dotsb,2^{k-1})\}$\;\;\; \COMMENT{try to reduce number of sub-intervals by half}
    \STATE $f_{[t_s,t_e]} \gets$ best approx. of $q$ on $[t_s,t_e]$ using break points in $S'$
    \STATE $\mathrm{err}\gets \dist_2(f_{[t_s,t_e]},q\cdot \mathbb{1}_{[t_s,t_e]} )+\frac{\sqrt{2^l}}{\sqrt{2\rho'}}Z$, $Z\sim\mathcal{N}(0,1)$
    \STATE $B\gets B-\frac{\rho'}{2^l}$
    \IF{$\mathrm{err}+\frac{\sqrt{2^l}}{\sqrt{2\rho'}}\sqrt{2k_1\ln(2)+2\ln(1/\beta)}\le \frac{\sqrt{2}-1}{\sqrt{2}}\cdot \frac{\sqrt{2^{k-1}mn}}{\sqrt{2B}}$} 
    \STATE $S'_L\gets S'\cap [t_s,(t_s+t_e)/2]$,\; $S'_R\gets S'\cap [(t_s+t_e)/2,t_e]$
    \STATE $\hat{S}_L\gets \mathrm{ReduceSegCGP}(q, t_s, (t_s+t_e)/2, k_1, l+1, S'_L,\beta,{B},\frac{\rho'}{2})$
    \STATE $\hat{S}_R\gets \mathrm{ReduceSegCGP}(q, (t_s+t_e)/2, t_e, k_1, l+1,S'_R,\beta,{B},\frac{\rho'}{2})$
    \STATE $\hat{S} \gets \hat{S}_L \cup \hat{S}_R$
    \ELSE
    \STATE $\hat{S}\gets S$
    \ENDIF
    \RETURN $\hat{S}$
    \end{algorithmic}
\end{algorithm}

Let $\mathrm{ErrBound}(f_{\tilde{a}},\rho,\beta;q,2)):=\dt(f_a,q)+\frac{\eta(\beta;\dim(U_f),2)}{\sqrt{2\rho}}$ for $f_{\tilde{a}}\in U_f$ obtained from privatizing $f_a$ as in Theorem~\ref{thm:cgp_funcRm}. Similar to the GP case, $\mathrm{ErrBound}(f_{\tilde{a}},\rho,\beta;q,2))$ bounds $\dt(f_{\tilde{a}},q)$ with probability $1-\beta$.
We also have a utility guarantee similar to the GP case for Algorithm~\ref{alg:algo_redseg_cgp}. 
\begin{lemma}
\label{lm:util_redseg_cgp}
    Let $\hat{S}$ be the set of breakpoints resulting from running $\mathrm{ReduceSegCGP}(q,t_s,t_e, k_1, k, 1, S, \frac{\beta}{2}, B, \rho')$. Let $\tilde{q}_{\hat{S}}$ be obtained from the $\rho/2$-CGP privatization of the best approximation of $q$ on the interval $[t_s, t_e]$ using breakpoints in $\hat{S}$, let $\tilde{q}_S$ be defined similarly for the set $S$. Then with probability $1-\frac{3}{2}\beta$, 
    \[\dt(\tilde{q}_{\hat{S}}, q)\le \mathrm{ErrBound}(\tilde{q}_S;q,\rho/2,\beta,1))+{(\sqrt{2}-1)\sqrt{\ln(1/\beta)/\rho}}.\]
\end{lemma}

\subsection{Experiments under CGP}
 \label{sec:app_cgpexp}
In this section, we provide evaluations for the methods in the CGP model, on the same set of experiments as those in Section~\ref{sec:exp}. We observe similar improvements in our CGP mechanisms over the baseline methods in all the experiments, for all parameter settings, as those shown for the GP mechanisms in Section~\ref{sec:exp}.

\begin{figure}[h]
      \centering
    \begin{subfigure}[t]{0.3\linewidth}
        \includegraphics[width = \textwidth]{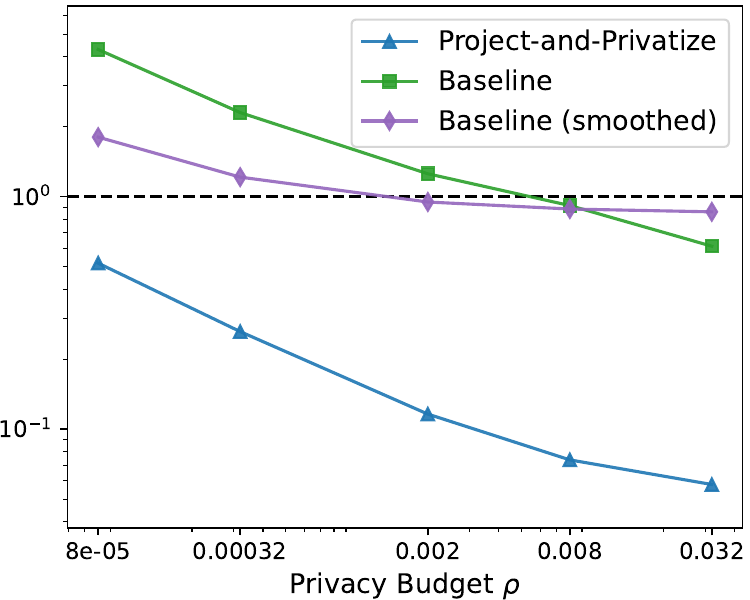}
        \subcaption{$L^2$ error}
        \label{fig:ECG-CGP_avg}
    \end{subfigure}
    \;\;\;\;\;
    \begin{subfigure}[t]{0.3\linewidth}
        \includegraphics[width = \textwidth]{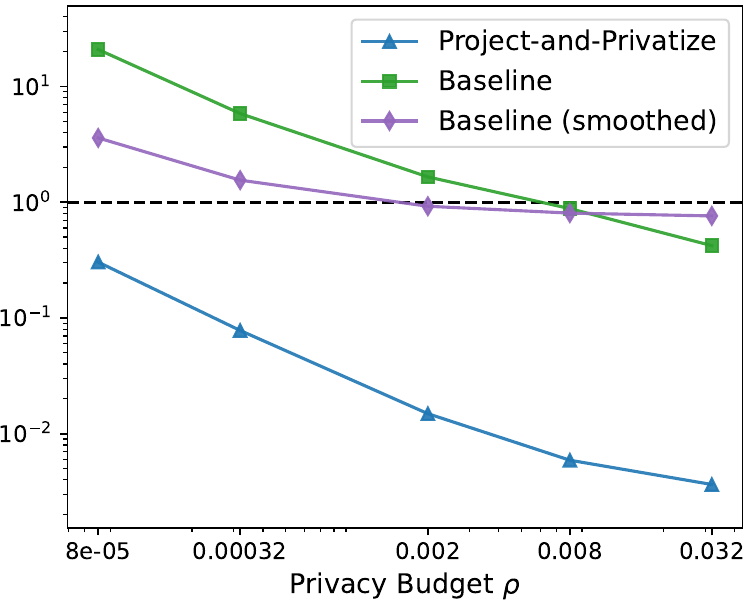} 
        \subcaption{$L^2$-squared error (MSE)}
        \label{fig:ECG-CGP_mse}
    \end{subfigure}
    \vspace{-4pt}
    \caption{CGP results on ECG dataset}
    \label{fig:ECG-CGP}
\end{figure}

\begin{figure}[h]
      \centering
    \begin{subfigure}[t]{0.3\linewidth}
        \includegraphics[width = \textwidth]{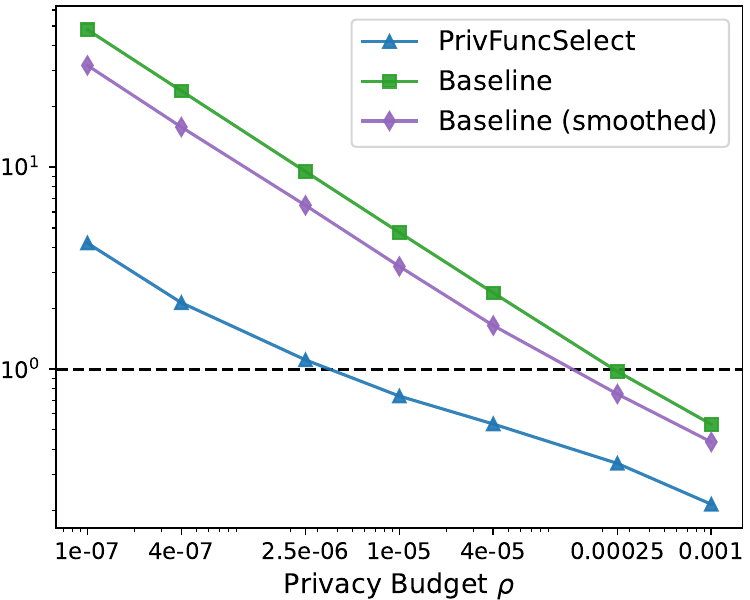} 
        \subcaption{$L^2$ error}
        \label{fig:Synth-Adapt-GCP_avg}
    \end{subfigure}
    \;\;\;\;\;
    \begin{subfigure}[t]{0.3\linewidth}
        \includegraphics[width = \textwidth]{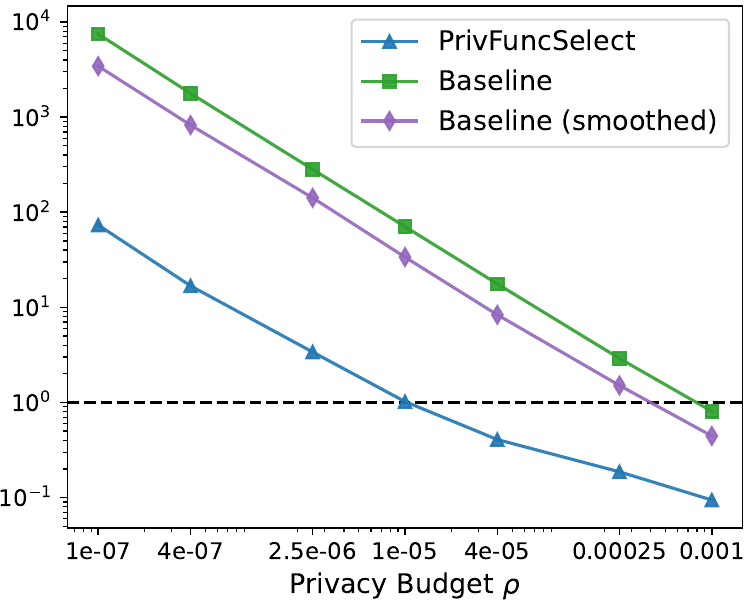}
        \subcaption{$L^2$-squared error (MSE)}
        \label{fig:Synth-Adapt-CGP_mse}
    \end{subfigure}
    \vspace{-4pt}
	\caption{GP results on synthetic Gaussian curves}
	\label{fig:Synth-Adapt-CGP}
\end{figure}

\begin{figure}[h]
      \centering
    \begin{subfigure}[t]{0.3\linewidth}
        \includegraphics[width = \textwidth]{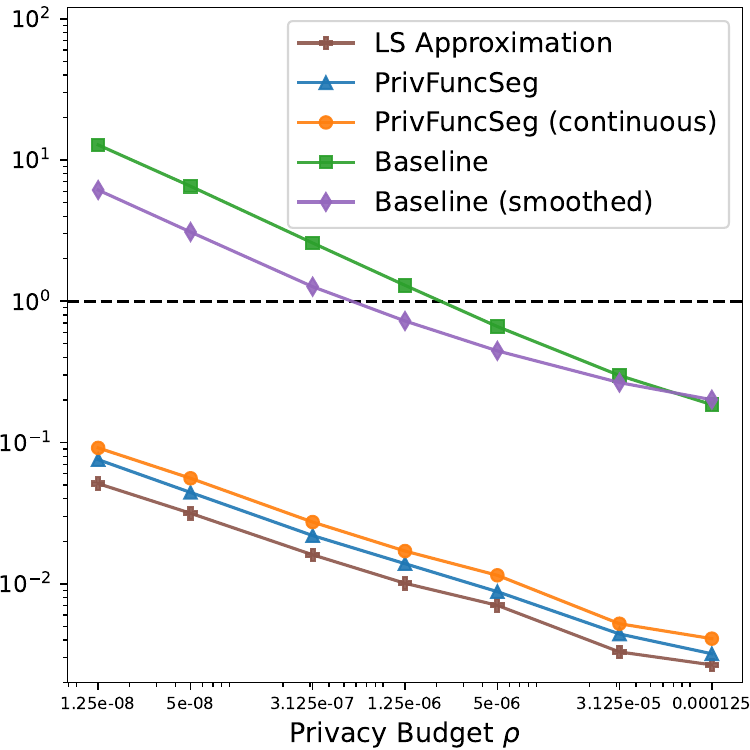}
        \subcaption{$L^2$ error}
        \label{fig:Taxi-CGP_avg}
    \end{subfigure}
   \;\;\;\;\;
    \begin{subfigure}[t]{0.3\linewidth}
        \includegraphics[width = \textwidth]{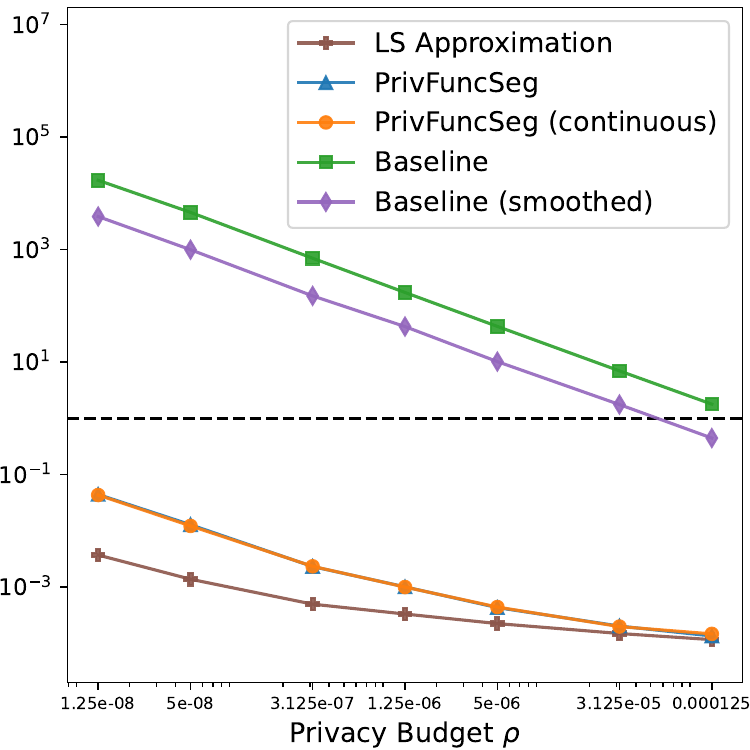} 
        \subcaption{$L^2$-squared error (MSE)}
        \label{fig:Taxi-CGP_mse}
    \end{subfigure}
    \vspace{-4pt}
	\caption{CGP results on the taxi trajectory dataset}
	\label{fig:Taxi-CGP}
\end{figure}
\;
\;
\begin{figure}[H]
    \centering
	\includegraphics[width = 0.85\linewidth]{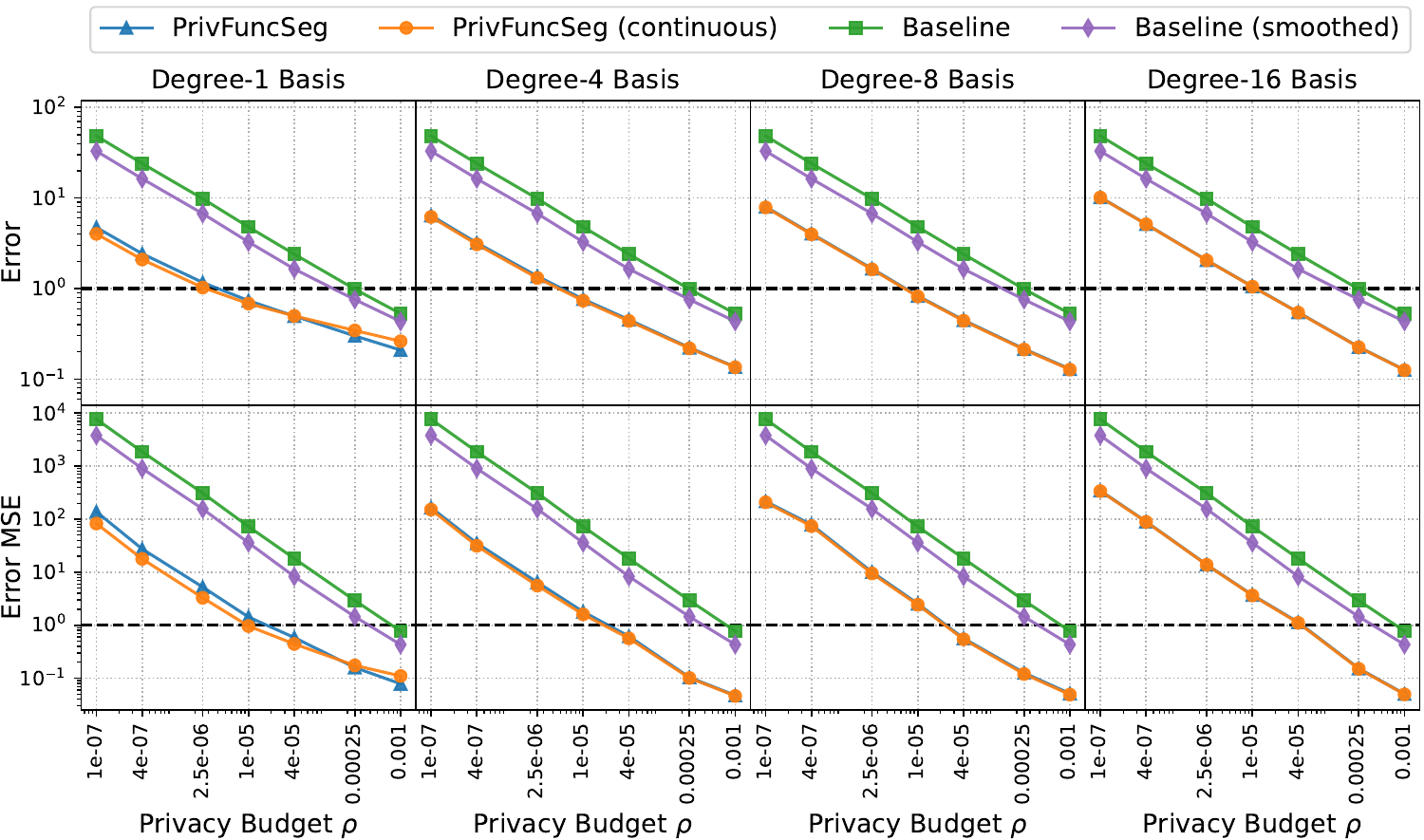}
    \vspace{-4pt}
	\caption{GP results on synthetic Gaussian curves}
	\label{fig:Synth-CGP}
\end{figure}

\newpage
\section{Deferred Proofs}
\subsection{Utility of SVT}
 The analysis on SVT below is mostly the same as that in \cite{dong2023universal}, and re-produced here for completeness. 
\begin{proof}  [Proof of Lemma~\ref{lm:svt_util}]
   We will use the one-sided and two-sided Laplace tail bounds: For $Z\sim \mathrm{Lap}(b)$, $\Pr[Z\le b\ln(1/\beta)]\ge 1-\beta/2$ and $\Pr[|Z|\le b\ln(1/\beta)]\ge 1-\beta$.
    
    Let $W$, $Z_{j}\sim\mathrm{Lap}(3K/\varepsilon)$ denote the Laplace random variables in the SVT call. The analysis requires bounding $|W|$ and $Z_j$ for $1\le j \le t$. Notice that SVT halts at $t$ if $Z_{t} + g_{t}(x) \ge \mathcal{T}+W$, i.e. $Z_{t}-W \ge \mathcal{T}-g_{t}(x)$. We have
    $g_{t}(x)-\mathcal{T} \ge 6K\ln(2/\beta)/{\varepsilon}$ by assumption, and 
    \begin{align}
\label{eqn:Wbound_svt}
    \Pr\left[|W| \le \frac{3K}{\varepsilon}\ln(2/\beta)\right] &\ge 1-\beta/2\\
\label{eqn:Ztbound_svt}
    \Pr\left[Z_{t} \ge -\frac{3K}{\varepsilon}\ln(2/\beta)\right]&= \Pr\left[-Z_{t} \le \frac{3K}{\varepsilon}\ln(2/\beta)\right] \ge 1-\beta/4\\
\nonumber
    \Pr\left[Z_{t}-W \ge -(g_{t}(x)-\mathcal{T})\right]&\ge \Pr\left[Z_{t}-|W| \ge -\frac{6K}{\varepsilon}\ln(2/\beta)\right] \ge 1-3\beta/4.
     \end{align}
     Thus, SVT halts at or before $t$ with probability $1-3\beta/4$ (by a union bound over \eqref{eqn:Wbound_svt}, \eqref{eqn:Ztbound_svt}). We assume this holds in the remaining analysis. 
     $\Pr[Z_j > \frac{3K}{\varepsilon}\ln(2t/\beta)]\le \beta/(4t)$ for each $1\le j \le t$. Then 
     \begin{equation}
     \label{eqn:Zmaxbound_svt}
     Z_{\bar{k}}\le \max_{1\le j\le t} Z_j \le \frac{3K}{\varepsilon}\ln(2t/\beta)
     \end{equation}
     where the last inequality holds with probability at least $1-\beta/4$. Thus, with probability at least $1-\beta$ (by union bound over \eqref{eqn:Wbound_svt}, \eqref{eqn:Ztbound_svt}, \eqref{eqn:Zmaxbound_svt}),
     \[
     g_{\bar{k}}(x)\ge \mathcal{T}+W-Z_{\bar{k}}\ge \mathcal{T}-|W|- Z_{\bar{k}} \ge \mathcal{T}-\frac{3K}{\varepsilon}\ln(2/\beta)-\frac{3K}{\varepsilon}\ln(2t/\beta).
     \]
\end{proof}

\subsection{Lipschitzness of finite-dimensional function space projection}
\begin{proof} [Proof of Lemma~\ref{lm:fidim_lipschitz}]
First, the space $L^2(I)$ of square-integrable functions is a Hilbert space whose inner product induces the $L^2$ metric $\dist_2(f,g):=\sqrt{\int_{I}(f(t)-g(t))^2dt}$.
Let $U\subseteq L^2(I)$ be any subset. 
Next, the span of $m$ linearly independent functions in $L^2(I)$ forms a subspace; let $U_f$ be any such subspace. Then for any pair $q,q'\in U$, by the Hilbert projection theorem\footnote{Technically speaking,  appealing to the Hilbert projection theorem is not necessary, since the constructions in Appendix~\ref{sec:app_ls} provide the projections we need.} there are $f_{a(q)}, f_{a(q')} \in U_f$ such that $f_{a(q)}=\argmin_{z\in U_f}\dist_2(q, z)$ and $f_{a(q')}=\argmin_{z\in U_f}\dist_2(q', z)$; moreover, $q-f_{a(q)}, q'-f_{a(q')}$ are orthogonal to $U_f$. Denote $c_{q,q'}:=f_{a(q)}-f_{a(q')}$, $c_{q,q'}^{\perp}:= (q-f_{a(q)})-(q'-f_{a(q')})$, then $c_{q,q'}\in U_f$ and $c_{q,q'}^{\perp}$ is orthogonal to $U_f$.
\begin{align*}
\dist\nolimits_2(q,q')^2 &=\dist\nolimits_2\left(f_{a(q)}+q-f_{a(q)}, f_{a(q')}+q'-f_{a(q')}\right)^2\\
&= \int_{I} \left(f_{a(q)}(t)+q(t)-f_{a(q)}(t) - (f_{a(q')}(t)+q'(t)-f_{a(q')}(t))\right)^2dt\\
&=\int_{I} (c_{q,q'}(t)^2+c_{q,q'}^{\perp}(t)^2 + 2c_{q,q'}(t)c_{q,q'}^{\perp}(t))dt\\
&= \int_{I} c_{q,q'}(t)^2 dt + \underbrace{\int_{I} c_{q,q'}^{\perp}(t)^2 dt}_{\ge 0} + 2\underbrace{\int_{I}c_{q,q'}(t)c_{q,q'}^{\perp}(t))dt}_{0}\\
&\ge \int_{I} c_{q,q'}(t)^2 dt = \dist\nolimits_2(f_{a(q)},f_{a(q')})^2.
\end{align*}

\end{proof}

\subsection{Utility of $\mathrm{Project}$-$\mathrm{and}$-$\mathrm{Privatize}$}
\begin{lemma} [Lemma~\ref{lm:finapprox_priv_distsquared}]
    Let $U_f$ be the subspace given by the span of $\{\phi_j(\cdot)\}_{j\in [m]}$. For $q\in U$, let $f_a\in U_f$ be a projection of $q$ in $U_f$. For $\tilde{a}:=a+\frac{1}{b}\Sigma^{1/2}Z$ where $Z$ is a random vector with mean zero, we have
        \[
    \mathbb{E}[\dt(q,f_{\tilde{a}})^2] = \dt(q,f_{a})^2 + \frac{1}{b^2}\mathbb{E}[\|Z\|^2].
    \]
\end{lemma}

\begin{proof}[Proof of Lemma~\ref{lm:finapprox_priv_distsquared}]
We need to discuss all three cases, though the derivations are similar. \\
1. Real-valued coefficients and basis functions.
\begin{alignat*}{2}
        \dt(q,f_{\tilde{a}})^2 &= &&\int_{I} \left(q(t)-\sum_{j\in [m]}\tilde{a}_j\phi_j(t)\right)^2 dt
        =\int_{I}\left(q(t)-\left(a + \frac{1}{b}\Sigma^{\frac{1}{2}} Z\right)^{\mathsf{T}}\begin{bmatrix}
            \phi_1(t)\\
            \vdots\\
            \phi_m(t)
        \end{bmatrix}\right)^2 dt\\
        &= &&\int_{I} \left(q(t)-\begin{bmatrix}
            a_1 & \dotsb & a_m
        \end{bmatrix}\begin{bmatrix}
            \phi_1(t)\\
            \vdots\\
            \phi_m(t)
        \end{bmatrix} -\frac{1}{b}\begin{bmatrix}
            z_1 & \dotsb z_m
        \end{bmatrix}\Sigma^{\frac{1}{2}}\begin{bmatrix}
            \phi_1(t)\\
            \vdots\\
            \phi_m(t)
        \end{bmatrix} \right)^2 dt\\
        &= &&\int_{I} (q(t)-\sum_{j\in [m]} a_j\phi_j(t))^2 dt - \int_{I} \frac{2}{b} \underbrace{\left(q(t)-\sum_{j\in [m]} a_j\phi_j(t)\right)\begin{bmatrix}
            \phi_1(t)& \dotsb & \phi_m(t)
        \end{bmatrix}\Sigma^{\frac{1}{2}} }_{=:
        \begin{bmatrix}
            v_1(t) & \dotsb & v_m(t)
        \end{bmatrix}
        }\begin{bmatrix}
                z_1\\
                \vdots\\
                z_m
            \end{bmatrix} dt \\
            &\;\;\; &&+ \int_{I}\frac{1}{b^2}\underbrace{\begin{bmatrix}
            z_1 & \dotsb z_m
        \end{bmatrix}\Sigma^{\frac{1}{2}}}_{\in \mathbb{R}^m}\begin{bmatrix}
                \phi_1(t)\\
                \vdots\\
                \phi_m(t)
            \end{bmatrix} \begin{bmatrix}
            \phi_1(t)& \dotsb & \phi_m(t)
        \end{bmatrix}\underbrace{\Sigma^{\frac{1}{2}}\begin{bmatrix}
            z_1 \\
            \vdots\\
            z_m
        \end{bmatrix} }_{\in\mathbb{R}^m}dt\\
        &= &&\dt(q,f_a)^2 - \frac{2}{b}        \begin{bmatrix}
            \int_{I} v_1(t) dt & \dotsb & \int_{I} v_m(t) dt
        \end{bmatrix}\begin{bmatrix}
                z_1\\
                \vdots\\
                z_m
            \end{bmatrix} \\
        &\;\;&&+\frac{1}{b^2}\begin{bmatrix}
            z_1 & \dotsb z_m
        \end{bmatrix}\Sigma^{\frac{1}{2}} \underbrace{\begin{bmatrix}
            \int_{I} \phi_1(t)^2 dt & \dotsb & \int_{I} \phi_1\phi_m(t) dt\\
            \vdots & & \vdots\\
            \int_{I} \phi_m(t)\phi_1(t)dt & \dotsb & \int_{I} \phi_m(t)\phi_m(t)dt
        \end{bmatrix} }_{\Sigma^{-1}, \text{\;by Fact\;}\ref{ft:integ_linear_mat}}\Sigma^{\frac{1}{2}}\begin{bmatrix}
            z_1 \\
            \vdots\\
            z_m
        \end{bmatrix} \\
        &=&&\dt(q,f_a)^2 - \frac{2}{b}        \begin{bmatrix}
            \int_{I} v_1(t) dt & \dotsb & \int_{I} v_m(t) dt
        \end{bmatrix}\begin{bmatrix}
                z_1\\
                \vdots\\
                z_m
            \end{bmatrix} + \frac{1}{b^2} \begin{bmatrix}
            z_1 &\dotsb &
            z_m
        \end{bmatrix}\begin{bmatrix}
            z_1 \\
            \vdots\\
            z_m
        \end{bmatrix},\\
\mathbb{E}[\dt(q,f_{\tilde{a}})^2] &= &&\dt(q,f_a)^2 -\frac{2}{b}\begin{bmatrix}
            \int_{I} v_1(t) dt & \dotsb & \int_{I} v_m(t) dt
        \end{bmatrix}\underbrace{\mathbb{E}\left(\begin{bmatrix}
                z_1\\
                \vdots\\
                z_m
            \end{bmatrix}\right)}_{0}+\frac{1}{b^2}\mathbb{E}\left[\|Z\|^2\right]\\
            &= &&\dt(q,f_a)^2 +\frac{1}{b^2}\mathbb{E}[\|Z\|^2].
    \end{alignat*}
    \\
2. Real-valued coefficients, vector-valued basis. Here $\phi_j = [\phi_{j,1},\dotsb,\phi_{j,n}]^{\mathsf{T}}$ for $j\in [m]$. Let $\Phi(t):=\begin{bmatrix}
         \phi_{1,1}(t) & \dotsb & \phi_{m,1}(t) \\
         \vdots & & \vdots\\
         \phi_{1,n}(t) & \dotsb & \phi_{m,n}(t)
     \end{bmatrix}$. Note that $f_a(t)=\sum_{j\in [m]} a_j [\phi_{j,1}(t),\dotsb,\phi_{j,n}(t)]^{\mathsf{T}} = \Phi(t) a$.
\begin{alignat*}{2}
     \dt(q,f_{\tilde{a}})^2 &=&& \int_{I} \left\|\begin{bmatrix}
         q_1(t) \\
         \vdots\\
         q_n(t)
     \end{bmatrix}-\sum_{j\in [m]}\tilde{a}_j\begin{bmatrix}
         \phi_1(t) \\
         \vdots\\
         \phi_m(t)
     \end{bmatrix}\right\|^2 dt\\
        &=&& \int_{I} \left\|\begin{bmatrix}
         q_1(t) \\
         \vdots\\
         q_n(t)
     \end{bmatrix}-\begin{bmatrix}
         \phi_{1,1}(t) & \dotsb & \phi_{m,1}(t) \\
         \vdots & & \vdots\\
         \phi_{1,n}(t) & \dotsb & \phi_{m,n}(t)
     \end{bmatrix}(a+\frac{1}{b}\Sigma^{\frac{1}{2}}Z)\right\|^2 dt\\
     &=&& \int_{I} \left\|\begin{bmatrix}
         q_1(t) \\
         \vdots\\
         q_n(t)
     \end{bmatrix}-\Phi(t) a\right\|^2 dt + \frac{1}{b^2}\int_I \left\|\Phi(t)\Sigma^{\frac{1}{2}}Z\right\|^2 dt - 2\frac{1}{b}\int_I \underbrace{ (q(t)-\Phi(t)a)^{\mathsf{T}} \Phi(t) \Sigma^{\frac{1}{2}}}_{=:[v_1(t),\dots, v_m(t)]} Zdt \\
     &=&& \dt(q,f_a)^2+\frac{1}{b^2}\int_I Z^{\mathsf{T}}\Sigma^{\frac{1}{2}}\Phi(t)^{\mathsf{T}}\Phi(t)\Sigma^{\frac{1}{2}} Z dt - \frac{2}{b} \begin{bmatrix}
            \int_{I} v_1(t) dt & \dotsb & \int_{I} v_m(t) dt
        \end{bmatrix}[z_1,\dots,z_m]^{\mathsf{T}}\\
        &=&& \dt(q,f_a)^2+\frac{1}{b^2}\int_I \underbrace{Z^{\mathsf{T}}\Sigma^{\frac{1}{2}}}_{\in \mathbb{R}^m}\begin{bmatrix}
            \sum_{r\in[n]}\phi_{1,r}(t)^2 & \dotsb & \sum_{r\in[n]}\phi_{1,r}(t)\phi_{m,r}(t)\\
            \vdots & \dotsb  &\vdots\\
            \sum_{r\in[n]}\phi_{m,r}(t)\phi_{1,r}(t) & \dotsb & \sum_{r\in[n]}\phi_{m,r}(t)^2
        \end{bmatrix}\underbrace{\Sigma^{\frac{1}{2}} Z}_{\in \mathbb{R}^m} dt \\
        &\;\;\; &&- \frac{2}{b} \begin{bmatrix}
            \int_{I} v_1(t) dt \\
            \vdots \\
             \int_{I} v_m(t) dt
        \end{bmatrix}^{\mathsf{T}}\begin{bmatrix}
            z_1\\
            \vdots\\
            z_m
        \end{bmatrix}\\
        &=&& \dt(q,f_a)^2+\frac{1}{b^2} {Z^{\mathsf{T}}\Sigma^{\frac{1}{2}}}\underbrace{\begin{bmatrix}
            \int_I\sum_{r\in[n]}\phi_{1,r}(t)^2dt & \dotsb & \int_I\sum_{r\in[n]}\phi_{1,r}(t)\phi_{m,r}(t)dt\\
            \vdots & \dotsb  &\vdots\\
            \int_I\sum_{r\in[n]}\phi_{m,r}(t)\phi_{1,r}(t)dt & \dotsb & \int_I\sum_{r\in[n]}\phi_{m,r}(t)^2dt
        \end{bmatrix}}_{\Sigma^{-1}, \text{\;by Fact\;}\ref{ft:integ_linear_mat}}{\Sigma^{\frac{1}{2}} Z} \\
        &\;\;\; &&- \frac{2}{b} \begin{bmatrix}
            \int_{I} v_1(t) dt \\
            \vdots \\
             \int_{I} v_m(t) dt
        \end{bmatrix}^{\mathsf{T}}\begin{bmatrix}
            z_1\\
            \vdots\\
            z_m
        \end{bmatrix}\\
        &=&& \dt(q,f_a)^2+\frac{1}{b^2}\|Z\|-\frac{2}{b}\begin{bmatrix}
            \int_{I} v_1(t) dt & \dotsb & \int_{I} v_m(t) dt
        \end{bmatrix}\begin{bmatrix}
            z_1\\
            \vdots\\
            z_m
        \end{bmatrix},\\
     \mathbb{E}[\dt(q,f_{\tilde{a}})^2]&=&&\dt(q,f_a)^2+\frac{1}{b^2}\mathbb{E}[\|Z\|]-\frac{2}{b}\begin{bmatrix}
            \int_{I} v_1(t) dt & \dotsb & \int_{I} v_m(t) dt
        \end{bmatrix}\underbrace{\mathbb{E}\left[\begin{bmatrix}
            z_1 & \dotsb &z_m
        \end{bmatrix}^{\mathsf{T}}\right]}_{0}.
\end{alignat*}
3. Vector-valued coefficients, real-valued basis. Here $a_j = [a_{j,1},\dotsb,a_{j,n}]$ for $j\in [n]$. Write $a_{:r}=[a_{1,r},\dotsb,a_{m,r}]^{\mathsf{T}}$ and $z_{:,r}=[z_{1,r},\dotsb,z_{m,r}]^{\mathsf{T}}$ for $r\in [n]$. We have the coefficient vector $a=[a_{:1},\dotsb,a_{:n}]^{\mathsf{T}}$, noise vector $Z=[z_{:1},\dotsb,z_{:n}]^{\mathsf{T}}$, and privatized coefficient vector $\tilde{a}=a+\frac{1}{b}\Sigma^{1/2}Z$, where $\Sigma^{1/2}$ is a diagonal block matrix with $\Sigma_0^{1/2}$ repeated in the diagonal. Let $\Phi(t):=[\phi_1(t),\dotsb,\phi_m(t)]^{\mathsf{T}}$.
\begin{align*}
    \dt(q,f_{\tilde{a}})&=\int_{I} \left\|\begin{bmatrix}
         q_1(t) \\
         \vdots\\
         q_n(t)
     \end{bmatrix}-\begin{bmatrix}
         \sum_{j\in [m]}\tilde{a}_{j,1}\phi_j(t) \\
         \vdots\\
         \sum_{j\in [m]}\tilde{a}_{j,n}\phi_j(t)
     \end{bmatrix}\right\|^2 dt \\
     &= \int_{I} \left\|\begin{bmatrix}
         q_1(t) \\
         \vdots\\
         q_n(t)
     \end{bmatrix}-\begin{bmatrix}
         [\phi_1(t),\dotsb,\phi_m(t)]\cdot (a_{:1}+\frac{1}{b}\Sigma_0^{1/2}z_{:1}) \\
         \vdots\\
         [\phi_1(t),\dotsb,\phi_m(t)]\cdot (a_{:n}+\frac{1}{b}\Sigma_0^{1/2}z_{:n})
     \end{bmatrix}\right\|^2 dt\\
     &= \int_{I} \left\|\begin{bmatrix}
         q_1(t) \\
         \vdots\\
         q_n(t)
     \end{bmatrix}-\begin{bmatrix}
         \Phi(t)^{\mathsf{T}} a_{:1} \\
         \vdots\\
         \Phi(t)^{\mathsf{T}} a_{:n}
     \end{bmatrix}\right\|^2 dt + \frac{1}{b^2}\int_{I} \left\|\begin{bmatrix}
         \Phi(t)^{\mathsf{T}} \Sigma_0^{1/2}z_{:1} \\
         \vdots\\
         \Phi(t)^{\mathsf{T}} \Sigma_0^{1/2}z_{:n}
     \end{bmatrix}\right\|^2 dt\\
     &{\;\;\;\;}- \frac{2}{b} \int_I [q_1(t),\dotsb,q_m(t)]\begin{bmatrix}
         \Phi(t)^{\mathsf{T}}\cdot \Sigma_0^{1/2}z_{:1}\\
         \vdots\\
        \Phi(t)^{\mathsf{T}}\cdot \Sigma_0^{1/2}z_{:n}
     \end{bmatrix}\\
     &= \dt(q,f_a)^2 + \frac{1}{b^2}\int_I \left[z_{:1}^{\mathsf{T}} \Sigma_0^{1/2}\Phi(t) ,\dotsb,z_{:n}^{\mathsf{T}} \Sigma_0^{1/2}\Phi(t)\right]\begin{bmatrix}
         \Phi(t)^{\mathsf{T}} \Sigma_0^{1/2}z_{:1} \\
         \vdots\\
         \Phi(t)^{\mathsf{T}} \Sigma_0^{1/2}z_{:n}
     \end{bmatrix}dt \\
     &{\;\;\;\;}-  \frac{2}{b} \int_I \underbrace{[q_1(t),\dotsb,q_m(t)]\begin{bmatrix}
         \Phi(t)^{\mathsf{T}} & &\\
         & \ddots &\\
         & & \Phi(t)^{\mathsf{T}}
     \end{bmatrix} \Sigma}_{=:[v_{1,1}(t),\dotsb,v_{m,1}(t),\dotsb,v_{m,n}(t)]} Z dt\\
     &= \dt(q,f_a)^2 + \frac{1}{b^2}\sum_{r\in [n]}\int_I \left(z_{:r}^{\mathsf{T}} \Sigma_0^{1/2}\Phi(t)\Phi(t)^{\mathsf{T}} \Sigma_0^{1/2}z_{:r} \right) dt \\
     &{\;\;\;\;}-\frac{2}{b} \begin{bmatrix}
         \int_I v_{1,1}(t)dt & \dotsb & \int_I v_{m,n}(t)
     \end{bmatrix} [z_{1,1},\dotsb,z_{m,n}]^{\mathsf{T}}\\
     \mathbb{E}[\dt(q,f_{\tilde{a}})&= \dt(q,f_a)^2 +\frac{1}{b^2}\mathbb{E}[\|Z\|^2] - \frac{2}{b} \begin{bmatrix}
         \int_I v_{1,1}(t)dt & \dotsb & \int_I v_{m,n}(t)
     \end{bmatrix} \underbrace{\mathbb{E}\left[[z_{1,1},\dotsb,z_{m,n}]^{\mathsf{T}}\right]}_{0}.
\end{align*}

\end{proof}

\subsection{Utility of $\mathrm{ReduceSeg}$ (Algorithm~\ref{alg:algo_redseg})}
\begin{proof} [Proof of Lemma~\ref{lm:util_redseg}]
    Let $\hat{t}_s, \hat{t}_e$ denote the end points of the running interval in $\mathrm{ReduceSeg}$. There are at most $2^{k_1}$ calls to $\mathrm{ReduceSeg}$, where at most $2^{k_1}$ $\mathrm{Lap}(1)$ random variables are drawn, which are simultaneously lower bounded by $-\ln(2^{k_1}/\beta)$ with probability at least $1-\beta/2$. Assume this holds in the remaining analysis. In each call, we have
    \begin{equation}
    \label{eqn:redseg_subdt_bound}
    \dt(f_{[\hat{t}_s,\hat{t}_e]},q\cdot\mathbb{1}_{[\hat{t}_s,\hat{t}_e]}) = \mathrm{err}-\frac{2^l}{\varepsilon'}Z \le \mathrm{err}+\frac{2^l}{\varepsilon'}\ln(2^{k_1}/\beta),
    \end{equation}
    where $f_{[\hat{t}_s,\hat{t}_e]}$ is computed using $S'$ (i.e. half of the sub-intervals). Thus, if line $9$ is true, then the approximation error $\dt(f_{[\hat{t}_s,\hat{t}_e]},q\cdot\mathbb{1}_{[\hat{t}_s,t_e]})$ is dominated by $\frac{2^{k-1}emn}{2(e-1)B}=|[\hat{t}_s,\hat{t}_e)\cap S'|\frac{emn}{2(e-1)B}$ after removing half of the sub-intervals.
    
    Assume the number of breakpoints is reduced at least once (i.e. at level $l=1$); otherwise $\hat{S}=S$ and there would be nothing to prove. Starting from $C=\emptyset$, add the pair $(\hat{t}_s,\hat{t}_e)$ to $C$ whenever line $9$ is false. Let $\mathrm{par}(\hat{t}_s,\hat{t}_e)$ denote the sub-interval at its parent call, so $[\hat{t}_s,\hat{t}_e]\subset \mathrm{par}(\hat{t}_s,\hat{t}_e)$; there were twice as many sub-intervals in the call on $\mathrm{par}(\hat{t}_s,\hat{t}_e)$, during which line $9$ was true. 
    We have 
    \begin{align*}
\dt(f_{[\hat{t}_s,\hat{t}_e]},q\cdot\mathbb{1}_{[\hat{t}_s,\hat{t}_e]}) &\le \dt(f_{\mathrm{par}(\hat{t}_s,\hat{t}_e)},q\cdot\mathbb{1}_{\mathrm{par}(\hat{t}_s,\hat{t}_e)}) \\
    &\le 2|[\hat{t}_s,\hat{t}_e)\cap \hat{S}|\frac{emn}{2(e-1)B}=|[\hat{t}_s,\hat{t}_e)\cap \hat{S}|\frac{emn}{(e-1)B},
    \end{align*}
    since the non-private approximation on $\mathrm{par}(\hat{t}_s,\hat{t}_e)$ has error bounded by the number of its sub-intervals times $\frac{emn}{2(e-1)B}$,
    and $B\ge \varepsilon/2$ is non-increasing. When the recursive calls terminate, we have that the original interval is equal to the union of sub-intervals defined by the pairs in $C$. Let $\pj_{\hat{S}}(q)$ denote the non-private best approximation of $q$ using intervals in $\hat{S}$. We have
    \begin{align*}
    \dt(\pj\nolimits_{\hat{S}}(q),q) &\le \sum_{(\hat{t}_s,\hat{t}_e)\in C} \dt(\pj\nolimits_{\hat{S}}(q\cdot \mathbb{1}_{[\hat{t}_s,\hat{t}_e]}),q\cdot  \mathbb{1}_{[\hat{t}_s,\hat{t}_e]})\\
    &\le \sum_{(\hat{t}_s,\hat{t}_e)\in C} \frac{|[\hat{t}_s,\hat{t}_e)\cap \hat{S}|emn}{(e-1)\varepsilon/2} = \frac{(|\hat{S}|-1)emn}{(e-1)\varepsilon/2},\\
    \dt(\tilde{q}_{\hat{S}}, q) &\le \dt(\pj\nolimits_{\hat{S}}(q),q)+\dt(\tilde{q}_{\hat{S}},\pj\nolimits_{\hat{S}}(q))\\
    &\le \frac{(|\hat{S}|-1)emn}{(e-1)\varepsilon/2} + \frac{\eta(\beta;(|\hat{S}|-1)m,1)}{\varepsilon/2} \;(\text{w.p. } 1-\beta)\\
    &\le \frac{(|\hat{S}|-1)emn}{(e-1)\varepsilon/2} +\frac{e}{(e-1)\varepsilon/2}\left(|\hat{S}|-1)mn+\ln(1/\beta)\right)\\
    &= \frac{e}{(e-1)\varepsilon/2} {(2|\hat{S}|-2)mn} +\frac{e}{(e-1)\varepsilon/2}\ln(1/\beta)\\
    &\le \frac{\eta(\beta;(|S|-1)mn,1)}{\varepsilon/2} \;\;\;\;\;\;\;\;\;\;\;\;\;\;\;\;(\text{since }2|\hat{S}|\le |S|) \\
    &\le \mathrm{ErrBound}(\tilde{q}_S;q,\varepsilon/2,\beta,1)). \lablast \label{eqn:redseg_errbound}
    \end{align*}
Combining \eqref{eqn:redseg_subdt_bound} and \eqref{eqn:redseg_errbound} above, we have with probability $1-\frac{3}{2}\beta$, $\dt(\tilde{q}_{\hat{S}}, q)\le \mathrm{ErrBound}(\tilde{q}_S;q,\varepsilon/2,\beta,1))$.
\end{proof}

\subsection{Utility of $\mathrm{ReduceSegCGP}$ (Algorithm~\ref{alg:algo_redseg_cgp})}
\begin{proof} [Proof of Lemma~\ref{lm:util_redseg_cgp}]
The proof proceeds similarly to that for Lemma~\ref{lm:util_redseg}.
 The $\mathcal{N}(0,1)$ random variables in the (at most) $2^{k_1}$ calls are simultaneously lower bounded by $-\sqrt{2\ln(2^{k_1}/\beta)}$ with probability at least $1-\beta/2$. Assume this holds. In each call, we have
    \[
    \dt(f_{[\hat{t}_s,\hat{t}_e]},q\cdot\mathbb{1}_{[\hat{t}_s,\hat{t}_e]}) = \mathrm{err}-\frac{\sqrt{2^l}}{\sqrt{2\rho'}}Z \le \mathrm{err}+\frac{\sqrt{2^l}}{\sqrt{2\rho'}}\sqrt{2\ln(2^{k_1}/\beta)},
    \]
    where $f_{[\hat{t}_s,\hat{t}_e]}$ is computed using $S'$. If line $9$ is true, then the approximation error $\dt(f_{[\hat{t}_s,\hat{t}_e]},q\cdot\mathbb{1}_{[\hat{t}_s,t_e]})$ is dominated by $\frac{(\sqrt{2}-1)\sqrt{2^{k-1}mn}}{\sqrt{2}\sqrt{2B}}=\frac{(\sqrt{2}-1)}{\sqrt{2}\sqrt{2B}}\sqrt{|[\hat{t}_s,\hat{t}_e)\cap {S'}|mn}$ after removing half of the sub-intervals .
    
    Assume the number of sub-intervals is reduced at least once. Starting from $C=\emptyset$, add the pair $(\hat{t}_s,\hat{t}_e)$ to $C$ whenever line $9$ is false. Let $\mathrm{par}(\hat{t}_s,\hat{t}_e)$ denote the sub-interval at its parent call.
    Note that we have \[\dt(f_{[\hat{t}_s,\hat{t}_e]},q\cdot\mathbb{1}_{[\hat{t}_s,\hat{t}_e]})\le \dt(f_{\mathrm{par}(\hat{t}_s,\hat{t}_e)},q\cdot\mathbb{1}_{\mathrm{par}(\hat{t}_s,\hat{t}_e)}) \le \frac{(\sqrt{2}-1)}{\sqrt{2}\sqrt{2B}}\sqrt{2|[\hat{t}_s,\hat{t}_e)\cap \hat{S}|mn},\] 
    since the non-private approximation on $\mathrm{par}(\hat{t}_s,\hat{t}_e))$ has error bounded by the square root of the number of its sub-intervals times $\frac{(\sqrt{2}-1)\sqrt{mn}}{\sqrt{2}\sqrt{2B}}$, 
    and $B\ge \rho/2$ is non-increasing. When the recursive calls terminate, we have that the original interval is equal to the union of sub-intervals defined by the pairs in $C$. Let $\pj_{\hat{S}}(q)$ denote the non-private best approximation of $q$ using intervals in $\hat{S}$. Then
    \begin{align*}
    \dt(\pj\nolimits_{\hat{S}}(q),q)^2 &= \sum_{(\hat{t}_s,\hat{t}_e)\in C} \dt(\pj\nolimits_{\hat{S}}(q\cdot \mathbb{1}_{[\hat{t}_s,\hat{t}_e]}),q\cdot  \mathbb{1}_{[\hat{t}_s,\hat{t}_e]})^2\\
    &\le \sum_{(\hat{t}_s,\hat{t}_e)\in C} \frac{(\sqrt{2}-1)^2}{{2\rho}}(2|[\hat{t}_s,\hat{t}_e)\cap \hat{S}|mn) = \frac{(\sqrt{2}-1)^2}{{\rho}}{(|\hat{S}|-1)mn},\\
    \dt(\tilde{q}_{\hat{S}}, q) &\le \dt(\pj\nolimits_{\hat{S}}(q),q)+\dt(\tilde{q}_{\hat{S}},\pj\nolimits_{\hat{S}}(q))\\
    &\le \frac{(\sqrt{2}-1)}{\sqrt{\rho}}\sqrt{(|\hat{S}|-1)mn} + \frac{\eta(\beta;(|\hat{S}|-1)mn,2)}{\sqrt{\rho}} \;\;\;(\text{with prob. } 1-\beta)\\
    &= \frac{(\sqrt{2}-1)}{\sqrt{\rho}}\sqrt{(|\hat{S}|-1)mn} +\frac{1}{\sqrt{\rho}}\underbrace{\sqrt{(|\hat{S}|-1)mn+2\sqrt{(|\hat{S}|-1)mn\ln(1/\beta)}+2\ln(1/\beta)}}_{\le \sqrt{\left(\sqrt{(|\hat{S}|-1)mn}+\sqrt{2\ln(1/\beta)}\right)^2}}\\
    &\le \frac{1}{\sqrt{\rho}}\left((\sqrt{2}-1)\sqrt{(|\hat{S}|-1)mn}+\left(\sqrt{(|\hat{S}|-1)mn}+\sqrt{2\ln(1/\beta)}\right)\right)\\
    &= \frac{1}{\sqrt{\rho}}\left(\sqrt{2}\sqrt{(|\hat{S}|-1)mn}+\sqrt{\ln(1/\beta)}+(\sqrt{2}-1)\sqrt{\ln(1/\beta)}\right)
    \\
    &= \frac{1}{\sqrt{\rho}}\left({\sqrt{2(|\hat{S}|-1)mn+2\sqrt{2(|\hat{S}|-1)mn\ln(1/\beta)}+\ln(1/\beta)}}+(\sqrt{2}-1)\sqrt{\ln(1/\beta)}\right)\\
    &\le
    \frac{1}{\sqrt{\rho}}\left({\sqrt{(|S|-1)mn+2\sqrt{(|S|-1)mn\ln(1/\beta)}+\ln(1/\beta)}}+(\sqrt{2}-1)\sqrt{\ln(1/\beta)}\right) \;\;\;\\
    &\;\;\;\;(\text{since }2|\hat{S}|\le |S|) \\
    &\le \frac{\eta(\beta;(|S|-1)mn,2)}{\sqrt{\rho}}+\frac{(\sqrt{2}-1)\sqrt{\ln(1/\beta)}}{\sqrt{\rho}}\\
    &\le \mathrm{ErrBound}(\tilde{q}_S;q,\rho/2,\beta,2))+\frac{(\sqrt{2}-1)\sqrt{\ln(1/\beta)}}{\sqrt{\rho}}.
    \end{align*}
\end{proof}

\section{Least squares projection}
\label{sec:app_ls}
\subsection{Real-valued basis functions and real-valued coefficients}
\label{sec:app_realbasis_realcoeff}
We discuss how the least squares {projection} can be computed for a $2$-dimensional $U_f$ spanned by $\{\phi_1(\cdot),\phi_2(\cdot)\}$. The extension to $m$-dimensional spaces for $m\ge 2$ follows analogously. In the least squares projection problem, given a function $q\in U$ we want to find $f_a\in U_f$ such that $\dt(q,f_a)^2=\int_I (q(t)-(a_1t+a_2))^2 dt$ is minimized; i.e., we want to solve the unconstrained program
\[
\text{minimize\;\;} \int_I (q(t)-(a_1\phi_1(t)+a_2\phi_2(t)))^2 dt,\;\; a=(a_1,a_2)\in \mathbb{R}^2.
\]
First, we show that the program is (strictly) convex. Denote the objective function as $F(a_1,a_2):=\int_I (q(t)-(a_1t+a_2))^2 dt$. We have
\begin{align*}
    \pdv{F}{a_1} &= \int_I \pdv{a_1}[q(t)-(a_1\phi_1(t)+a_2\phi_2(t)))^2] dt = \int_I -2\phi_1(t)\left(q(t)-(a_1\phi_1(t)+a_2\phi_2(t))\right) dt\\
    &=2\int_I \left(a_1\phi_1(t)^2+a_2\phi_1(t)\phi_2(t)-q(t)\phi_1(t)\right) dt,\\
    \pdv{F}{a_2} &= \int_I \pdv{a_2}[q(t)-(a_1\phi_1(t)+a_2\phi_2(t)))^2] dt = \int_I -2\phi_2(t)\left(q(t)-(a_1\phi_1(t)+a_2\phi_2(t))\right) dt\\
    &=2\int_I \left(a_2\phi_2(t)^2+a_1\phi_1(t)\phi_2(t)-q(t)\phi_2(t)\right) dt,\\
    \pdv[2]{F}{a_1} &= 2\int_I \phi_1(t)^2 dt,\;\;\;\;\; \pdv[2]{F}{a_2} =2\int_I \phi_2(t)^2 dt\\
    \pdv[2]{F}{a_1}{a_2} &= 2\int_I \phi_1(t)\phi_2(t) dt = \pdv[2]{F}{a_2}{a_1}.
\end{align*}
Therefore, the Hessian matrix is $H_F=\begin{bmatrix}
    \pdv[2]{F}{a_1} & \pdv[2]{F}{a_1}{a_2}\\
    \pdv[2]{F}{a_2}{a_1} & \pdv[2]{F}{a_2} 
\end{bmatrix}=2\Sigma^{-1}$, which is positive definite, and the program is strictly convex. Then, the unique minimizer of $F(a_1,a_2)$ is the point that makes $\pdv{F}{a_1}=0=\pdv{F}{a_2}$. That is,
\begin{align}
\nonumber
&\begin{bmatrix}
    2\int_I \left(a_1\phi_1(t)^2+a_2\phi_1(t)\phi_2(t)-q(t)\phi_1(t)\right) dt \\
    2\int_I \left(a_2\phi_2(t)^2+a_1\phi_1(t)\phi_2(t)-q(t)\phi_2(t)\right) dt
\end{bmatrix} = \begin{bmatrix}
    0\\
    0
\end{bmatrix}\\
\nonumber
\iff &\begin{bmatrix}
    a_1 \int_I \phi_1(t)^2 dt + a_2 \int_I \phi_1(t)\phi_2(t) dt\\
    a_1 \int_I \phi_1(t)\phi_2(t) dt + a_2 \int_I \phi_2(t)^2 dt
\end{bmatrix} = \begin{bmatrix}
    \int_I q(t)\phi_1(t) dt\\
    \int_I q(t)\phi_2(t) dt
\end{bmatrix}\\
\iff &\underbrace{\begin{bmatrix}
    \int_I \phi_1(t)^2 dt & \int_I \phi_1(t)\phi_2(t) dt\\
    \int_I \phi_1(t)\phi_2(t) dt & \int_I \phi_2(t)^2 dt
\end{bmatrix}}_{\Sigma^{-1}}\begin{bmatrix}
    a_1\\
    a_2
\end{bmatrix} = \begin{bmatrix}
    \int_I q(t)\phi_1(t) dt\\
    \int_I q(t)\phi_2(t) dt
\end{bmatrix} 
\iff \begin{bmatrix}
    a_1\\
    a_2
\end{bmatrix} = \Sigma \begin{bmatrix}
    \int_I q(t)\phi_1(t) dt\\
    \int_I q(t)\phi_2(t) dt\end{bmatrix}.
    \label{eqn:l2proj_real}
\end{align}
Let $a^*=\begin{bmatrix}
    a^*_1\\
    a^*_2
\end{bmatrix} := \Sigma \begin{bmatrix}
    \int_I q(t)\phi_1(t) dt\\
    \int_I q(t)\phi_2(t) dt\end{bmatrix}$. Then we can verify that 
    \[\int_I (q(t)-(a^*_1\phi_1(t)+a^*_2\phi_2(t))\phi_1(t)dt = 0 = \int_I (q(t)-(a^*_1\phi_1(t)+a^*_2\phi_2(t))\phi_2(t)dt,\] 
    i.e., $q-(a^*_1\phi_1+a^*_2\phi_2)$ is orthogonal to $U_f$.

\subsection{Vector-valued basis functions and real-valued coefficients}
\label{sec:app_vecbasis_realcoeff}
For vector-valued functions $u,v:I \rightarrow \mathbb{R}^n$, consider with inner product 
\[
\langle u, v\rangle = \int_I u(t)^T v(t) dt = \int_I (u_1(t)v_1(t)+\dotsb+u_n(t)v_n(t)) dt
\]
where $u=(u_1,\dotsb,u_n), v=(v_1,\dotsb, v_n)$. Note that the space containing vector-valued functions $u$ such that $\langle u, u\rangle < \infty$ is complete (i.e. a Hilbert space).

We have vector-valued function $q=(q_1,\dotsb,q_n)\in U$, space $U_f$ spanned by the set of vectors $\{\phi_j(\cdot)\}_{j\in [m]}$ where each $\phi_j = (\phi_{j,1},\dotsb,\phi_{j,n})\in L^2(I)$, and wish to minimize
\[
F(a):=\int_I\left\|\begin{bmatrix}
    q_1(t)\\
    \vdots\\
    q_n(t)
\end{bmatrix}-a_1\begin{bmatrix}
    \phi_{1,1}\\
    \vdots\\
    \phi_{1,n}
\end{bmatrix}+\dotsb a_m\begin{bmatrix}
    \phi_{m,1}\\
    \vdots\\
    \phi_{m,n}
\end{bmatrix}\right\|^2 dt
\]
over $a=(a_1,\dotsb,a_m)\in \mathbb{R}^m$.
\begin{align*}
F(a) &= \int_I \sum_{r=1}^n (q_r(t)-(a_{1}\phi_{1,r}(t)+\dotsb +a_{m}\phi_{m,r}(t)))^2 dt\\
\pdv{F}{a_{j}} &= \int_I \sum_{r=1}^n -2\phi_{j,r}(t)\left(q_r(t)-(a_{1}\phi_{1,r}(t)+\dotsb +a_{m}\phi_{m,r}(t))\right) dt\\
\pdv[2]{F}{a_{j}} &= \int_I \sum_{r=1}^n 2\phi_{j,r}(t)^2 dt,\;\;\;\;\; \pdv[2]{F}{a_{j}}{a_{k}} = \int_I \sum_{r=1}^n 2\phi_{j,r}(t)\phi_{k,r}(t) dt \text{\;\;for\;} j\neq k.
\end{align*}
The Hessian matrix
\begin{align*}
    H_F &= \begin{bmatrix}
    \pdv[2]{F}{a_{1}} & \pdv[2]{F}{a_{1}}{a_{2}} &\dotsb &\pdv[2]{F}{a_{1}}{a_{m}}\\
    \pdv[2]{F}{a_{2}}{a_{1}} & \ddots & & \\
    \vdots & & & \vdots\\
    \pdv[2]{F}{a_{m}}{a_{1}} & \dotsb & & \pdv[2]{F}{a_{m}} 
\end{bmatrix} \\
&= \begin{bmatrix}
    \int_I \sum_{r=1}^n 2\phi_{1,r}(t)^2 dt & \int_I \sum_{r=1}^n 2\phi_{1,r}(t)\phi_{2,r}(t) dt & \dotsb &   \int_I \sum_{r=1}^n 2\phi_{1,r}(t)\phi_{m,r}(t) dt\\
    \int_I \sum_{r=1}^n 2\phi_{2,r}(t)\phi_{1,r}(t) dt & \ddots & & \vdots\\
    \vdots & & &\\
     \int_I \sum_{r=1}^n 2\phi_{m,r}(t)\phi_{1,r}(t) dt & \dotsb & & \int_I \sum_{r=1}^n 2\phi_{m,r}(t)^2 dt 
\end{bmatrix}\\
 &= 2\sum_{r=1}^n \underbrace{\begin{bmatrix}
    \int_I \phi_{1,r}(t)^2 dt & \int_I \phi_{1,r}(t)\phi_{2,r}(t) dt & \dotsb &  \int_I \phi_{1,r}(t)\phi_{m,r}(t) dt\\
    \int_I \phi_{2,r}(t)\phi_{1,r}(t) dt & \ddots & & \vdots\\
    \vdots & & & \\
     \int_I \phi_{m,r}(t)\phi_{1,r}(t) dt & \dotsb & & \int_I \phi_{m,r}(t)^2 dt
\end{bmatrix} }_{{\Sigma^{-1}}^{(r)}}.
\end{align*}
Thus, $H_F=2\Sigma^{-1}$ and is positive definite. For all $j\in [m]$
\begin{align*}
    \pdv{F}{a_{j}} = 0 &\iff \sum_{r=1}^n \int_I \phi_{j,r}(t)\left(a_1\phi_{1,r}(t)+\dotsb+a_m\phi_{m,r}(t) dt\right) = \sum_{r=1}^n\int_I \phi_{j,r}(t)q_r(t)dt \\
    &\iff a_1\sum_{r=1}^n  \int_I \phi_{j,r}(t)\phi_{1,r}(t) dt +\dotsb+ a_m\sum_{r=1}^n  \int_I \phi_{j,r}(t)\phi_{n,r}(t) dt = \sum_{r=1}^n\int_I \phi_{j,r}(t)q_r(t)dt\\
    &\iff \begin{bmatrix}
        \sum_{r=1}^n  \int_I \phi_{j,r}(t)\phi_{1,r}(t) dt & \dotsb & \sum_{r=1}^n  \int_I \phi_{j,r}(t)\phi_{m,r}(t) dt
    \end{bmatrix} \begin{bmatrix}
        a_1\\
        \vdots\\
        a_m
    \end{bmatrix} =\sum_{r=1}^n\int_I \phi_{j,r}(t)q_r(t)dt.
\end{align*}
The minimizer is therefore given by
\begin{align}
\nonumber
    \sum_{r=1}^n \begin{bmatrix}
        \int_I \phi_{1,r}(t)^2 dt &\dotsb & \int_I \phi_{1,r}(t)\phi_{m,r}(t) dt\\
        \vdots & &\\
        \int_I \phi_{m,r}(t)\phi_{1,r}(t) dt &\dotsb & \int_I \phi_{m,r}(t)^2 dt
    \end{bmatrix} 
    \begin{bmatrix}
        a_1\\
        \vdots\\
        a_m
    \end{bmatrix} &=  \begin{bmatrix}
        \sum_{r=1}^n \int_I \phi_{1,r}(t)q_r(t)dt\\
        \vdots\\
        \sum_{r=1}^n \int_I \phi_{m,r}(t)q_r(t)dt
    \end{bmatrix}\\
    \iff  \begin{bmatrix}
        a_1\\
        \vdots\\
        a_m
    \end{bmatrix} &= \Sigma \begin{bmatrix}
        \sum_{r=1}^n \int_I \phi_{1,r}(t)q_r(t)dt\\
        \vdots\\
        \sum_{r=1}^n \int_I \phi_{m,r}(t)q_r(t)dt
    \end{bmatrix}.
    \label{eqn:l2proj_vecbasis}
\end{align}

Let $a^*=\begin{bmatrix}
        a^*_1\\
        \vdots\\
        a^*_m
    \end{bmatrix} := \Sigma \begin{bmatrix}
        \sum_{r=1}^n \int_I \phi_{1,r}(t)q_r(t)dt\\
        \vdots\\
        \sum_{r=1}^n \int_I \phi_{m,r}(t)q_r(t)dt
    \end{bmatrix}$. Then $a^*$ is the unique minimizer of $F(a)$ and 
    \[
    \int_I(q(t)-\sum_{j\in [m]} a^*_j \phi_{j}(t))^T \begin{bmatrix}
        \phi_{l,1}(t) \\
        \vdots \\
        \phi_{l,n}(t)
    \end{bmatrix} dt = 0
    \]
    for $l\in [m]$.
    
\subsection{Real-valued basis functions and vector-valued coefficients}
\label{sec:app_realbasis_veccoeff}
We have vector-valued function $q=(q_1,\dotsb,q_n)\in U$, $m$-dimensional space $U_f\subset L^2(I)$ spanned by $\{\phi_j(\cdot)\}_{j\in [m]}$, and wish to minimize \[
F(a):=\int_I\left\|\begin{bmatrix}
    q_1(t)\\
    \vdots\\
    q_n(t)
\end{bmatrix}-\left(\begin{bmatrix}
    a_{1,1}\\
    \vdots\\
    a_{1,n}
\end{bmatrix}\phi_1(t)+\dotsb \begin{bmatrix}
    a_{m,1}\\
    \vdots\\
    a_{m,n}
\end{bmatrix}\phi_m(t)\right)\right\|^2 dt
\]
over ${a=(a_{1,1},\dotsb,a_{m,1},\dotsb,a_{m,n})\in \mathbb{R}^{m\times n}}$. As in Section~\ref{sec:app_realbasis_realcoeff} above, we show that the minimization program with objective function $F(a)$ is convex.
\begin{align*}
    F(a)&= \int_I\sum_{r=1}^n (q_r(t)-(a_{1,r}\phi_1(t)+\dotsb a_{m,r}\phi_m(t)))^2 dt,\\
    \pdv{F}{a_{j,r}} &= \int_I -2\phi_j(t)(q_r(t)-(a_{1,r}\phi_1(t)+\dotsb a_{m,r}\phi_m(t)))dt\\
    \pdv[2]{F}{a_{j,r}} &= 2\int_I \phi_j(t)^2 dt,\;\;\;\;\;\;\pdv[2]{F}{a_{j,r}}{a_{k,r}} = 2\int_I \phi_j(t)\phi_k(t) dt,\;\;\; \;\;\;\pdv[2]{F}{a_{j,r}}{a_{k,s}} = 0 \text{\;for\;} r\neq s.
\end{align*}
The Hessian matrix 
\begin{align*}
H_F&=\begin{bmatrix}
    \pdv[2]{F}{a_{1,1}} &{\dotsb} & \pdv[2]{F}{a_{1,1}}{a_{m,1}} &\pdv[2]{F}{a_{1,1}}{a_{1,2}} &{\dotsb}& \pdv[2]{F}{a_{1,1}}{a_{m,n}}\\
    \pdv[2]{F}{a_{2,1}}{a_{1,1}} & \pdv[2]{F}{a_{2,1}} & {\dotsb} & & & \pdv[2]{F}{a_{2,1}}{a_{m,n}}\\
    \vdots &  & \ddots &  & & \vdots\\
    \pdv[2]{F}{a_{m,n}}{a_{1,1}} & {\dotsb} & {\dotsb} & & & \pdv[2]{F}{a_{m,n}}
\end{bmatrix} 
=
\begin{bmatrix}
  H_F^{(0)} & 0 & \dotsb &\\
  \vdots & \ddots & \vdots\\
   0 &\dotsb & H_F^{(0)} \\
\end{bmatrix}
\end{align*}
where 
\[
 H_F^{(0)} =\begin{bmatrix}
     2\int_I \phi_1(t)^2 dt & {\dotsb} &  2\int_I \phi_1(t)\phi_m(t) dt \\
     \vdots & \ddots & \vdots\\
     2\int_I \phi_m(t)\phi_1(t) dt & \dotsb & 2\int_I \phi_m(t)^2 dt 
 \end{bmatrix}.
\]
Thus, $H_F=2\Sigma^{-1}$ and is positive definite. We find the minimizer of $F(a)$ by setting the gradient vector $[\pdv{F}{a_{1,1}},\dots,\pdv{F}{a_{m,1}},\dots,\pdv{F}{a_{m,n}}]^T=0$. We have for all $j\in [m], r\in [n]$
\begin{align*}
\pdv{F}{a_{j,r}} = 0 &\iff \int_I \phi_j(t)(a_{1,r}\phi_1(t)+\dotsb a_{m,r}\phi_m(t))) dt = \int_I \phi_j(t) q_r(t) dt\\
& \iff \begin{bmatrix}
    \int_I \phi_j(t)\phi_1(t) dt & \dotsb & \int_I \phi_j(t)\phi_m(t) dt
\end{bmatrix}
\begin{bmatrix}
    a_{1,r}\\
    \vdots \\
    a_{m,r}
\end{bmatrix} = \int_I \phi_j(t) q_r(t) dt,
\end{align*}
so the minimizer of $F(a)$ is given by
\[
\underbrace{\begin{bmatrix}
     \int_I \phi_1(t)^2 dt & {\dotsb} &  \int_I \phi_1(t)\phi_m(t) dt \\
     \vdots & \ddots & \vdots\\
     \int_I \phi_m(t)\phi_1(t) dt & \dotsb & \int_I \phi_m(t)^2 dt 
 \end{bmatrix}}_{\Sigma_0^{-1}}
 \begin{bmatrix}
    a_{1,r}\\
    \vdots \\
    a_{m,r}
\end{bmatrix} =\begin{bmatrix}
    \int_I \phi_1(t) q_r(t) dt\\
    \vdots\\
    \int_I \phi_m(t) q_r(t) dt
\end{bmatrix}
\]
for $r\in [n]$. In terms of the matrix $\Sigma$,
\begin{equation}
\label{eqn:l2proj_veccoeff}
\begin{bmatrix}
    a_{1,1}\\
    \vdots\\
    a_{m,1}\\
    a_{1,2}\\
    \vdots\\
    a_{m,n}
\end{bmatrix}= \underbrace{\begin{bmatrix}
    \Sigma_0 &  &\\
     &\ddots &\\
    & & \Sigma_0
\end{bmatrix}}_{\Sigma}\begin{bmatrix}
    \int_I \phi_1(t) q_1(t) dt\\
    \vdots\\
    \int_I \phi_m(t) q_1(t) dt\\
    \int_I \phi_1(t) q_2(t) dt\\
   \vdots\\
    \int_I \phi_m(t) q_n(t) dt
\end{bmatrix}.
\end{equation}

\subsection{Approximations of functions in $C([0,T])$ using basis in $L^2(\mathbb{R})$}
\label{sec:app_bounded_via_unbounded}
In Section~\ref{sec:infbasis_finitefunc}, we discussed how a continuous function $q\in U$ can be approximated using basis functions in $L^2(\mathbb{R})$ by treating $q$ as a function in $L^2(\mathbb{R})$. In particular, we let $L(q)$ denote the extended function that satisfies $L(q)(t)=q(t)$ for $t\in [0,T]$, and $L(q)(t)=0$ otherwise. Let the basis functions be $\{\phi_{j}\}_{j\in [m]} \subset L^2(\mathbb{R})$. By Appendix~\ref{sec:app_realbasis_realcoeff} above, the least squares projection of $L(q)$ onto the span of $\{\phi_{j}(\cdot)\}_{j\in [m]}$ is given by
\[
\begin{bmatrix}
    \int_{-\infty}^{\infty} \phi_{1}(t)^2 dt & \dotsb & \int_{-\infty}^{\infty} \phi_{1}(t)\phi_m(t) dt\\
    \vdots & &\\
    \int_{-\infty}^{\infty} \phi_{m}(t)\phi_1(t) dt &\dotsb & \int_{-\infty}^{\infty} \phi_{m}(t)^2 dt
\end{bmatrix}\begin{bmatrix}
    a_1\\
    \vdots\\
    a_m
\end{bmatrix} = \begin{bmatrix}
    \int_{-\infty}^{\infty} L(q)(t)\phi_{1}(t) dt \\
    \vdots\\
    \int_{-\infty}^{\infty} L(q)(t)\phi_{m}(t) dt 
\end{bmatrix} =\begin{bmatrix}
    \int_{0}^{T} q(t)\phi_{1}(t) dt \\
    \vdots\\
    \int_{0}^{T} q(t)\phi_{m}(t) dt 
\end{bmatrix}.
\]
On the other hand, the least squares projection of $q$ onto the span of $\{\phi_{j}(\cdot)\mathbb{1}_{[0,T]}(\cdot)\}_{j\in [m]}$ is given by
\[
\begin{bmatrix}
    \int_{0}^{T} \phi_{1}(t)^2 dt & \dotsb & \int_{0}^{T} \phi_{1}(t)\phi_m(t) dt\\
    \vdots & &\\
    \int_{0}^{T} \phi_{m}(t)\phi_1(t) dt &\dotsb & \int_{0}^{T} \phi_{m}(t)^2 dt
\end{bmatrix}\begin{bmatrix}
    a_1\\
    \vdots\\
    a_m
\end{bmatrix}  =\begin{bmatrix}
    \int_{0}^{T} q(t)\phi_{1}(t) dt \\
    \vdots\\
    \int_{0}^{T} q(t)\phi_{m}(t) dt 
\end{bmatrix}.
\]
Thus, while using basis functions in  $L^2(\mathbb{R})$ might be more convenient in some cases, the quality of the approximation could be different.

\section{Additional Statistics on Experimental Evaluation}
In this section, we provide additional statistics on our experimental evaluation  presented in Section~\ref{sec:exp} and Appendix~\ref{sec:app_cgpexp}, for the results on the ECG and taxi trajectory datasets. Observe that for the taxi trajectory dataset, the median errors for the baseline methods are noticeably smaller than the mean errors (Fig.~\ref{fig:Taxi-GP_quant} vs. \ref{fig:Taxi-GP_avg}, and Fig.~\ref{fig:Taxi-CGP_quant} vs. \ref{fig:Taxi-CGP_avg}), indicating that there were extreme curves in the taxi trajectory dataset that severely damaged the performance of the baseline methods; on the other hand, the median and mean errors are similar in our methods, so our methods are less impacted by such curves. Overall, the advantage of our methods is clear.
\begin{figure}[h]
      \centering
    \begin{subfigure}[t]{0.3\linewidth}
        \includegraphics[width = \textwidth]{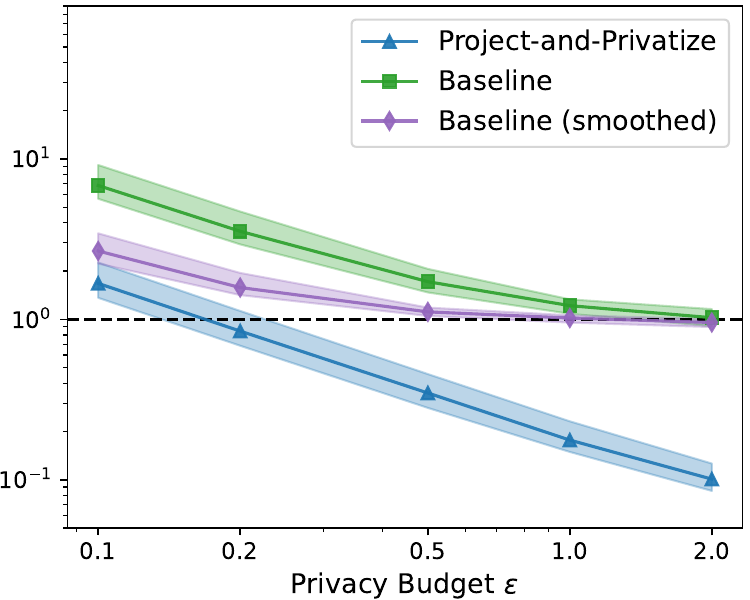}
        \subcaption{GP $L^2$ error}
        \label{fig:ECG-GP_quant}
    \end{subfigure}
    \;\;\;\;\;
    \begin{subfigure}[t]{0.3\linewidth}
        \includegraphics[width = \textwidth]{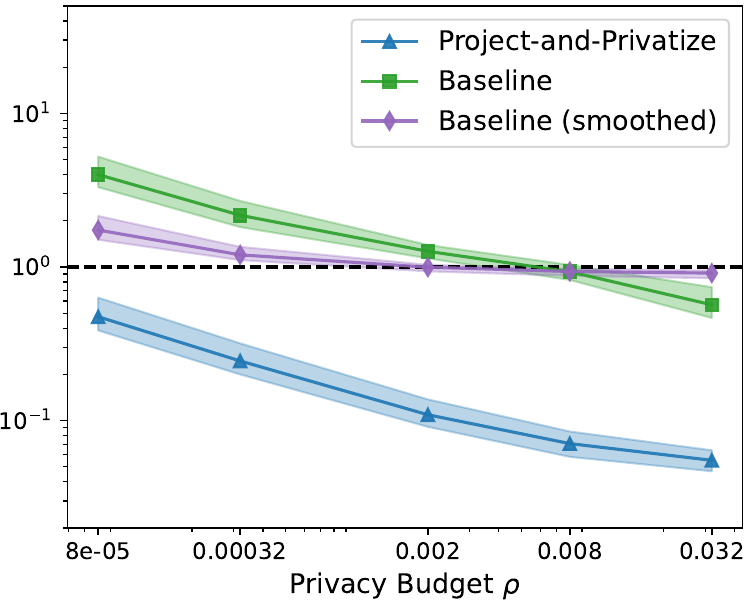} 
        \subcaption{CGP $L^2$ error}
        \label{fig:ECG-CGP_quant}
    \end{subfigure}
    \vspace{-4pt}
    \caption{$L^2$ error on the ECG dataset: median, $25$th and $75$th percentiles.}
    \label{fig:ECG_quant}
\end{figure}
\begin{figure}[h]
      \centering
    \begin{subfigure}[t]{0.3\linewidth}
        \includegraphics[width = \textwidth]{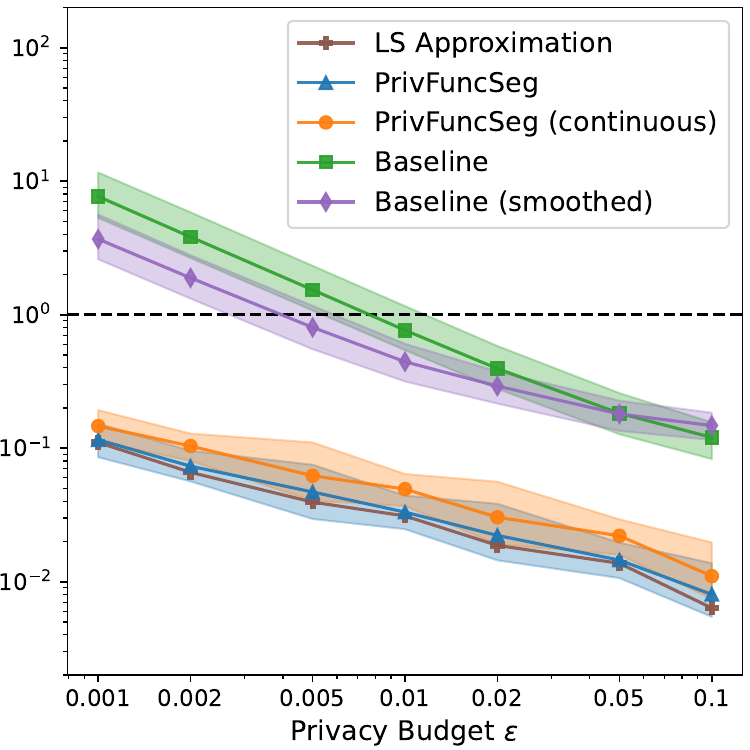}
        \subcaption{GP $L^2$ error}
        \label{fig:Taxi-GP_quant}
    \end{subfigure}
    \;\;\;\;\;
    \begin{subfigure}[t]{0.3\linewidth}
        \includegraphics[width = \textwidth]{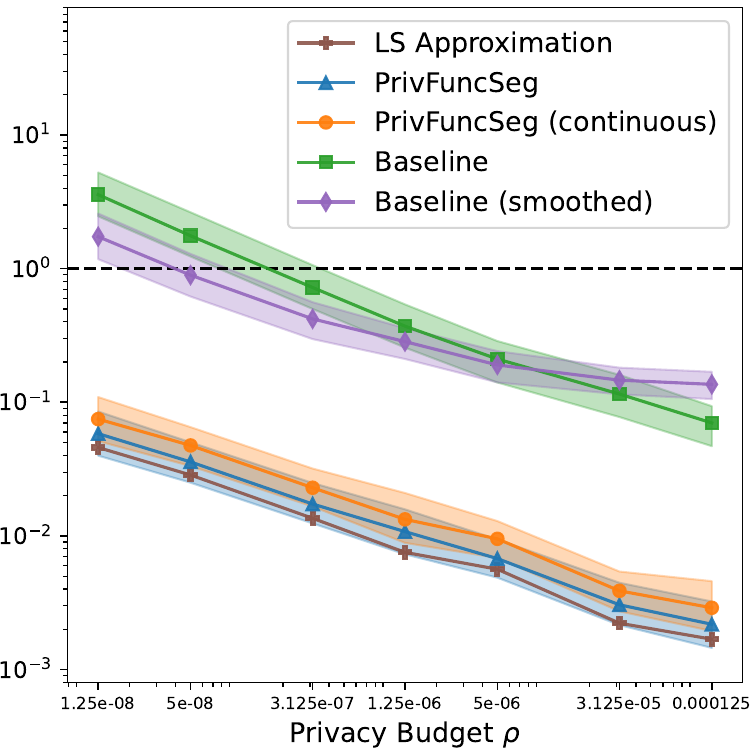} 
        \subcaption{CGP $L^2$ error}
        \label{fig:Taxi-CGP_quant}
    \end{subfigure}
    \vspace{-4pt}
    \caption{$L^2$ error on the taxi trajectory dataset: median, $25$th and $75$th percentiles.}
    \label{fig:Taxi_quant}
\end{figure}
\end{document}